\documentclass[oneside,english]{amsart}
\usepackage[T1]{fontenc}
\usepackage[latin9]{inputenc}
\usepackage{geometry}
\usepackage{verbatim}
\usepackage{amsthm}
\usepackage{amstext}
\usepackage{amssymb}
\usepackage{graphicx}
\usepackage{esint}
\usepackage{url}
\usepackage[all]{xy}

\usepackage{color}

\makeatletter
\numberwithin{equation}{section}
\numberwithin{figure}{section}
  \theoremstyle{plain}
  \newtheorem*{thm*}{\protect\theoremname}
  \theoremstyle{plain}
  \newtheorem*{lem*}{\protect\lemmaname}
  \theoremstyle{plain}
  \newtheorem*{prop*}{\protect\propositionname}
  \theoremstyle{plain}
  \newtheorem*{cor*}{\protect\corollaryname}
\theoremstyle{plain}
\newtheorem{thm}{\protect\theoremname}[section]
  \theoremstyle{plain}
  \newtheorem{lem}[thm]{\protect\lemmaname}
  \theoremstyle{remark}
  \newtheorem{rem}[thm]{\protect\remarkname}
  \theoremstyle{plain}
  \newtheorem{prop}[thm]{\protect\propositionname}
  \theoremstyle{plain}
  \newtheorem{cor}[thm]{\protect\corollaryname}

\usepackage{mathrsfs}
\usepackage{subfigure}

\AtBeginDocument{

}

\makeatother

\makeatletter
\def\Ddots{\mathinner{\mkern1mu\raise\p@
\vbox{\kern7\p@\hbox{.}}\mkern2mu
\raise4\p@\hbox{.}\mkern2mu\raise7\p@\hbox{.}\mkern1mu}}
\makeatother

\usepackage{babel}
  \providecommand{\corollaryname}{Corollary}
  \providecommand{\lemmaname}{Lemma}
  \providecommand{\propositionname}{Proposition}
  \providecommand{\remarkname}{Remark}
  \providecommand{\theoremname}{Theorem}
\providecommand{\theoremname}{Theorem}

\begin{document}
\global\long\def\SLE{\mathrm{SLE}}
 \global\long\def\SLEk{\mathrm{SLE}_{\kappa}}
 \global\long\def\SLEkappa#1{\mathrm{SLE}_{#1}}
 \global\long\def\SLEkapparho#1#2{\mathrm{SLE}_{#1}(#2)}

\global\long\def\PR{\mathsf{P}}
 \global\long\def\EX{\mathsf{E}}
 \global\long\def\sF{\mathcal{F}}
 \global\long\def\sFalt{\overleftarrow{\sF}}

\global\long\def\bR{\mathbb{R}}
 \global\long\def\bZ{\mathbb{Z}}
 \global\long\def\bN{\mathbb{N}}
 \global\long\def\bZpos{\mathbb{Z}_{> 0}}
 \global\long\def\bZnn{\mathbb{Z}_{\geq 0}}
 \global\long\def\bQ{\mathbb{Q}}

\global\long\def\bC{\mathbb{C}}
 \global\long\def\Rsphere{\overline{\bC}}
 \global\long\def\re{\Re\mathfrak{e}}
 \global\long\def\im{\Im\mathfrak{m}}
 \global\long\def\arg{\mathrm{arg}}
 \global\long\def\ii{\mathfrak{i}}

\global\long\def\bD{\mathbb{D}}
 \global\long\def\bH{\mathbb{H}}

\global\long\def\sZ{\mathcal{Z}}
 \global\long\def\sD{\mathcal{D}}
 \global\long\def\sC{\mathcal{C}}
 \global\long\def\sL{\mathcal{L}}
 \global\long\def\sA{\mathcal{A}}

\global\long\def\bx{\boldsymbol{x}}
\global\long\def\bw{\boldsymbol{w}}

\global\long\def\dist{\mathrm{dist}}

\global\long\def\eps{\varepsilon}
 \global\long\def\const{\mathrm{const.}}
 \global\long\def\half{\frac{1}{2}}
 \global\long\def\Ord{\mathcal{O}}
 \global\long\def\ord{\mathit{o}}
 \global\long\def\OO{\mathcal{O}}
 \global\long\def\oo{\mathit{o}}

\global\long\def\domain{\Lambda}
\global\long\def\bdry{\partial}
 \global\long\def\cl#1{\overline{#1}}

\global\long\def\Ampl{\zeta}
 \global\long\def\Corr{\chi}
 \global\long\def\lft{-}
 \global\long\def\rgt{+}
 \global\long\def\rgtlft{\pm}

 \global\long\def\Catalan{\mathrm{C}}

\global\long\def\ud{\mathrm{d}}
 \global\long\def\der#1{\frac{\ud}{\ud#1}}
 \global\long\def\pder#1{\frac{\partial}{\partial#1}}
 \global\long\def\pdder#1{\frac{\partial^{2}}{\partial#1^{2}}}

\global\long\def\set#1{\left\{  #1\right\}  }
 \global\long\def\setcond#1#2{\left\{  #1\;\big|\;#2\right\}  }

\global\long\def\CubeInt{\widetilde{\rho}}
 \global\long\def\SimplexInt{\rho}
 \global\long\def\FWint{\varphi}
 \global\long\def\FWintalt{{\overleftarrow{\FWint}}}
 \global\long\def\AsyInt{\alpha}
 \global\long\def\DecoInt{\widetilde{\omega}}

\global\long\def\SurfSimplex{\mathfrak{R}}
 \global\long\def\SurfCube{\SurfSimplex^{\approx}}
 \global\long\def\SurfFW{\mathfrak{L}^{\Supset}}
 \global\long\def\SurfFWalt{\mathfrak{L}^{\Subset}}
 \global\long\def\SurfAsy{\mathfrak{M}^{\sim\!\!\!\supset}}

\global\long\def\SurfSimplexAnch{\mathfrak{R}^{(x_{0})}}
 \global\long\def\SurfCubeAnch{\SurfSimplex^{(x_{0});\approx}}
 \global\long\def\SurfFWAnch{\mathfrak{G}^{(x_{0});\Supset}}

\global\long\def\fSimplex{f}
 \global\long\def\fCube{\fSimplex^{\approx}}
 \global\long\def\fFW{\fSimplex^{\Supset}}
 \global\long\def\fFWalt{\fSimplex^{\Subset}}
 \global\long\def\fAsy{\fSimplex^{\sim\!\!\!\supset}}

\global\long\def\braid{\sigma}
 \global\long\def\twoFone{\,{}_{2}F_{1}}

\global\long\def\sl{\mathfrak{sl}}
 \global\long\def\Uqsltwo{\mathcal{U}_{q}(\mathfrak{sl}_{2})}
 \global\long\def\Hcp{\Delta}

\global\long\def\chamber{\mathfrak{X}}
\global\long\def\chamberalt{\overleftarrow{\chamber}}
 \global\long\def\Wchamber{\mathfrak{W}}
 \global\long\def\sR{\mathcal{R}}
 \global\long\def\FKcone{L\mathrm{-cone}}

\global\long\def\sP{\mathcal{P}}

\global\long\def\Arch{\mathrm{Arch}}
 \global\long\def\Conn{\mathrm{Conn}}
 \global\long\def\arch#1#2{\left[#1\phantom{}^{\frown}#2\right]}
 \global\long\def\nested{\boldsymbol{\underline{\Cap}}}
 \global\long\def\unnested{\boldsymbol{\underline{\cap\cap}}}


\global\long\def\sciOp{\text{\Rightscissors}}
 \global\long\def\tieOp{\wp}

\global\long\def\qnum#1{\left[#1\right] }
 \global\long\def\qfact#1{\left[#1\right]! }
 \global\long\def\qbin#1#2{\left[\begin{array}{c}
 #1\\
#2 
\end{array}\right]}

\global\long\def\Hom{\mathrm{Hom}}
 \global\long\def\End{\mathrm{End}}
 \global\long\def\Aut{\mathrm{Aut}}
 \global\long\def\Rad{\mathrm{Rad}}
 \global\long\def\Ext{\mathrm{Ext}}

\global\long\def\dmn{\mathrm{dim}}
 \global\long\def\spn{\mathrm{span}}
 \global\long\def\tens{\otimes}
 \global\long\def\Pf{\mathrm{Pf}}
 \global\long\def\sgn{\mathrm{sign}}

\global\long\def\Mat{\mathrm{Mat}}
 \global\long\def\unitmat{\mathbb{I}}
 \global\long\def\id{\mathrm{id}}
 \global\long\def\isom{\cong}

\global\long\def\Wd{\mathsf{M}}
 \global\long\def\Wbas{e}
 \global\long\def\Tbas{\tau}
\global\long\def\HWsp{\mathrm{H}}

\global\long\def\Mob{\mu}

\global\long\def\diadim{\delta}

\global\long\def\SymmGrp{\mathfrak{S}}
\global\long\def\symgrp{\SymmGrp}

\global\long\def\pialt{\overleftarrow{\pi}}

\title[Conformally covariant boundary correlation functions]{Conformally covariant boundary correlation functions\\
with a quantum group}
\author{K.~Kyt\"ol\"a and E.~Peltola}

\

\vspace{2.5cm}

\begin{center}
\LARGE \bf \scshape Conformally covariant boundary correlation functions \\ with a quantum group
\end{center}

\vspace{0.75cm}

\begin{center}
{\large \scshape Kalle Kyt\"ol\"a}\\
{\footnotesize{\tt kalle.kytola@aalto.fi}}\\
{\small{Department of Mathematics and Systems Analysis}}\\
{\small{P.O. Box 11100, FI-00076 Aalto University, Finland}}\bigskip{}
\\
{\large \scshape Eveliina Peltola}\\
{\footnotesize{\tt eveliina.peltola@unige.ch}}\\
{\small{Section de Math\'{e}matiques, Universit\'{e} de Gen\`{e}ve,}}\\
{\small{2--4 rue du Li\`{e}vre, C.P. 64, 1211 Gen\`{e}ve 4, Switzerland}}
\end{center}

\vspace{0.75cm}

\begin{center}
\begin{minipage}{0.85\textwidth} \footnotesize
{\scshape Abstract.}
Particular boundary correlation functions of conformal field theory
are needed to answer some questions related to random conformally invariant
curves known as Schramm-Loewner evolutions ($\SLE$). In this article,
we introduce a correspondence and establish its fundamental properties,
which are used in the companion articles \cite{JJK-SLE_boundary_visits,KP-pure_partition_functions_of_multiple_SLEs}
for explicitly solving two such problems. The correspondence associates
Coulomb gas type integrals to vectors in a tensor product representation
of a quantum group, a $q$-deformation of the Lie algebra $\sl_2$.
We show that desired properties of the functions are guaranteed by
natural representation theoretical properties of the vectors.
\end{minipage}
\end{center}

\vspace{0.75cm}

\bigskip{}


\section{Introduction}

Boundary correlation functions in conformal field theories in general,
and, in particular, questions about random conformally invariant
curves, frequently lead to quite similar systems of partial differential
equations whose boundary conditions are specified in terms of the
asymptotic behavior of solutions. The main result of this article
provides a systematic method to construct explicit solutions. The
method is a form of the so-called ``hidden quantum group symmetry
of conformal field theories''
\cite{FFK-braid_matrices,
MR-comment_on_quantum_group_symmetry_in_CFT,
BMP-quantum_group_structure_in_the_Fock_space_resolutions_of_slnhat_representations,
GS-quantum_group_meaning_of_the_Coulomb_gas,
FW-topological_representation_of_Uqsl2,
PS-common_structures_between_finite_systems_and_CFTs,
RRR-contour_picture_of_quantum_groups_in_CFT,
Varchenko-ICM1990-multidimensional_hypergeometric_functions}
(see also the textbooks \cite{Fuchs-affine_Lie_algebras_and_quantum_groups,Varchenko-multidimensional_hypergeometric_functions_and_representation_theory_of_Lie_algebras_and_quantum_groups,
GRS-quantum_groups_in_2d_physics}).
In this article, we establish properties which are directly relevant
for solving the PDE problems arising in the theory of Schramm-Loewner
evolutions ($\SLE$). In two companion articles, the results are applied
to produce explicit answers to the questions of boundary visit probabilities
of chordal $\SLE$s \cite{JJK-SLE_boundary_visits} and the pure geometries
of multiple $\SLE$s \cite{KP-pure_partition_functions_of_multiple_SLEs}.
With appropriate modifications, the method can also be applied to
bulk correlation functions --- in \cite{FP-monodromy_invariant_correlations}
it is applied to the construction of single-valued bulk correlation
functions of conformal field theory.

At the heart of the method are integral solutions to differential
equations, crucial to the entire Coulomb gas formalism of
conformal field theory \cite{DF-four_point_correlation_functions}.
Such an idea was used already by Euler for solving the
hypergeometric differential equation, and it can be summarized as follows.
Let $\sD$ be a differential operator acting in a variable $x$.
Suppose that $f(x,w)$ is a function of $x$ and an auxiliary variable
$w$ such that the differential operator $\sD$ acting on $f$
gives an exact form in the $w$-variable:
$( \sD f ) \, \ud w = \ud \eta$.
Let $\Gamma$
be an integration surface in the $w$-variable, and
define a function $F$ of $x$ by the integral
\[ F(x) = \int_\Gamma f(x,w) \, \ud w . \]
If the order of integration and differentiation can be exchanged in
\[  \sD F = \sD \int_\Gamma f \, \ud w
    = \int_\Gamma (\sD f) \, \ud w = \int_\Gamma \ud \eta , \]
then by Stokes' theorem we have
\[  \sD F = \int_{\Gamma} \ud \eta
    = \int_{\bdry \Gamma} \eta . \]
If the surface $\Gamma$ is closed, $\bdry \Gamma = \emptyset$,
then the right hand side above vanishes, and $F$ therefore
satisfies the differential equation $\sD F = 0$.
When an appropriate
auxiliary function $f$ is known, the remaining difficulty in
solving the differential equation $\sD F = 0$ lies in choosing a
surface $\Gamma$ which is not only closed, but also such that
the constructed function $F = \int_\Gamma f \, \ud w$ satisfies
whatever boundary conditions are imposed.

The method of this
article exploits a hidden quantum group structure in the choice
of the appropriate integration surface~$\Gamma$
for solving certain partial differential equations of conformal field theory.
We construct a linear correspondence
from a representation of a quantum group to
functions defined by integrals. In this correspondence, representation
theoretic operations conveniently allow both to verify
the closedness of the integration surface, and consequently the
differential equations, and to analyze the boundary behavior
of the function. The properties are formulated in a systematic
and readily applicable fashion, as demonstrated in the examples
provided later in this introduction.

\subsection{\label{sub: PDEs from CFT}Partial differential equations from conformal field theory}
Let us parametrize the central charge $c$ of the conformal field theory
by a number $\kappa>0$ via $c=c(\kappa) = 13 - \frac{3\kappa}{2} - \frac{24}{\kappa}$.
According to the seminal paper \cite{BPZ-infinite_conformal_symmetry_in_2d_QFT},
the correlation functions of
general fields can (usually) be reduced to correlation functions of the
primary fields in the corresponding ``conformal families''. Moreover,
when the conformal weights of the primary fields lie in the so-called Kac table,
their correlation functions can be expected to satisfy partial differential equations,
by virtue of degeneracies in the representation theory of the Virasoro algebra.
For primary fields in the first row of the Kac table,
the conformal weights are of the form
\begin{align}\label{eq: Kac labeled conformal weights}
h_{1,d} =\; & \frac{(d-1)(2(d+1)-\kappa)}{2\kappa}, \qquad \text{ with } d \in \bZ_{>0},
\end{align}
and explicit expressions for
these partial differential equations have been found by Benoit and
Saint-Aubin~\cite{BSA-degenerate_CFTs_and_explicit_expressions}. These are the
PDEs considered in the present article. Specifically, for a conformal field theory in the upper half-plane
$\bH = \set{z \in \bC \; \big| \; \im(z) > 0}$, a boundary
correlation function of $n$ such primary fields is a function defined on the chamber
\begin{align}\label{eq: chamber}
\chamber_{n}=\; & 
\set{(x_{1},\ldots,x_{n})\in\bR^{n}\;\Big|\; x_{1}<\cdots<x_{n}} ,
\end{align}
and if the fields at $x_j$, for $j=1,\ldots,n$, have conformal weights $h_{1,d_j}$, respectively, 
then the Benoit~\& Saint-Aubin
partial differential equations for the correlation function
$F \colon \chamber_n \to \bC$ read
\begin{align}\label{eq: BSA PDE for a function}
\sum_{k=1}^{d_{j}}\sum_{\substack{p_{1},\ldots,p_{k}\geq1\\
p_{1}+\cdots+p_{k}=d_{j}}}
\frac{(-4/\kappa)^{d_{j}-k}\,(d_{j}-1)!^{2}}{\prod_{u=1}^{k-1}(\sum_{i=1}^{u}p_{i})(\sum_{i=u+1}^{k}p_{i})}\times\sL_{-p_{1}}^{(j)}\cdots\sL_{-p_{k}}^{(j)}
\; F(x_1 , \ldots , x_n) \; = \; 0 ,
\end{align}
where $\sL_{-p}^{(j)} = -\sum_{i\neq j}\left((x_{i}-x_{j})^{1-p}\pder{x_{i}}+(1-p)\, h_{1,d_{i}}\,(x_{i}-x_{j})^{-p}\right) $.
Moreover, covariance of the correlation function under global conformal
transformations requires that under any M\"obius transformation $\Mob \colon \bH \to \bH$
of the upper half-plane such that
$\Mob(x_1) < \cdots < \Mob(x_n)$, we have
\begin{align}\label{eq: Mobius covariance for a function}
F(x_{1},\ldots,x_{n}) =\; & 
\prod_{i=1}^{n} \Mob'(x_{i})^{h_{1,d_i}} \times F(\Mob(x_{1}),\ldots,\Mob(x_{n})) .
\end{align}
In the main result of the present article, we systematically construct 
integral solutions to the PDEs~\eqref{eq: BSA PDE for a function} and
covariance condition~\eqref{eq: Mobius covariance for a function}, and provide 
tools for analyzing their boundary conditions.

\subsection{\label{sub: quantum group}The role of the quantum group}

The method we introduce in this article relies on the representation
theory of the quantum group $\Uqsltwo$ in the generic,
semisimple case (for the
precise definitions, see Section~\ref{sec: q stuff}). Informally,
$\Uqsltwo$ is a deformation of the Lie algebra $\sl_{2}$ of traceless
complex $2\times2$-matrices, with a complex deformation parameter~$q$
that we assume not to be a root of unity or zero.
The deformation parameter is related to the central charge
$c = c(\kappa)$ of the conformal field theory via $q = e^{\ii \pi 4 / \kappa}$.
Our assumption on $q$ corresponds to~$\kappa \notin \bQ$.

As an algebra,
$\Uqsltwo$ is generated by an invertible Cartan element $K$, and
raising and lowering operators $E$ and $F$, which shift the eigenvalues
of $K$ by multiplicative factors $q^{2}$ and $q^{-2}$, respectively.
To state our result, the following representation theoretical notions are needed.
The algebra $\Uqsltwo$ has, for all positive integers $d$, an irreducible
representation $\Wd_{d}$ of dimension $d$, which $q$-deforms the
$d$-dimensional irreducible representation of~$\sl_{2}$. Tensor products of
representations are defined by equipping $\Uqsltwo$ with a Hopf algebra
structure, and they decompose into direct sums of irreducible
subrepresentations.
By trivial representation we mean the one-dimensional representation
$\Wd_1 \isom \bC$ or a direct sum of copies of it.
Finally, we say that a vector $v$ in a representation
is a highest weight vector if it is annihilated by the
raising operator, i.e., $E.v=0$.

To construct $n$-point boundary correlation functions,
we form a tensor product of $n$ irreducible representations of the
quantum group, and to its vectors we associate certain functions
of integral form. 
We show that under this association, desired properties of the functions
follow from natural representation theoretical properties of the vectors.
Functions $\sF[v]$
associated to highest weight vectors $v$ are well-defined on the
chamber domain~\eqref{eq: chamber}, 
and they satisfy a system of Benoit~\& Saint-Aubin 
partial differential equations. 
Asymptotics of the functions can be read off from projections to subrepresentations.
Homogeneity degree of the function is related to the eigenvalue of
the Cartan element $K$, and for vectors in the trivial subrepresentation,
the associated function is covariant under all M\"obius transformations.

\subsection{\label{sub: informal statement}The main result}
We now outline
the main result of this article, whose
precise statement will be given in 
Theorems~\ref{thm: SCCG correspondence non-hwv}
and \ref{thm: SCCG correspondence hwv} in 
Section~\ref{sub: the main correspondence result},
once all relevant notation and conventions have been introduced. Examples
of its applications are discussed in Sections~\ref{sub: example applications} and~\ref{sub: another example application}.

For the precise definition of our correspondence,
at intermediate steps we need an auxiliary anchor point $x_{0}$,
and we have to use functions defined on the restricted chamber
\begin{align}\label{eq: restricted chamber}
\chamber_{n}^{(x_{0})} = \; & 
\set{(x_{1},\ldots,x_{n})\in\bR^{n}\;\Big|\; x_{0}<x_{1}<\cdots<x_{n}}.
\end{align}
The correspondence consists of linear mappings
\begin{align*}
\sF^{(x_{0})} \;\colon\; & \bigotimes_{i=1}^{n}
\Wd_{d_{i}} \to \mathcal{C}^\infty \big( \chamber_{n}^{(x_{0})} \big)
\end{align*}
from the tensor product 
$\bigotimes_{i=1}^{n} \Wd_{d_{i}} = \Wd_{d_{n}}\tens\cdots\tens\Wd_{d_{1}}$
of $n$ irreducible representations of the
quantum group $\Uqsltwo$ to a space of smooth functions of $n$ variables.
By construction, detailed in Section~\ref{sub: definition of the 
correspondence}, the mapping $\sF^{(x_{0})}$ 
sends appropriate basis vectors to functions of the form
\begin{align*}
C \times \prod_{1\leq i<j\leq n}(x_{j}-x_{i})^{\frac{2}{\kappa}(d_{i}-1)(d_{j}-1)} \times
    \underset{\!\!\!\SurfFW}{\int \!\!\! \cdot \! \cdot \! \cdot \!\!\! \int}
    \prod_{\substack{1\leq i\leq n \\ 1\leq r\leq\ell}}(w_{r}-x_{i})^{-\frac{4}{\kappa}(d_{i}-1)}
    \prod_{1\leq r<s\leq\ell}(w_{s}-w_{r})^{\frac{8}{\kappa}}
    \; \ud w_1 \cdots \ud w_{\ell} ,
\end{align*}
where $C$ is a phase factor and the integrations are over a family $\SurfFW$ 
of non-intersecting loops based at the anchor point $x_0$.
The function depends on the vector through the topology of the integration
surface~$\SurfFW$, as well as the phase factor $C$.
The family $\SurfFW$ of non-intersecting loops is schematically illustrated
in Figure~\ref{fig: intro non-intersecting loops}.
\noindent 
\begin{figure}
\includegraphics[width=1\textwidth]{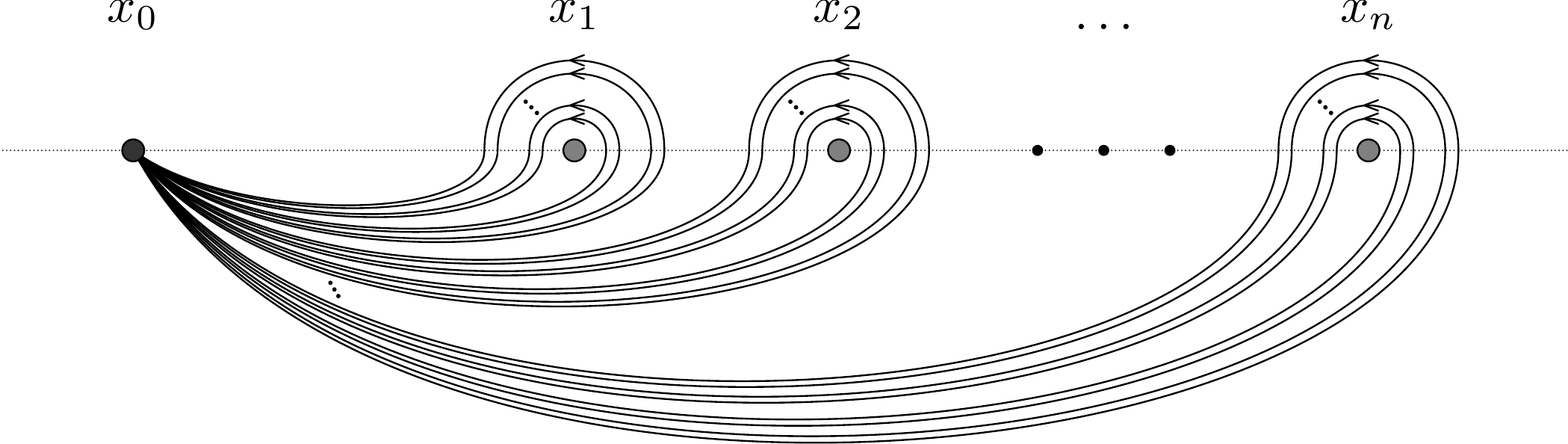}
\caption{\label{fig: intro non-intersecting loops}
Under our correspondence, basis vectors in a tensor product representation
of a quantum group are mapped to integral form functions.
The integration is over
a family $\SurfFW$ of non-intersecting loops based at the anchor point $x_0$,
with nested loops surrounding the points $x_1 , \ldots, x_n$.
The numbers of loops and the phase factor $C$ multiplying the integral depend on the
basis vector, as detailed in Sections~\ref{sub: basis functions as FW integrals}
and~\ref{sub: definition of the correspondence}.
}
\end{figure}

We refer to the mappings $\sF^{(x_{0})}$ as 
``the spin chain~- Coulomb gas correspondence'', 
since they map the state spaces $\bigotimes_{i=1}^{n}\Wd_{d_{i}}$
of finite quantum spin chains to spaces of screened correlation functions
in the Coulomb gas formalism of conformal field theory. Such
integral form correlation functions of CFT have been studied since 
\cite{DF-multipoint_correlation_functions}.

Our main theorem --- informally stated --- is the following.
\begin{thm*}[Theorems~\ref{thm: SCCG correspondence non-hwv} and \ref{thm: SCCG correspondence hwv}
in Section~\ref{sub: the main correspondence result}]
Under the mappings $v\mapsto\sF^{(x_{0})}[v]$,
properties of the vector $v\in\bigotimes_{i=1}^{n}\Wd_{d_{i}}$
ensure properties of the corresponding function
\begin{align*}
\sF^{(x_{0})}[v] \colon \chamber_{n}^{(x_{0})} \to \bC
\end{align*}
as follows:
\begin{description}
\item[{(well-def.)}] For any highest weight vector $v$,
the function $\sF^{(x_{0})}[v]$ is independent of $x_0$,
and thus gives rise to a well-defined function 
on the chamber $\chamber_{n}$, denoted by $\sF[v]$.\smallskip{}
\item [{(PDE)}] If $v$ is a highest weight vector, then the function $F = \sF[v]\colon\chamber_{n}\to\bC$
satisfies, for each $j=1,\ldots,n$, the linear homogeneous
partial differential equation~\eqref{eq: BSA PDE for a function}
of order $d_j$ equal to the dimension of a factor $\Wd_{d_j}$ in the tensor product.
\smallskip{}

\item [{(COV)}] The functions have the following covariance properties:

\begin{itemize}
\item The function $\sF[v]$ (resp. $\sF^{(x_{0})}[v]$) is translation
invariant.
\item If $v$ is an eigenvector of the Cartan generator $K$, then the function
$\sF[v]$ (resp. $\sF^{(x_{0})}[v]$) is homogeneous of a degree that
depends on the eigenvalue.
\item If $v$ belongs to the trivial subrepresentation of $\bigotimes_{i=1}^{n}\Wd_{d_{i}}$,
then $F = \sF[v]$ transforms covariantly under all M\"obius transformations as
in~\eqref{eq: Mobius covariance for a function}.\smallskip{}

\end{itemize}
\item [{(ASY)}] Suppose that $v$ belongs to the subrepresentation of $\bigotimes_{i=1}^{n}\Wd_{d_{i}}$
obtained by picking the $d$-dimensional irreducible direct summand
in the tensor product of the $j$:th and $j+1$:st factors $\Wd_{d_{j}}$
and $\Wd_{d_{j+1}}$, and let 
\begin{align*}
\hat{v} \in \Big(\bigotimes_{i>j+1}\Wd_{d_{i}}\Big)\tens\Wd_{d}
\tens\Big(\bigotimes_{i<j}\Wd_{d_{i}}\Big)
\end{align*}
denote the vector obtained by identifying $v$ as a vector in an $(n-1)$-fold
tensor product representation.
Then, as $|x_{j+1}-x_{j}|\to0$, the function $\sF[v]$ has the asymptotic
behavior
\begin{align*}
\sF[v]\sim\; & B\times(x_{j+1}-x_{j})^{\Delta}\times\sF[\hat{v}],
\end{align*}
where the constant $B=B_{d}^{d_{j},d_{j+1}}$ and the exponent $\Delta=\Delta_{d}^{d_{j},d_{j+1}}$
are explicit. The analogous statement holds also for $\sF^{(x_{0})}[v]$.
\end{description}
\end{thm*}
In practice, this theorem is applied as follows. In specific problems, we are looking
for particular solutions to systems of PDEs of conformal field theory. Typically, the sought
solution has specific M\"obius covariance properties and 
specific boundary conditions as the distance of some of its arguments 
tend to zero. By our correspondence, the task of finding a function with these properties
is translated to the problem of finding a corresponding vector in the tensor product representation.
The different parts (PDE), (COV), and (ASY) of the theorem state that a careful choice
of the vector would ensure the desired properties of the function ---
even the most delicate boundary conditions for the function can be guaranteed by (ASY) if
the vector has appropriate projections to certain subrepresentations.
All of the requirements are explicit
linear conditions on the vector living in a finite dimensional vector space, and we moreover have a variety
of representation theoretical tools at our disposal to solve for such a vector.
By outlining a few case studies in Sections~\ref{sub: example applications}
and~\ref{sub: another example application}, we exemplify how the correspondence thus allows
us to translate the original, possibly rather complicated problem to an explicitly solvable one,
and to eventually express the function of interest explicitly as a linear combination
of integral form functions.

Before example applications, we make a few further observations about the interpretation
of the constructed correspondence, and comparisons to related research.
\begin{itemize}
\item The functions $\sF^{(x_{0})}[v]$ are given by integrals over
auxiliary screening variables,
and slightly informally,
the quantum group can be thought of as acting on the integration surfaces:
the generator $F$ increases the dimension of the integration surface
(i.e., the number of screening variables), and the generator $E$
decreases it.
\item A precise version of the quantum group action on integration surfaces
has been given by Felder and Wieczerkowski, who define an action of
$\Uqsltwo$ on a suitable homology theory \cite{FW-topological_representation_of_Uqsl2}.
Our representation of the quantum group could be obtained from this
homology via a degenerate evaluation. The evaluation in particular
renders the infinite dimensional Verma modules in the work of Felder
and Wieczerkowski into just the finite dimensional irreducible 
representations~$\Wd_{d}$.
\item The special properties of the highest weight vectors can be seen to
arise from the closedness of the associated integration surface ---
a version of Stokes' formula can be used in these cases without boundary
terms (Lemma~\ref{lem: general integration by parts} and 
Corollary~\ref{cor: integration by parts for hwv}).
\item The most common way to obtain M\"obius covariance in the Coulomb gas
formalism of conformal field theory is to ensure a charge neutrality,
which takes into account a background charge. The integrals associated
to vectors in the trivial subrepresentation do not satisfy this simplest
charge neutrality, but rather fall short of it precisely by the amount
of the background charge. Our full M\"obius covariance statement for
these vectors requires a little more work 
(Proposition~\ref{prop: full Mobius covariance}).
\item In the presentation of this article, we have opted for straightforwardness
and self-containedness. We have therefore tried to avoid invoking
results from the literature, where they would require significant
additional theory --- in particular regarding general quantum groups,
the homology theory of \cite{FW-topological_representation_of_Uqsl2},
and the structure of charged Fock spaces and properties of vertex
operators \cite{Felder-BRST_approach,FF-representations,IK-representation_theory_of_the_Virasoro_algebra}.
\end{itemize}

\subsection{\label{sub: example applications}Two example applications to Schramm-Loewner evolutions}

We next illustrate the use of
our main result by
briefly describing two applications, which both arise from the theory
of $\SLE$s.

SLEs are random curves in planar domains that were introduced by Schramm
\cite{Schramm-LERW_and_UST} as candidates of scaling limits of interfaces
in statistical mechanics models at criticality. $\SLE$s are constructed
by a growth process of the curve, encoded in a Loewner chain, in such
a way as to ensure the fundamental properties of conformal invariance
(of the law associated to different domains) and domain Markov property
(which describes the conditional law of the continuation of the curve,
given a segment of it). In all $\SLE$ variants, a single parameter $\kappa>0$
captures some of the most important properties of the curve --- in
physical terms, $\kappa$ determines the universality class of the
underlying statistical mechanics model and the central charge $c(\kappa)$
of the conformal field theory. In our applications, $\kappa$
determines the deformation parameter of the quantum group according
to $q=e^{\ii\pi4/\kappa}$.

The simplest setup of $\SLE$s concerns curves living in a simply connected
domain, starting from one marked boundary point and ending at another.
A classification result, sometimes referred to as ``Schramm's principle'',
states that such random non-self-traversing curves with domain Markov
property and conformally invariant laws are uniquely characterized
by the parameter $\kappa$. The term chordal $\SLEk$ is used for
these random curves. Chordal $\SLEk$, with particular values of $\kappa$,
are known to be the scaling limits of interfaces in the presence of the
simplest (Dobrushin) boundary conditions in various critical models
of statistical mechanics --- see e.g. \cite{Smirnov-critical_percolation,
LSW-LERW_and_UST,
CN-critical_percolation_exploration_path,
Zhan-scaling_limits_of_planar_LERW,
CDHKS-convergence_of_Ising_interfaces_to_SLE}.

\subsubsection{\label{sss: multiple SLEs}\textbf{Application to multiple $\SLE$s}}

Multiple $\SLE$s arise from trying to generalize Schramm's principle
to cases where several interfaces are present: in a simply connected
domain with $2N$ marked boundary points we may have $N$ curves connecting
pairwise the marked points.
Such processes have been studied in
\cite{Dubedat-commutation,BBK-multiple_SLEs,Graham-multiple_SLEs,LK-configurational_measure,PW-Global_multiple_SLEs_and_pure_partition_functions},
and in some cases they are known to be the scaling limits of lattice model
interfaces in the presence of alternating boundary conditions
\cite{Izyurov-PhD_thesis,CS-universality_in_2d_Ising,KS-configurations_of_FK_Ising_interfaces}.

Multiple random curves with conformally invariant laws and domain Markov
properties are no longer specified by the parameter $\kappa$ alone,
but for a fixed $\kappa<8$, a finite dimensional convex set of possible
laws exists. The extremal points of the set of multiple $\SLE$ laws were
called pure geometries in \cite{BBK-multiple_SLEs}. Coulomb gas integrals
for this problem were considered in \cite{Dubedat-Euler_integrals,Kytola-local_mgales}.
The problem of explicit description of the pure geometries is very closely related to
crossing probabilities, for which formulas using Coulomb gas integrals have been recently obtained
in the series of articles \cite{FK-solution_space_for_a_system_of_null_state_PDEs_1,
FK-solution_space_for_a_system_of_null_state_PDEs_2,FK-solution_space_for_a_system_of_null_state_PDEs_3,
FK-solution_space_for_a_system_of_null_state_PDEs_4,
FSKZ-A_formula_for_crossing_probabilities_of_critical_systems_inside_polygons}.
In comparison with the approaches of \cite{Dubedat-Euler_integrals,FK-solution_space_for_a_system_of_null_state_PDEs_1},
the Coulomb gas integrals in our approach have the advantage of treating
all the $2N$ marked points on equal footing, at the expense of increasing
the dimension of the integration surface by one.

As summarized precisely in \cite[Appendix A]{KP-pure_partition_functions_of_multiple_SLEs},
an explicit Loewner chain construction of a multiple $\SLE$ uses a partition
function \cite{Dubedat-commutation,BBK-multiple_SLEs}, which is most
convenient to write down for the reference domain of the upper 
half-plane $\bH=\set{z\in\bC\;\big|\;\im(z)>0}$ as
\begin{align*}
 & \sZ(x_{1},\ldots,x_{2N}),
\end{align*}
where the marked points on the real line $x_{i}\in\bR=\bdry\bH$,
for $i=1,\ldots,2N$, are ordered as 
\begin{align*}
 & x_{1}<\cdots<x_{2N}.
\end{align*}
A stochastic reparametrization invariance of the random curves requires
the PDEs
\begin{align}\label{eq: multiple SLE PDEs}
 & \left[\frac{\kappa}{2}\pdder{x_{i}}+\sum_{j\neq i}\left(\frac{2}{x_{j}-x_{i}}\pder{x_{j}}-\frac{2h}{(x_{j}-x_{i})^{2}}\right)\right]\sZ(x_{1},\ldots,x_{2N})=0\qquad\text{for all }i=1,\ldots,2N,
\end{align}
where $h=\frac{6-\kappa}{2\kappa}$, see \cite{Dubedat-commutation}.
Moreover, conformal invariance of the law of the random curves requires
the following M\"obius covariance: 
\begin{align}\label{eq: multiple SLE Mobius covariance}
\sZ(x_{1},\ldots,x_{2N})=\; & 
\prod_{i=1}^{2N}\Mob'(x_{i})^{h}\times\sZ(\Mob(x_{1}),\ldots,\Mob(x_{2N}))
\end{align}
for any conformal map $\Mob$ of $\bH$ onto itself, which preserves
the order of the marked points.

The pure geometries of multiple $\SLE$s are labeled by planar pair
partitions $\alpha$ of $2N$ points, and the problem is to find the
corresponding partition functions $\sZ_{\alpha}$. The requirements of
covariance~\eqref{eq: multiple SLE Mobius covariance}
and PDEs~\eqref{eq: multiple SLE PDEs} are the same for all pure geometries $\alpha$,
but the boundary conditions depend on $\alpha$.
The asymptotic behavior of $\sZ_{\alpha}$
as $|x_{j+1}-x_{j}|\to0$ depends on whether the points indexed
$j$ and $j+1$ form a pair of $\alpha$ or not
--- see \cite{BBK-multiple_SLEs, KP-pure_partition_functions_of_multiple_SLEs, PW-Global_multiple_SLEs_and_pure_partition_functions}.
More precisely, when $x_{j-1} < \xi < x_{j+2}$, we have
\begin{align}\label{eq: multiple SLE asymptotics}
\lim_{x_{j},x_{j+1}\to\xi}
\frac{\sZ_{\alpha}(x_{1},\ldots,x_{2N})}{|x_{j+1}-x_{j}|^{\Delta}}
=\; & \begin{cases}
0\quad & \text{if }\set{j,j+1}\notin\alpha\\
\sZ_{\hat{\alpha}}(x_{1},\ldots,x_{j-1},x_{j+2},\ldots,x_{2N}) & \text{if }\set{j,j+1}\in\alpha
\end{cases},
\end{align}
where $\Delta=-2h=\frac{\kappa-6}{\kappa}$, and $\hat{\alpha}$ is
the planar pair partition of $2N-2$ points obtained from $\alpha$
by removing the pair $\set{j,j+1}$.

The results of the present article are used in \cite{KP-pure_partition_functions_of_multiple_SLEs}
to explicitly construct solutions to 
\eqref{eq: multiple SLE PDEs}, \eqref{eq: multiple SLE Mobius covariance}, 
and \eqref{eq: multiple SLE asymptotics}, in the form $\sZ_{\alpha} = \sF[v_{\alpha}]$,
where $\sF$ is our spin chain - Coulomb gas correspondence map.
Specifically, one forms the $2N$-fold tensor product of two-dimensional 
irreducibles
$\Wd_{2}$ of the quantum group. The trivial subrepresentation in
this tensor product, consisting of vectors $v\in\Wd_{2}^{\tens2N}$
such that $E.v=0$ and $K.v=v$, is of dimension equal to the Catalan
number $\Catalan_{N}=\frac{1}{N+1}\binom{2N}{N}$, which coincides with
the number of planar pair partitions of $2N$ points. One then wants
to judiciously choose in this subrepresentation $\Catalan_{N}$ linearly
independent vectors
\begin{align*}
 & v_{\alpha}\in\Wd_{2}^{\tens2N},
\end{align*}
indexed by the the planar pair partitions $\alpha$, so that 
$\sZ_{\alpha} = \sF[v_{\alpha}]$ will be the desired multiple $\SLE$
partition functions of the pure geometry $\alpha$.

For vectors $v_{\alpha}\in\Wd_{2}^{\tens2N}$, the $2N$ differential
equations of order two, guaranteed by the (PDE) part of our main theorem,
turn out to be exactly the equations~\eqref{eq: multiple SLE PDEs}
needed for the reparametrization invariance of the multiple $\SLE$. Moreover,
the full M\"obius covariance guaranteed by the (COV) part is
exactly~\eqref{eq: multiple SLE Mobius covariance}, as we wanted.
The main task is then to choose $v_{\alpha}$ in such a way that the
boundary conditions~\eqref{eq: multiple SLE asymptotics} are satisfied.
For this, the (ASY) part will be used.

The asymptotics property (ASY) refers to the decomposition of a tensor
product of two representations into irreducible subrepresentations,
which in this case simply reads
\begin{align*}
\Wd_{2}\tens\Wd_{2}\isom\; & \Wd_{1}\oplus\Wd_{3}.
\end{align*}
For this case, the possible exponent values appearing in the statement
(ASY) are $\Delta_{1}^{2,2}=-2h=\frac{\kappa-6}{\kappa}$ and $\Delta_{3}^{2,2}=\frac{2}{\kappa}>\Delta_{1}^{2,2}$.
Letting $\hat{\pi}_{j,j+1}^{(1)}$ denote the projection determined
by picking the one-dimensional irreducible in the direct sum decomposition
of the tensor product of the $j$:th and $j+1$:st factors, 
Equation~\eqref{eq: multiple SLE asymptotics} will be guaranteed if we have
\begin{align*}
\hat{\pi}_{j,j+1}^{(1)}(v_{\alpha})=\; & \begin{cases}
0\qquad & \text{if }\set{j,j+1}\notin\alpha\\
\frac{1}{B} \, v_{\hat{\alpha}}\quad & \text{if }\set{j,j+1}\in\alpha
\end{cases},
\end{align*}
where $v_{\hat{\alpha}}\in\Wd_{2}^{\tens(2N-2)}$ is the vector corresponding
to the planar pair partition $\hat{\alpha}$ of $2N-2$ points obtained from $\alpha$
by removing the pair $\set{j,j+1}$.

In \cite{KP-pure_partition_functions_of_multiple_SLEs}, this problem
is analyzed in detail, and in particular, it is shown that
there is a unique collection of vectors
$v_{\alpha}$ satisfying the above requirements,
up to an overall normalization. The explicit construction
of the multiple $\SLE$ partition functions $\sZ_{\alpha} = \sF[v_{\alpha}]$
for the pure geometries thus crucially relies on the results of the present article.

\subsubsection{\textbf{Application to boundary visit probabilities for chordal $\SLE$}}

For concrete probabilistic information about the $\SLE$ random curves,
it is natural to study the probabilities for $\SLE$ curves to visit
small neighborhoods of given points.
In fact, probabilities to visit infinitesimal neighborhoods
can be appropriately renormalized to obtain finite
amplitudes known as $\SLE$ Green's functions, see
\cite{LS-natural_parametrization_for_SLE,
AKL-the_Greens_function_for_radial_SLE,
LW-multi_point_Greens_functions_for_SLE,
LZ-SLE_curves_and_natural_parametrization,
Lawler-Minkowski_content_of_the_intersection_of_SLE_wuth_real_line}.
A second application of our method, considered in \cite{JJK-SLE_boundary_visits}, concerns
finding explicit formulas for the order-refined multi-point boundary Green's function of the chordal
SLE, i.e., the probability amplitude for visits to several
boundary points in a prescribed order.

For simplicity, we fix the
domain to be the upper half-plane $\bH=\set{z\in\bC\;\big|\;\im(z)>0}$,
the starting point at $x\in\bR$ and the end point at $\infty$. We order
the points to be visited on the real line $\bR=\bdry\bH$, and label
them by a superscript $-$ or $+$, according to whether they are
on the left or right of the starting point $x$, and thus denote
\begin{align*}
 & y_{L}^{-}<\cdots<y_{2}^{-}<y_{1}^{-}<x<y_{1}^{+}<y_{2}^{+}<\cdots<y_{R}^{+},
\end{align*}
where $L$ and $R$ are the number of points to be visited on the
left and right, respectively. With a given order $\omega$ of visits,
a suitably renormalized probability of visits
\begin{align*}
 & \Ampl_{\omega}(y_{L}^{-},\ldots,y_{1}^{-},x,y_{1}^{+},\ldots,y_{R}^{+})
\end{align*}
can then be defined, see \cite{JJK-SLE_boundary_visits} for details.
The task is to find an explicit expression for it.
This function $\Ampl_{\omega}$ is by construction translation invariant and homogeneous.
It can moreover be argued to satisfy linear homogeneous partial differential
equations: a second order equation, and $L+R$ third order equations.
These PDEs are of the Benoit~\& Saint-Aubin type \eqref{eq: BSA PDE for a function}.
They do not depend on the order $\omega$ of visits,
but again the boundary conditions do. When two of the arguments are close
to each other, the probabilities of visits are asymptotic to similar
ones with one point removed from the list of visits. Here we content
ourselves to noting that these conditions amount to specifying the
asymptotic behavior of $\Ampl_{\omega}$ on a codimension one boundary
of its domain of definition: they concern the cases when either 
$|y_{i+1}^{\pm}-y_{i}^{\pm}|\to0$
or $|y_{1}^{\pm}-x|\to0$. For the full list of equations and boundary
conditions for this chordal $\SLE$ boundary visit question, we refer
to \cite{JJK-SLE_boundary_visits}.

To apply the correspondence of the present article to the problem of finding
these boundary visit probabilities of chordal $\SLE$, one considers the tensor product
representation
\begin{align*}
 & \Wd_{3}^{\tens R}\tens\Wd_{2}\tens\Wd_{3}^{\tens L}.
\end{align*}
Again the task is to judiciously choose vectors $v_{\omega}$ in it,
such that the multi-point boundary Green's function corresponding to
the visit order $\omega$ can be obtained in the form 
$\Ampl_{\omega} = \sF[v_{\omega}]$.
The highest weight vector condition $E.v_{\omega}=0$ guarantees, by
the (PDE) part of our main theorem, the desired second and third order partial differential
equations, and the Cartan generator eigenvalue equation 
$K.v_{\omega}=q\, v_{\omega}$
guarantees, by the (COV) part, the correct homogeneity degree in
addition to translation invariance. The most subtle requirement
is the boundary conditions when either $|y_{i+1}^{\pm}-y_{i}^{\pm}|\to0$
or $|y_{1}^{\pm}-x|\to0$. The (ASY) part is again suitable for
this purpose: the decomposition $\Wd_{3}\tens\Wd_{3}\isom\Wd_{1}\oplus\Wd_{3}\oplus\Wd_{5}$
applies to the case $|y_{i+1}^{\pm}-y_{i}^{\pm}|\to0$, and the decompositions
$\Wd_{3}\tens\Wd_{2}\isom\Wd_{2}\oplus\Wd_{4}$ and $\Wd_{2}\tens\Wd_{3}\isom\Wd_{2}\oplus\Wd_{4}$
apply respectively to the cases $|y_{1}^{+}-x|\to0$ and $|x-y_{1}^{-}|\to0$.
SLE boundary visits are treated in more detail 
in~\cite{JJK-SLE_boundary_visits},
and the requirements for the vectors $v_{\omega}$
are shown to have a unique solution in general in \cite[Section~5]{KP-pure_partition_functions_of_multiple_SLEs}.
Again, the main results of the
present article are thus instrumental for finding the explicit formulas
for these order-refined multi-point boundary visit probabilities of
the chordal $\SLE$.

\subsection{\label{sub: another example application}Application to monodromy invariant bulk correlation functions}

The first applications of the spin chain~- Coulomb gas correspondence described above
concern boundary correlation functions relevant for random conformally invariant curves.
We next briefly describe an application to bulk correlation functions of conformal field theory,
presented in detail in \cite{FP-monodromy_invariant_correlations}.

In statistical mechanics models at criticality, the scaling limits of both boundary and
bulk correlation functions are argued to be described by conformal field theory, and
in particular, to be M\"obius covariant and satisfy partial differential equations such as the
ones studied in this article. The solutions of the partial differential
equations are in general multivalued, whereas bulk correlation functions
in statistical mechanics models are usually manifestly single-valued.
In \cite{FP-monodromy_invariant_correlations},
Flores and Peltola show how to use the spin chain~- Coulomb gas correspondence 
of the present article
to construct such bulk correlation functions of conformal field theory.
The key idea is again to translate the task of finding the appropriate functions 
to the task of finding vectors with corresponding properties in a finite-dimensional
representation of the quantum group.

In the construction of \cite{FP-monodromy_invariant_correlations},
the bulk correlation function of $n$ spinless primary fields of conformal
dimensions $2 h_{1,d_j}$, with $j=1, \ldots, n$, is associated to a vector
in the tensor product representation 
$\big( \bigotimes_{j=1}^n \Wd_{d_j} \big)^{\tens 2}$.
The function associated to a simple tensor $v \tens \bar{v}$, with
$v,\bar{v} \in \bigotimes_{j=1}^n \Wd_{d_j}$, is
\begin{align}\label{eq: bulk construction}
\sF^{(x_0)}[v](z_1 , \ldots, z_n) \times \sF^{(x_0)}[\bar{v}](\bar{z}_1 , \ldots, \bar{z}_n) ,
\end{align}
where the two factors
are the analytic continuations of the functions
$\sF^{(x_0)}[v], \sF^{(x_0)}[\bar{v}] \colon \chamber_n^{(x_0)} \to \bC$
constructed and studied in the present article.
For a function of the form~\eqref{eq: bulk construction},
the partial differential equations and covariance required of bulk correlation functions 
are guaranteed by (PDE) and (COV) parts
of our main theorem, respectively,
when the components $v$ and $\bar{v}$ lie in the trivial subrepresentation of
$\bigotimes_{j=1}^n \Wd_{d_j}$.

In general, however, the function~\eqref{eq: bulk construction} is multivalued, and
its monodromy around a loop that winds $z_j$ positively around $z_{j+1}$
is given by the action of the R-matrix of the quantum group on the $j$:th
and $j+1$:st components of $v$ and a conjugate R-matrix on the $j$:th
and $j+1$:st components of $\bar{v}$.
Arbitrary monodromies are described by the representation of the pure braid 
group
on $n$ strands generated by such loops.
The function is single-valued precisely when the corresponding vector is
invariant under the pure braid group.
Using a generalization of a quantum Schur-Weyl duality~\cite{FP-quantum_Schur_Weyl}, Flores and Peltola show
the uniqueness (up to scalar multiples) of 
a vector $\sum_k v_k \tens \bar{v}_k$ which is both
braiding invariant and for which $v_k , \bar{v}_k$ lie in the trivial subrepresentation.
This gives the uniqueness of and an explicit expression for the single-valued
bulk correlation function.

The spin chain~- Coulomb gas correspondence developed in the present article thus
also underlies the construction and analysis in \cite{FP-monodromy_invariant_correlations}
of the monodromy invariant bulk correlation function of conformal field theory.

\subsection{Organization of this article}

In Section~\ref{sec: q stuff}, we introduce notation, fix conventions,
and prove auxiliary results about $q$-combinatorics, the quantum
group $\Uqsltwo$, its representations, and their tensor products.
Section~\ref{sec: various integral functions} contains definitions
and properties of functions that are used in defining and studying
the correspondence. This part does not use the quantum group in any
way --- only some $q$-combinatorial lemmas are used.
Section~\ref{sec: FW correspondence}
begins with the definition of the correspondence, and proceeds with the
proofs of the properties stated in the main theorems in the order that we
have found the most straightforward. It concludes with the precise
statement of the main result, Theorems~\ref{thm: SCCG correspondence non-hwv}
and \ref{thm: SCCG correspondence hwv}.
In Section~\ref{sec: further properties}, we treat two further
properties: a generalization of the asymptotics statement and
a hidden manifestation of the periodicity of the domain boundary.
From the main text, we postpone some necessary but routine
contour manipulation arguments to Appendix~\ref{app: contour manipulations},
and the lengthy proof of a lemma about differential operators to
Appendix~\ref{app: exact form lemma}.
Finally, in Appendix~\ref{app: alternative conventions}, 
we give formulations of our main results with alternative conventions,
which may occasionally appear more natural.

\subsection*{Acknowledgments}

We thank Michel Bauer, Denis Bernard, Dmitry Chelkak, Steven Flores,
Philippe Di Francesco, Azat Gainutdinov, Christian Hagendorf,
Cl\'ement Hongler, Konstantin Izyurov, Niko Jokela, Matti J\"arvinen,
Peter Kleban, Hubert Saleur, and Jacob Simmons for interesting discussions,
useful comments, and suggested improvements. We also thank the anonymous referees
for useful suggestions.

This work was supported by the Academy of Finland. E.P. was supported by
the Finnish National Doctoral Programme in Mathematics and its Applications
and Vilho, Yrj\"o and Kalle V\"ais\"al\"a Foundation. The work was carried out
while E.P. was affiliated with the University of Helsinki.

\bigskip{}

\section{\label{sec: q stuff}The quantum group}

The main purpose of this section is to fix notation and conventions
about the quantum group $\Uqsltwo$. We also include auxiliary results
of $q$-combinatorial flavor, which are needed later on in the article.

\subsection{\label{sub: q numbers and combinatorics}Q-numbers and some combinatorial formulas}

Let $q\in\bC\setminus\set 0$, and assume further that $q$ is not
a root of unity, i.e. $q^{m}\neq1$ for all $m\in\bZ\setminus\set 0$.
Define, for $m\in\bZ$ and for $n,k\in\bN$, $0\leq k\leq n$, the
$q$-integers as
\begin{align}
\qnum m=\; & \frac{q^{m}-q^{-m}}{q-q^{-1}} \label{eq: q num basic def}\\
=\; & q^{m-1}+q^{m-3}+\cdots+q^{3-m}+q^{1-m}, \label{eq: q num geometric series}
\end{align}
the $q$-factorials as
\begin{align}\label{eq: q factorial}
\qfact n=\; & \prod_{m=1}^{n}\qnum m,
\end{align}
and the $q$-binomial coefficients as
\begin{align}\label{eq: q binomial coefficient}
\qbin nk=\; & \frac{\qfact n}{\qfact k \qfact{n-k}}.
\end{align}

The following ``$q$-combinatorial formulas'' will be used in later
calculations: specifically,
part~(a) will be used in Lemma~\ref{lem: general integration by parts},
part~(b) in Lemma~\ref{lem: simplex in terms of hypercube},
part~(c) in Lemma~\ref{lem: two variable FW basis function},
and part~(d) in Lemma~\ref{lem: two point closed integration surface}.

\begin{lem}\label{lem: q-combinatorics}\
\begin{description}
\item [{(a)}] We have
\begin{align*}
\qnum{\ell}\qnum{d-\ell}
=\; & \frac{1}{q-q^{-1}}\sum_{u=0}^{\ell-1}(q^{d-1-2u}-q^{-d+1+2u}).
\end{align*}

\item [{(b)}] For a permutation $\sigma\in \symgrp_{n}$ of $\set{1,\ldots,n}$,
denote the set of inversions of $\sigma$ by
\begin{align*}
{\rm inv}(\sigma)=\; & \set{(i,j) \;\big|\; i<j\text{ and }\sigma(i)>\sigma(j)}.
\end{align*}
Then we have 
\begin{align*}
\sum_{\sigma\in \symgrp_{n}}q^{-2\times\#{\rm inv}(\sigma)}
=\; & q^{-{\binom{n}{2}}}\qfact n.
\end{align*}

\item [{(c)}] We have
\begin{align*}
\sum_{1\leq r_{1}<r_{2}<\cdots<r_{k}\leq n}q^{-2\sum_{j=1}^{k}(r_{j}-j)}
=\; & q^{-k(n-k)}\qbin nk.
\end{align*}

\item [{(d)}] We have, for any $\beta$ and $n\in\bN$
\begin{align*}
\sum_{m=0}^{n}\qbin nm(-1)^{m}q^{m\beta}
=\; & q^{\half n\beta}\prod_{s=0}^{n-1}(q^{\half(n-1-\beta)-s}-q^{\half(\beta+1-n)+s}).
\end{align*}

\end{description}
\end{lem}
\begin{proof}
For part (a), use the definition $\qnum{d-\ell}=\frac{q^{d-\ell}-q^{\ell-d}}{q-q^{-1}}$
and the finite geometric series $\qnum{\ell}=q^{\ell-1}+q^{\ell-3}+\cdots+q^{-\ell+1}$
to get the asserted formula
\begin{align*}
(q-q^{-1})\qnum{\ell}\qnum{d-\ell}
=\; & (q^{d-\ell}-q^{\ell-d})(q^{\ell-1}+q^{\ell-3}+\cdots+q^{-\ell+1})\\
=\; & (q^{d-1}+q^{d-3}+\cdots+q^{d-2\ell+1})-(q^{2\ell-1-d}+q^{2\ell-3-d}+\cdots+q^{-d+1}).
\end{align*}
For part (b), use the same finite geometric series to rewrite the
right hand side as
\begin{align*}
q^{-\binom{n}{2}}\qfact n
=\; & \prod_{k=1}^{n}\left(\sum_{l=0}^{k-1}q^{-2l}\right).
\end{align*}
The inversions $(i,j)\in{\rm inv}(\sigma)$ of $\sigma\in \symgrp_{n}$
can be grouped according to the smaller index $i$: we have
\begin{align*}
\#{\rm inv}(\sigma)
=\; & \sum_{i=1}^{n}\#{\rm inv}_{i}(\sigma),\qquad\text{where }\\
{\rm inv}_{i}(\sigma)=\; & \set{j \;\big|\; j>i\text{ and }\sigma(j)<\sigma(i)}.
\end{align*}
In fact, permutations $\sigma$ are in bijection with the sequences
$\big(\#{\rm inv}_{i}(\sigma)\big)_{i=1}^{n}$. Then, in the expansion
of the product $\prod_{k=1}^{n}(\sum_{l=0}^{k-1}q^{-2l})$ the choice
of the term in the $k$:th factor can be thought of as corresponding
to the choice of $\#{\rm inv}_{n+1-k}(\sigma)$. The product is thus
expanded as a sum over permutations $\sigma$, with coefficients $\prod_{i=1}^{n}q^{-2\#{\rm inv}_{i}(\sigma)}=q^{-2\#{\rm inv}(\sigma)}$.

The proofs of (c) and (d) are both based on $q$-Pascal triangles.
Let $L_{k}^{(n)}$ and $R_{k}^{(n)}$ denote the left hand side and
right hand side of assertion (c), respectively. Splitting the sum
defining $L_{k}^{(n)}$ according to whether $r_{k}=n$ or $r_{k}<n$,
we obtain a Pascal triangle type recursion
\begin{align*}
L_{k}^{(n)}=\; & q^{-2(n-k)}L_{k-1}^{(n-1)}+L_{k}^{(n-1)}.
\end{align*}
A straightforward calculation, using the definition~\eqref{eq: q num basic def}
of $q$-integers, shows that we also have the recursion
$R_{k}^{(n)}=q^{-2(n-k)}R_{k-1}^{(n-1)}+R_{k}^{(n-1)}$.
The equality $L_{k}^{(n)}=R_{k}^{(n)}$ follows from this recursion,
together with the initial observation $L_{k}^{(1)}=R_{k}^{(1)}$.

For the proof of (d), proceed by induction on $n$. For $n=0$ both
sides are equal to $1$. Let $R^{(n)}(\beta)$ denote the right hand
side, and note that we can write
\begin{align*}
R^{(n+1)}(\beta)=\; & (1-q^{n+\beta})\times R^{(n)}(\beta-1).
\end{align*}
Assuming the asserted formula for $R^{(n)}(\beta-1)$, we expand in
powers of $q^{\beta}$ as follows
\begin{align*}
R^{(n+1)}(\beta)
=\; & \sum_{m}(-1)^{m}q^{m\beta}\left(q^{-m}\qbin nm+q^{1+n-m}\qbin n{m-1}\right).
\end{align*}
It now suffices to apply the following simple identity
\begin{align*}
q^{-m}\qbin nm+q^{1+n-m}\qbin n{m-1}=\; & \qbin{n+1}m.
\end{align*}
\end{proof}

\subsection{\label{sub: the quantum group}The quantum group and its representations}

We now give a definition of the quantum group $\Uqsltwo$ by generators
and relations. We also concretely describe the irreducible representations
$\Wd_{d}$, and record needed results about the decompositions of their tensor
products.

\subsubsection{\label{subsub: definition of the quantum group}\textbf{Definition of the quantum group}}

The quantum group $\Uqsltwo$ is the associative unital algebra
over $\bC$ generated by $E,F,K,K^{-1}$ subject to the relations
\begin{align}
KK^{-1}=\; & 1=K^{-1}K,\qquad KE=q^{2}EK,\qquad KF=q^{-2}FK,
\label{eq: quantum group relations}\\
EF-FE=\; & \frac{1}{q-q^{-1}}\left(K-K^{-1}\right).\nonumber 
\end{align}
There is a unique Hopf algebra structure on $\Uqsltwo$ with the coproduct,
an algebra homomorphism
\begin{align*}
\Hcp\;\colon\; & \Uqsltwo\rightarrow\Uqsltwo\tens\Uqsltwo,
\end{align*}
given on the generators by the expressions
\begin{align}\label{eq: coproduct}
\Hcp(E) = E\tens K+1\tens E, \qquad \Hcp(K) =  K\tens K, \qquad
\Hcp(F) = F\tens1+K^{-1}\tens F. 
\end{align}
With the coproduct, we make the tensor product of two representations
$M'$ and $M''$ again a representation. The action of $\Uqsltwo$
on $M'\tens M''$ is defined so that if
\begin{align*}
\Hcp(X)=\; & \sum_{i}X_{i}'\tens X_{i}''\,\in\,\Uqsltwo\tens\Uqsltwo
\end{align*}
and $v'\in M'$, $v''\in M''$, then
\begin{align*}
X.(v'\tens v'')=\; & \sum_{i}(X_{i}'.v')\tens(X_{i}''.v'')\,\in\, M'\tens M''.
\end{align*}
Note that we generally cannot canonically identify $M'\tens M''$
with $M''\tens M'$ as representations, because the coproduct $\Hcp$
is not cocommutative. However, the coproduct is coassociative, that
is $(\id\tens\Hcp)\circ\Hcp=(\Hcp\tens\id)\circ\Hcp$, and therefore
the canonical identification $(M'\tens M'')\tens M'''\isom M'\tens(M''\tens M''')$
is an isomorphism of representations. More generally, we may talk
about multiple tensor products without specifying the positions of
parentheses. For calculations with $n$-fold tensor products, one needs
the $(n-1)$-fold coproduct 
$\Hcp^{(n)}\colon\Uqsltwo\rightarrow\Big(\Uqsltwo\Big)^{\tens n}$
\begin{align*}
\Hcp^{(n)}=\; & (\Hcp\tens\id^{\tens(n-2)})\circ(\Hcp\tens\id^{\tens(n-3)})\circ\cdots\circ(\Hcp\tens\id)\circ\Hcp.
\end{align*}

We record the following easily verified expressions
for the $(n-1)$-fold coproducts of the generators
for later use in e.g. Lemmas~\ref{lem: asymptotics for tensor basis in two consequtive} and~\ref{lem: isomorphism of hwv spaces},
Propositions~\ref{prop: full Mobius covariance} and~\ref{prop: general asymptotics with subrepresentations},
and Corollary~\ref{cor: integration by parts for hwv}.
\begin{lem}\label{lem: multiple coproducts}
We have
\begin{align*}
\Hcp^{(n)}(K)=\; & K^{\tens n}\\
\Hcp^{(n)}(E)=\; & \sum_{i=1}^{n}1^{\tens(n-i)}\tens E\tens K^{\tens(i-1)}\\
\Hcp^{(n)}(F)=\; & \sum_{i=1}^{n}(K^{-1})^{\tens(n-i)}\tens F\tens1^{\tens(i-1)}.
\end{align*}
\end{lem}

\subsubsection{\label{subsub: irreps of the quantum group}\textbf{Irreducible representations of the quantum group}}

We will use representations which can be thought of as $q$-deformations
of the irreducible representations of the semisimple Lie algebra $\sl_{2}$.
For the statement of the lemma below, recall the definition of $q$-integers
$\qnum m$ from Equation~\eqref{eq: q num basic def}.
\begin{lem}\label{lem: representations of quantum sl2}
For every positive integer
$d$, there is an irreducible representation $\Wd_{d}$ of $\Uqsltwo$
with a basis $\Wbas_{0}^{(d)},\Wbas_{1}^{(d)},\ldots,\Wbas_{d-1}^{(d)}$
and the action of the generators defined by
\begin{align*}
K.\Wbas_{l}^{(d)}=\; & q^{d-1-2l}\,\Wbas_{l}^{(d)}\\
F.\Wbas_{l}^{(d)}=\; & \begin{cases}
\Wbas_{l+1}^{(d)} & \text{if }l\neq d-1\\
0 & \text{if }l=d-1
\end{cases}\\
E.\Wbas_{l}^{(d)}=\; & \begin{cases}
\qnum l\qnum{d-l} \Wbas_{l-1}^{(d)} & \text{if }l\neq0\\
0 & \text{if }l=0
\end{cases}.
\end{align*}
Any $d$-dimensional irreducible representation of $\Uqsltwo$, where
the $K$-eigenvalues are integer powers of $q$, is isomorphic to
$\Wd_{d}$.\end{lem}
\begin{proof}
It is easy to check that the formulas defining the action respect
the relations~\eqref{eq: quantum group relations}. Moreover, since
the $q$-integers $\qnum m$ are non-vanishing for $m\neq0$, $\Wd_{d}$
is clearly irreducible. If $V$ is an irreducible representation,
then $K$ is diagonalizable on $V$, because the sum of $K$-eigenspaces
is a subrepresentation. If in a finite dimensional irreducible representation
$V$ the $K$-eigenvalues are integer powers of $q$, then an eigenvector
$v_{0}$ of $K$ of eigenvalue $q^{m_{0}}$ with $m_{0}\in\bZ$ maximal
must be annihilated by $E$. By a standard calculation one then shows
that the linear span of the $m_{0}+1$ vectors $v_{0},F.v_{0},F^{2}.v_{0},\ldots,F^{m_{0}}.v_{0}$
forms a subrepresentation isomorphic to $\Wd_{m_{0}+1}$. The last
assertion follows from this.
\end{proof}

\subsubsection{\label{subsub: quantum Clebsch-Gordan}\textbf{Tensor products of the irreducible representations}}

Tensor products of the representations defined in 
Section~\ref{subsub: irreps of the quantum group}
decompose to direct sums of irreducible subrepresentations of the
same type. Concrete descriptions of such decompositions, as
given in the following $q$-analogue of the Clebsch-Gordan formulas,
will be needed especially in
Lemmas~\ref{lem: two point closed integration surface}~--~\ref{lem: asymptotics for tensor basis in two consequtive},
Propositions~\ref{prop: asymptotics with subrepresentations}~--~\ref{prop: 
closed integration surface}, and in
Section~\ref{sub: cyclic permutations}.
Recall that the action of $\Uqsltwo$ on tensor products is defined
using the coproduct~\eqref{eq: coproduct}, and recall also the definition
of $q$-factorials $\qfact n$ from Equation~\eqref{eq: q factorial}.
\begin{lem}\label{lem: tensor product representations of quantum sl2}
Consider the tensor product representation $\Wd_{d''}\tens\Wd_{d'}$.
For any 
\begin{align*}
m\in\; & \set{0,1,\ldots,\min(d',d'')-1},
\end{align*}
denote $d=d'+d''-1-2m$, and define
\begin{align}\label{eq: tensor product hwv coefficients}
T_{0;m}^{l',l''} = 
T_{0;m}^{l',l''}(d',d'')=\; & \delta_{l'+l'',m}\times(-1)^{l'}\frac{\qfact{d'-1-l'}
\,\qfact{d''-1-l''}}{\qfact{l'}\qfact{d'-1}\qfact{l''}\qfact{d''-1}}
\,\frac{q^{l'(d'-l')}}{(q-q^{-1})^{m}}
\end{align}
Then the vector
\begin{align}\label{eq: tensor product hwv}
\Tbas_{0}=\; & \Tbas_{0}^{(d;d',d'')}
=\sum_{l',l''}T_{0;m}^{l',l''} 
\times(\Wbas_{l''}\tens\Wbas_{l'})
\end{align}
satisfies $E.\Tbas_{0}=0$ and $K.\Tbas_{0}=q^{d-1}\,\Tbas_{0}$ (i.e.,
$\tau_{0}$ is a highest weight vector). The subrepresentation of
$\Wd_{d''}\tens\Wd_{d'}$ generated by $\Tbas_{0}$ is isomorphic
to $\Wd_{d}$.

The tensor product representation decomposes to a direct sum of irreducibles
\begin{align}\label{eq: decomposition of tensor product}
\Wd_{d''}\tens\Wd_{d'}\isom\; & \Wd_{d'+d''-1}\oplus\Wd_{d'+d''-3}\oplus\cdots\oplus\Wd_{|d'-d''|+3}\oplus\Wd_{|d'-d''|+1}.
\end{align}
\end{lem}
\begin{proof}
Because the coefficients $T_{0;m}^{l',l''}$ above are non-zero only
for $l'+l''=m$, we can write the vector $\Tbas_{0}$ as
$\Tbas_{0} = \sum_{k}T_{0;m}^{k,m-k}\times(\Wbas_{m-k}\tens\Wbas_{k})$.
Since $\Hcp(K)=K\tens K$, we have
\begin{align*}
K.(\Wbas_{m-k}\tens\Wbas_{k})=\; & (q^{d''-1-2(m-k)}\,
\Wbas_{m-k})\tens(q^{d'-1-2k}\Wbas_{k})=q^{d-1}(\Wbas_{m-k}\tens\Wbas_{k})
\end{align*}
with $d=d'+d''-1-2m$, and therefore obviously 
$K.\Tbas_{0}=q^{d-1}\,\Tbas_{0}$. Using $\Hcp(E)=E\tens K+1\tens E$, we get
\begin{align*}
E.(\Wbas_{m-k}\tens\Wbas_{k})=\; & \qnum{m-k}\qnum{d''-m+k}q^{d'-1-2k}\left(\Wbas_{m-k-1}\tens\Wbas_{k}\right)+\qnum k\qnum{d'-k}\left(\Wbas_{m-k}\tens\Wbas_{k-1}\right),
\end{align*}
and thus
\begin{align*}
E.\Tbas_{0}=\; & \sum_{k}\Bigg(\qnum{m-k}\qnum{d''-m+k}q^{d'-1-2k}T_{0;m}^{k,m-k} \\
& \qquad \qquad \qquad  + \qnum{k+1}\qnum{d'-k-1}T_{0;m}^{k+1,m-1-k}\Bigg)\times\left(\Wbas_{m-k-1}\tens\Wbas_{k}\right).
\end{align*}
The coefficients satisfy 
$T_{0;m}^{k+1,m-1-k} 
= -T_{0;m}^{k,m-k}\times
\frac{\qnum{m-k}\qnum{d''-m+k}}{\qnum{k+1}\qnum{d'-k-1}}\,q^{d'-1-2k}$,
and hence, we get $E.\Tbas_{0}=0$.

From Lemma~\ref{lem: representations of quantum sl2} and the properties
$E.\Tbas_{0}=0$ and $K.\Tbas_{0}=q^{d-1}\,\Tbas_{0}$ it is clear that
the vector $\Tbas_{0}$ generates a subrepresentation of $\Wd_{d''}\tens\Wd_{d'}$
isomorphic to $\Wd_{d}$. The dimension $d'd''$ of the representation
$\Wd_{d''}\tens\Wd_{d'}$ equals the sum of dimensions $d=d'+d''-1-2m$
over the allowed values of $m$. Thus, the entire representation is
a direct sum of these subrepresentations.
\end{proof}

\begin{rem}\label{rem: tensor product different bases}
\emph{
In view of Equation~\eqref{eq: decomposition of tensor product}, 
we may freely interpret
$\Wd_{d}$ as a subrepresentation of $\Wd_{d''}\tens\Wd_{d'}$,
with the embedding normalized so as to map the basis vectors $\Wbas_{l}^{(d)}$
of Lemma~\ref{lem: representations of quantum sl2} to the vectors
\begin{align*}
\Tbas_{l}^{(d;d',d'')}:=\; & F^{l}.\Tbas_{0}^{(d;d',d'')}
\end{align*}
with $\Tbas_{0}^{(d;d',d'')}$ given by 
the formulas~\eqref{eq: tensor product hwv}
and \eqref{eq: tensor product hwv coefficients}. We denote the coefficients
of these vectors in the tensor product basis by 
$T_{l;m}^{l',l''} = T_{l;m}^{l',l''}(d',d'')$, 
so that 
\begin{align}\label{eq: tensor product submodule basis}
\Tbas_{l}^{(d;d',d'')}=\; & \sum_{l',l''}T_{l;m}^{l',l''} \times (\Wbas_{l''}\tens\Wbas_{l'}).
\end{align}
The vectors $\Tbas_{l}^{(d;d',d'')}$ also form a basis of $\Wd_{d''}\tens\Wd_{d'}$.
}
\end{rem}
\bigskip{}

\section{\label{sec: various integral functions}Various forms of the integral functions}

In this section, we introduce the functions in terms of which the
spin chain - Coulomb gas correspondence is defined and studied. The
functions 
\begin{align*}
 & \SimplexInt_{m_{1},\ldots,m_{n}}^{(x_{0})},
 \; \CubeInt_{m_{1},\ldots,m_{n}}^{(x_{0})},
 \; \FWint_{l_{1},\ldots,l_{n}}^{(x_{0})},
 \; \AsyInt_{l_{1},\ldots,l_{j-1};l,m;l_{j+2},\ldots,l_{n}}^{(x_{0})}
\end{align*}
of $n$ variables $x_1 , \ldots , x_n$ will be indexed by
an anchor point $x_0 \in \bR$ and various $n$-tuples of non-negative
integers $m_j, l_j, l, m \in \bZ_{\geq 0}$.
All these functions will be defined by integrals of essentially the same multivalued
integrand --- the differences lie in the choice of the integration
surface and the choices of branch and rephasing of the integrand,
which are often easiest to indicate by figures. The definition of
the correspondence will use $\FWint_{l_{1},\ldots,l_{n}}^{(x_{0})}$
as basis functions, and the other functions are used for proving properties
of the correspondence. For this purpose, various properties of the
functions and relations among them are stated in this section.
Some proofs are postponed to Appendix~\ref{app: contour manipulations}.

A parameter $\kappa\in(0,\infty)\setminus\bQ$ is fixed throughout,
and the deformation parameter $q=e^{\ii\pi4/\kappa}$ is chosen. We
also fix the number $n\in\bN$ of variables, and real parameters $d_{1},\ldots,d_{n}$,
which later in Section~\ref{sec: FW correspondence} will be taken
to be dimensions of representations of type $\Wd_{d}$.

We use the shorthand notation $\boldsymbol{x}=(x_{1},\ldots,x_{n})$ for
the arguments of the functions. The domain of definition will be either
the chamber $\chamber_{n}$ or the restricted chamber $\chamber_{n}^{(x_{0})}$,
given by \eqref{eq: chamber} or \eqref{eq: restricted chamber},
respectively, so that we always assume the variables ordered according
to
\begin{align*}
 & x_{0}<x_{1}<\cdots<x_{n}.
\end{align*}

For fixed $\boldsymbol{x}$ and $x_{0}$, the value of the function will
be written as an integral of Dotsenko-Fateev type \cite{DF-multipoint_correlation_functions},
as in the Coulomb gas formalism of conformal field theory. The integrand
is a branch of the following multivalued function, a product of powers
of differences,
\begin{align}
f^{(\ell)}(\boldsymbol{x};\boldsymbol{w}) =
\; & f_{d_{1},\ldots,d_{n}}^{(\ell)}(x_{1},\ldots,x_{n};w_{1},\ldots,w_{\ell})\nonumber \\
=\; & \prod_{1\leq i<j\leq n}(x_{j}-x_{i})^{\frac{2}{\kappa}(d_{i}-1)(d_{j}-1)}\prod_{\substack{1\leq i\leq n\\
1\leq r\leq\ell
}
}(w_{r}-x_{i})^{-\frac{4}{\kappa}(d_{i}-1)}\prod_{1\leq r<s\leq\ell}(w_{s}-w_{r})^{\frac{8}{\kappa}},
\label{eq: integrand with generic phase}
\end{align}
and the auxiliary variables $w_{1},\ldots,w_{\ell}$ are to be integrated
over. More precisely, the integrand will be defined on some simply
connected subset of
\begin{align*}
\Wchamber^{(\ell)} = \; \Wchamber_{x_{1},\ldots,x_{n}}^{(\ell)}
:=\; & \set{(w_{1},\ldots,w_{\ell})\in\left(\bC\setminus\set{x_{1},\ldots,x_{n}}\right)^{\ell}\;\Big|\; w_{r}\neq w_{s}\text{ for all }r\neq s}.
\end{align*}

\begin{rem}\label{rem: defining a rephased branch of integrand}
\emph{
The logarithmic
differential of the multivalued function $f^{(\ell)}$ is the single-valued
one-form
\begin{align*}
\ud\left(\log\Big(f^{(\ell)}(\boldsymbol{x};\boldsymbol{w})\Big)\right)
=\; & \sum_{r=1}^{\ell}\left(\sum_{i=1}^{n}\frac{4(1-d_{i})/\kappa}{w_{r}-x_{i}}+\sum_{s\neq r}\frac{8/\kappa}{w_{r}-w_{s}}\right)\ud w_{r}.
\end{align*}
Thus, to define a branch of the integrand on a simply connected subset
of $\Wchamber^{(\ell)}$, it is sufficient to give its value at some
point $\boldsymbol{w}'$, and then set 
\begin{align*}
f_{{\rm branch}}(\boldsymbol{x};\boldsymbol{w})
=\; & f_{{\rm branch}}(\boldsymbol{x};\boldsymbol{w}')\times\exp\left(\int_{\boldsymbol{w'}}^{\boldsymbol{w}}\ud\left(\log\Big(f^{(\ell)}(\boldsymbol{x};\cdot)\Big)\right)\right),
\end{align*}
where the path of integration from $\boldsymbol{w'}$ to $\boldsymbol{w}$
stays in the simply connected subset.
}
\end{rem}

We will frequently partition the variables $w_{1},\ldots,w_{\ell}$
to $n$ subsets of sizes $m_{1},\ldots,m_{n}$, in which case we use
the notation
\begin{align}\label{eq: partition w variables}
I^{(i)}=\; & I_{m_{1},\ldots,m_{n}}^{(i)}=\set{r\in\bZ\;\Bigg|\;\sum_{j=1}^{i}m_{j}\geq r>\sum_{j=1}^{i-1}m_{j}}\qquad\qquad(i=1,\ldots,n)
\end{align}
for the partition of the indices.

\subsection{\label{sub: real integral functions}Real integral functions as integrals over a product of simplices}

The integrand~\eqref{eq: integrand with generic phase} has a constant
phase on the following simply connected real subset of $\Wchamber^{(\ell)}$,
a product of simplices of dimensions $m_{1},\ldots,m_{n}$ with $\sum_{i=1}^{n}m_{i}=\ell$,
\begin{align*}
\sR_{m_{1},\ldots,m_{n}}=\Big\{(w_{1},\ldots,w_{\ell})\in\bR^{\ell}\;\Big|\;\; & x_{0}<w_{1}<w_{2}<w_{3}<\cdots<w_{m_{1}}<x_{1},\\
 & x_{1}<w_{m_{1}+1}<\cdots<w_{m_{1}+m_{2}}<x_{2},\\
 & \qquad\qquad\vdots\\
 & x_{n-1}<w_{m_{1}+\cdots+m_{n-1}+1}<\cdots<w_{m_{1}+\cdots+m_{n}}<x_{n}\Big\}.
\end{align*}
We define the real-valued functions 
$\SimplexInt_{m_{1},\ldots,m_{n}}^{(x_{0})}\colon\chamber^{(x_{0})}\to\bR$
as integrals over this set 
\begin{align}\label{eq: real integral over product of simplices}
\SimplexInt_{m_{1},\ldots,m_{n}}^{(x_{0})}(\boldsymbol{x})
:=\; & \underset{\sR_{m_{1},\ldots,m_{n}}}{\int} \big| f^{(\ell)}(\boldsymbol{x};\boldsymbol{w}) \big| \;\ud w_{1}\cdots\ud w_{\ell}.
\end{align}

In applications, it is usually desirable to write the final results
in terms of these functions, because of their transparent definition
and real-valuedness.

The integrals $\SimplexInt_{m_{1},\ldots,m_{n}}^{(x_{0})}(\boldsymbol{x})$
are convergent for large enough $\kappa$ --- the precise condition is
\begin{align}\label{eq: large enough kappa}
\kappa>\; & 4\times\left(\max_{1\leq i\leq n}d_{i}-1\right),
\end{align}
and this will often be implicitly assumed. Nevertheless, our main
results are valid for all irrational positive $\kappa$ --- they are
obtained by meromorphic analytic continuation in $\kappa$. Indeed,
the analytic continuation of the real integrals 
$\SimplexInt_{m_{1},\ldots,m_{n}}^{(x_{0})}(\boldsymbol{x})$
can be done by regularizing the divergent integrals, as discussed
in \cite{JJK-SLE_boundary_visits}. If such a regularization is performed
by the method of counterterms, one can see, in principle explicitly,
that for $\re(\kappa)>0$ the only singularities are isolated poles
at some rational values of $\kappa$.
For concrete examples of the analytic continuation and poles,
$\SimplexInt_{m_{1}}^{(x_{0})}(x_1)$ and $\SimplexInt_{0,m_{2}}^{(x_{0})}(x_1,x_2)$
can be reduced to the Selberg integral given in 
Remark~\ref{rem: Selberg integral}.

\subsection{\label{sub: deformed hypercube integral functions}Integrals over a product of deformed hypercubes}

It is natural to extend the integrand above from the real submanifold
$\sR_{m_{1},\ldots,m_{n}}\subset\Wchamber^{(\ell)}$ to an open subset.
A convenient choice for intermediate manipulations is the simply connected
subset
\begin{align*}
\widetilde{\Wchamber}_{m_{1},\ldots,m_{n}}
:=\; & \Big\{\boldsymbol{w}\in\Wchamber^{(\ell)}\;\Big|\;\forall i\;\forall r\in I^{(i)}:\;\re(x_{i-1})<\re(w_{r})<\re(x_{i}),\\
 & \qquad\qquad\qquad\forall i\;\forall r,s\in I^{(i)},\; r<s:\; w_{s}-w_{r}\in\bC\setminus\ii\,\bR_{+}\Big\}.
\end{align*}
On $\widetilde{\Wchamber}_{m_{1},\ldots,m_{n}}$, we choose a branch
of the multivalued function $f^{(\ell)}(\boldsymbol{x};\cdot)$ of 
\eqref{eq: integrand with generic phase},
and rephase it so that it becomes real and positive on $\sR_{m_{1},\ldots,m_{n}}\subset\widetilde{\Wchamber}_{m_{1},\ldots,m_{n}}$.
This function
\begin{align*}
f_{m_{1},\ldots,m_{n}}^{\approx}(\boldsymbol{x};\cdot)
\;\colon\; & \widetilde{\Wchamber}_{m_{1},\ldots,m_{n}}\rightarrow\bC
\end{align*}
can be defined for example using 
Remark~\ref{rem: defining a rephased branch of integrand},
i.e., by fixing a point $\boldsymbol{w}'\in\sR_{m_{1},\ldots,m_{n}}$,
setting the value at that point equal to the absolute value of $f^{(\ell)}(\boldsymbol{x};\boldsymbol{w}')$,
and analytically continuing by integrating the single-valued logarithmic
differential.

In particular, we can write the real integral function $\SimplexInt_{m_{1},\ldots,m_{n}}^{(x_{0})}$
as an integral of $f^{\approx}(\boldsymbol{x};\cdot)$ 
\begin{align*}
\SimplexInt_{m_{1},\ldots,m_{n}}^{(x_{0})}(\boldsymbol{x})
=\; & \int_{\sR_{m_{1},\ldots,m_{n}}}f_{m_{1},\ldots,m_{n}}^{\approx}(\boldsymbol{x};w_{1},\ldots,w_{\ell})\;\ud w_{1}\cdots\ud w_{\ell} .
\end{align*}
It will instead be easier to express the basis functions of our correspondence 
in terms of the closely related function 
\begin{align}\label{eq: deformed hypercube integral}
\CubeInt_{m_{1},\ldots,m_{n}}^{(x_{0})}(\boldsymbol{x})
:=\; & \int_{\widetilde{\sR}_{m_{1},\ldots,m_{n}}}f_{m_{1},\ldots,m_{n}}^{\approx}(\boldsymbol{x};w_{1},\ldots,w_{\ell})\;\ud w_{1}\cdots\ud w_{\ell},
\end{align}
where the integration surface $\widetilde{\sR}_{m_{1},\ldots,m_{n}}\subset\widetilde{\Wchamber}_{m_{1},\ldots,m_{n}}$
is such that for any $r\in I^{(i)}$, the variable $w_{r}$ is integrated
from $x_{i-1}$ to $x_{i}$. In view of the definition of the simply
connected set $\widetilde{\Wchamber}_{m_{1},\ldots,m_{n}}$, this
unambiguously determines the homotopy type of the integration surface
$\widetilde{\sR}_{m_{1},\ldots,m_{n}}$, and consequently the function
$\CubeInt_{m_{1},\ldots,m_{n}}^{(x_{0})}(\boldsymbol{x})$. 
Figure~\ref{fig: tilde integration}
illustrates how the variables turn around each other in this integration,
and indicates the choice of a point where the integrand is rephased
to be positive.

\noindent 
\begin{figure}
\includegraphics[width=1\textwidth]{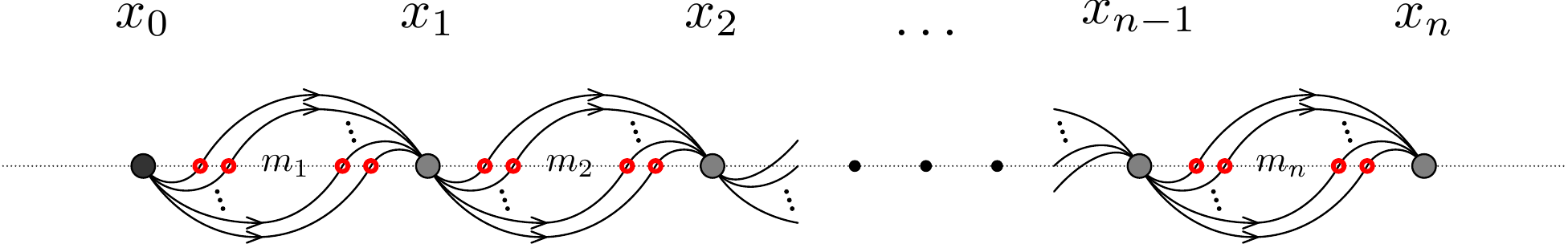}
\caption{\label{fig: tilde integration}
The integration surface $\widetilde{\sR}_{m_{1},\ldots,m_{n}}$
is a Cartesian product of deformed hypercubes of dimensions $m_{1},\ldots,m_{n}$.
The point $w_{1}<\cdots<w_{\ell}$, where the integrand 
$f_{m_{1},\ldots,m_{n}}^{\approx}(\boldsymbol{x};\cdot)$
is rephased to be positive, is marked by red circles. For $r,s\in I^{(i)}$,
$s<r$, the integration path of $w_{r}$ thus remains below the path
of $w_{s}$.}
\end{figure}

\begin{lem}\label{lem: simplex in terms of hypercube}
The function 
\begin{align*}
\CubeInt_{m_{1},\ldots,m_{n}}^{(x_{0})}
\;\colon\; & \chamber_{n}^{(x_{0})}\rightarrow\bC
\end{align*}
is related to the real integral 
function~\eqref{eq: real integral over product of simplices}
by
\begin{align*}
\CubeInt_{m_{1},\ldots,m_{n}}^{(x_{0})}(\boldsymbol{x})= & \left(\prod_{i=1}^{n}q^{-\binom{m_{i}}{2}}\qfact{m_{i}}\right)\times\SimplexInt_{m_{1},\ldots,m_{n}}^{(x_{0})}(\boldsymbol{x}),\qquad\text{for }\boldsymbol{x}\in\chamber_{n}^{(x_{0})}.
\end{align*}
\end{lem}
The proof is a straightforward contour deformation argument,
decomposing the contour in \eqref{eq: deformed hypercube integral}
to $\prod_{i=1}^n (m_i!)$ pieces. We postpone the details of the proof to Appendix~\ref{app: contour manipulations}.

\subsection{\label{sub: basis functions as FW integrals}Basis functions as integrals over families of non-intersecting loops}

Let us now define the integrals $\FWint_{l_{1},\ldots,l_{n}}^{(x_{0})}$
which will serve as our basis functions in defining the
spin chain~- Coulomb gas correspondence. The integration surfaces are certain
families of non-intersecting loops used also in \cite{FW-topological_representation_of_Uqsl2}.
In contrast to the earlier integrals $\SimplexInt_{m_{1},\ldots,m_{n}}^{(x_{0})}$
and $\CubeInt_{m_{1},\ldots,m_{n}}^{(x_{0})}$, for the convergence
of these integrals it suffices that $\mbox{\ensuremath{\re}(\ensuremath{\kappa})}>0$.
In the choices of indices, we try to consistently use $l_j \in \bZ_{\geq 0}$ for numbers of loops,
and $m_j \in \bZ_{\geq 0}$ for multiplicities of paths.

Let
\begin{align*}
l_{1},\ldots,l_{n}\in\bZ_{\geq0},\qquad & \ell=\sum_{i=1}^{n}l_{i},
\end{align*}
and define the partition $\left(I_{l_{1},\ldots,l_{n}}^{(i)}\right)_{i=1}^{n}$
of the indices of the $w$-variables as in~\eqref{eq: partition w variables},
now with parts of sizes $l_{1},\ldots,l_{n}$.
The integration surface $\SurfFW_{l_{1},\ldots,l_{n}}$ --- a family
of non-intersecting loops illustrated in Figure~\ref{fig: bag integration}
--- and the associated integrand $f_{l_{1},\ldots,l_{n}}^{\Supset}$
are defined as follows:
\begin{itemize}
\item Each of the variables $w_{r}$, for $r\in I^{(i)}$, makes a simple
loop in $\bC\setminus\set{x_{1},\ldots,x_{n}}$ starting and ending
at the anchor point $x_{0}$, and encircling the point $x_{i}$ once
in the positive direction. The loop of $w_{r}$, for $r\in I^{(i)}$,
must never cross the lines $x_{j}+\ii\bR_{+}$ for $j\neq i$, nor
the lines $w_{s}+\ii\bR_{+}$ for $s\in I^{(j)}$ with $j<i$.
\item The $l_{i}$ loops around $x_{i}$ are nested in such a way that if
$r,s\in I^{(i)}$ and $r<s$, then the loop of $w_{s}$ encircles
the loop of $w_{r}$.
\item There is a point $\boldsymbol{w}'=(w_{1}',\ldots,w_{\ell}')\in\SurfFW_{l_{1},\ldots,l_{n}}$,
also illustrated in Figure~\ref{fig: bag integration}, such that
\begin{align*}
 & x_{1}<w_{1}'<w_{2}'<w_{3}'<\cdots<w_{l_{1}}'<x_{2},\\
 & x_{2}<w_{l_{1}+1}'<w_{l_{1}+2}'<\cdots<w_{l_{1}+l_{2}}'<x_{3},\\
 & \qquad\qquad\vdots\\
 & x_{n}<w_{l_{1}+\cdots+l_{n-1}+1}'<\cdots<w_{l_{1}+\cdots+l_{n}}'.
\end{align*}
As the integrand, we choose a branch and rephasing of the function
$f^{(\ell)}(\boldsymbol{x};\cdot)$ of~\eqref{eq: integrand with generic phase},
so as to make it positive at $\boldsymbol{w}'$. Such an integrand
\begin{align*}
f_{l_{1},\ldots,l_{n}}^{\Supset}(\boldsymbol{x};\cdot)
\;\colon\; & \SurfFW_{l_{1},\ldots,l_{n}}\rightarrow\bC
\end{align*}
can again be defined using the single-valued logarithmic differential
as explained in Remark~\ref{rem: defining a rephased branch of integrand}.
\end{itemize}
The integrals of $f^{\Supset}$ over $\SurfFW$,
\begin{align}\label{eq: FW basis integral}
\FWint_{l_{1},\ldots,l_{n}}^{(x_{0})}(\boldsymbol{x})
=\; & \int_{\SurfFW_{l_{1},\ldots,l_{n}}}f_{l_{1},\ldots,l_{n}}^{\Supset}(\boldsymbol{x};w_{1},\ldots,w_{\ell})\;\ud w_{1}\cdots\ud w_{\ell},
\end{align}
define functions
\begin{align*}
\FWint_{l_{1},\ldots,l_{n}}^{(x_{0})}
\;\colon\; & \chamber_{n}^{(x_{0})}\rightarrow\bC.
\end{align*}

\noindent 
\begin{figure}
\includegraphics[width=1\textwidth]{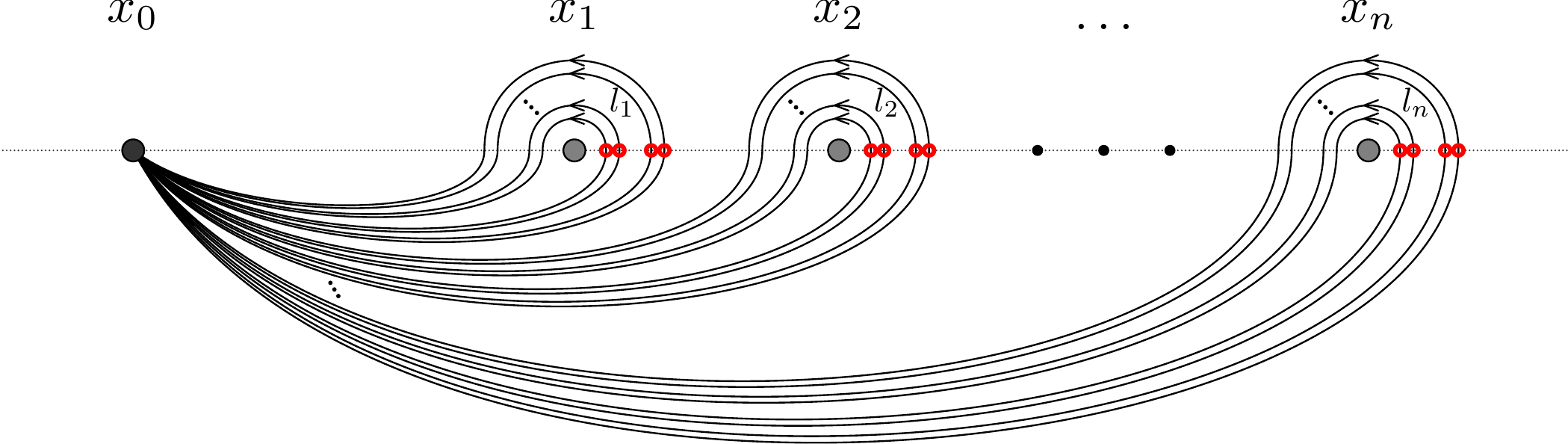}
\caption{\label{fig: bag integration}
The integration surface $\SurfFW_{l_{1},\ldots,l_{n}}$.
The point where the integrand $f_{l_{1},\ldots,l_{n}}^{\Supset}(\boldsymbol{x};\cdot)$
is rephased to be positive is marked by red circles.}
\end{figure}

The transformation rules of these functions under translation and
scaling are the following.
\begin{lem}\label{lem: translation and scaling of basis functions}
For any $\xi\in\bR$, we have
\begin{align*}
\FWint_{l_{1},\ldots,l_{n}}^{(x_{0}+\xi)}(x_{1}+\xi,\ldots,x_{n}+\xi)
=\; & \FWint_{l_{1},\ldots,l_{n}}^{(x_{0})}(x_{1},\ldots,x_{n}),
\end{align*}
and for any $\lambda>0$ we have
\begin{align*}
\FWint_{l_{1},\ldots,l_{n}}^{(\lambda x_{0})}(\lambda x_{1},\ldots,\lambda x_{n})
=\; & \lambda^{\Delta^{d_{1},\ldots,d_{n}}(\ell)}\times\FWint_{l_{1},\ldots,l_{n}}^{(x_{0})}(x_{1},\ldots,x_{n}),
\end{align*}
where $\ell=\sum_{i=1}^{n}l_{i}$ and
\begin{align*}
\Delta^{d_{1},\ldots,d_{n}}(\ell)
=\; & \frac{2}{\kappa}\sum_{i<j}(d_{i}-1)(d_{j}-1)-\frac{4}{\kappa}\ell\sum_{i}(d_{i}-1)+\frac{8}{\kappa}\frac{\ell(\ell-1)}{2}+\ell.
\end{align*}
\end{lem}
\begin{proof}
For the first statement, make the changes of variables $u_{r}=w_{r}+\xi$
in the integrals.
For the second, make the change of variables $u_{r} = \lambda w_{r}$,
and notice that the integrand has the scaling
\begin{align*}
\fFW_{l_{1},\ldots,l_{n}}(\lambda\boldsymbol{x};\lambda\boldsymbol{w})
=\; & \lambda^{\frac{2}{\kappa}\sum_{i<j}(d_{i}-1)(d_{j}-1)-\frac{4}{\kappa}\ell\sum_{i}(d_{i}-1)+\frac{8}{\kappa}\frac{\ell(\ell-1)}{2}}
\times\fFW_{l_{1},\ldots,l_{n}}(\boldsymbol{x};\boldsymbol{w}),
\end{align*}
and that the Jacobian of the change of variables in the $\ell$-dimensional
integral is $\lambda^{\ell}$.
\end{proof}
As another simple remark, we record the fact that
when $d_i=1$ for some $i$, so that the trivial representation $\Wd_1$ is a factor
in the tensor product $\bigotimes_j \Wd_{d_j}$,
then the function is actually
independent of the corresponding variable.
\begin{lem}\label{lem: trivial dependence of basis functions on singlet variables}
If $d_{i}=1$ for some $i\in\set{1,\ldots,n}$, then for any fixed values
of the other variables $x_{j}$, $j\neq i$, the function $\FWint_{l_{1},\ldots,l_{n}}^{(x_{0})}$
is constant as a function of the variable $x_{i}$.
\end{lem}
\begin{proof}
For two values of the variable $x_{i}$, one can choose the integration
surface $\SurfFW_{l_{1},\ldots,l_{n}}$ to be the same, in such a
way as to obtain the correct homotopy type for both cases. If $d_{i}=1$,
the integrand $\fFW_{l_{1},\ldots,l_{n}}(\boldsymbol{x};\boldsymbol{w})$
is constant as a function of $x_{i}$, because all the exponents of
differences in which $x_{i}$ appears are proportional to $d_{i}-1$.
Therefore, for two values of the variable $x_{i}$, the values of
the function $\FWint_{l_{1},\ldots,l_{n}}^{(x_{0})}$ are in fact
given by the same integral.
\end{proof}

The basis functions $\FWint_{l_{1},\ldots,l_{n}}^{(x_{0})}$ can be
decomposed to linear combinations of the real integrals 
$\SimplexInt_{m_{1},\ldots,m_{n}}^{(x_{0})}$.
We begin with the simplest case of one point, $n=1$. Even this simple
case will have an important consequence, 
Lemma~\ref{lem: vanishing for too many bags}.
\begin{lem}\label{lem: explicit formula for one point FW integral}
Consider the case $n=1$, and denote $d_{1}=d \in \bZpos$, $l_{1}=\ell$, $x_{1}=x$. 
Then we have
\begin{align*}
\FWint_{\ell}^{(x_{0})}(x)
=\; & \left(\qfact{\ell}\prod_{m=1}^{\ell}(q^{d-m}-q^{m-d})\right)
\times\SimplexInt_{\ell}^{(x_{0})}(x).
\end{align*}
In particular, $\FWint_{\ell}^{(x_{0})}$ is identically zero if $\ell\geq d$.
\end{lem}
The proof is a straightforward contour manipulation argument, whose details
are left to Appendix~\ref{app: contour manipulations}.

In the next lemma, we write the basis functions
$\FWint^{(x_{0})}$ in terms of the integrals $\CubeInt^{(x_{0})}$, in the case $n=2$.
The integrals $\CubeInt^{(x_{0})}$, in turn, could be written in terms of the real
integrals $\SimplexInt^{(x_{0})}$
by Lemma~\ref{lem: simplex in terms of hypercube}. We use this result
several times in the course of proving the main theorems:
as such in Lemma~\ref{lem: two point closed integration surface}, and in small
variations in Lemma~\ref{lem: asymptotics for hwv in two consequtive} and
Proposition~\ref{prop: closed integration surface}.
\begin{lem}\label{lem: two variable FW basis function}
Let $d_{1},d_{2}\in\bZpos$
and $0\leq l_{1}<d_{1}$, $0\leq l_{2}<d_{2}$. Then we have 
\begin{align*}
\FWint_{l_{1},l_{2}}^{(x_{0})}(x_{1},x_{2})
=\; & q^{\binom{l_{1}}{2}+\binom{l_{2}}{2}}(q-q^{-1})^{l_{1}+l_{2}}\frac{\qfact{d_{1}-1}\qfact{d_{2}-1}}{\qfact{d_{1}-l_{1}-1}\qfact{d_{2}-l_{2}-1}} \\
& \qquad \times \sum_{m=0}^{l_{2}}q^{m(m-l_{2}+d_{1}-1)}\qbin{l_{2}}m\CubeInt_{l_{1}+m,l_{2}-m}^{(x_{0})}(x_{1},x_{2}).
\end{align*}
\end{lem}
The proof is postponed to Appendix~\ref{app: contour manipulations}.

In the case of general $n$, the basis functions $\FWint^{(x_{0})}$
can still be written, in principle explicitly, in terms of the real
integrals $\SimplexInt^{(x_{0})}$. The proof of our main results,
however, do not rely on this explicit formula.
\begin{lem}\label{lem: general FW basis function}
We have
\begin{align*}
\FWint_{l_{1},\ldots,l_{n}}^{(x_{0})}(\boldsymbol{x})
=\; & \sum_{m_{1},\ldots,m_{n}}C_{l_{1},\ldots,l_{n}}^{m_{1},\ldots,m_{n}}\times\CubeInt_{m_{1},\ldots,m_{n}}^{(x_{0})}(\boldsymbol{x}),
\end{align*}
with some coefficients $C_{l_{1},\ldots,l_{n}}^{m_{1},\ldots,m_{n}}$,
which are zero unless $\sum_i l_{i}=\sum_i m_{i}$ and $\sum_{i=1}^{j}l_{i}\leq\sum_{i=1}^{j}m_{i}$
for all $j = 1 , \ldots, n$.
\end{lem}
The proof is sketched in Appendix~\ref{app: contour manipulations}.

The following important particular case already follows from the simple
calculations performed in the proof of 
Lemma~\ref{lem: explicit formula for one point FW integral}.
\begin{lem}\label{lem: vanishing for too many bags}
Whenever $l_{i}\geq d_{i}$
for some $i=1,\ldots,n$, we have
\begin{align*}
\FWint_{l_{1},\ldots,l_{n}}^{(x_{0})}(x_{1},\ldots,x_{n})\equiv\; & 0.
\end{align*}
\end{lem}
The short proof is given in Appendix~\ref{app: contour manipulations}.

\subsection{\label{sub: asymptotic integral functions}Mixed integral functions for asymptotics}

To extract the asymptotic behavior of our functions as $|x_{j+1}-x_{j}|\to0$,
we rewrite them in yet another way. We still keep $d_{1},\ldots,d_{n}$
fixed, and for any
\begin{align*}
 & l_{1},\ldots,l_{j-1},l_{j+2},\ldots,l_{n}\in\bZ_{\geq0},\qquad l\in\bZ_{\geq0},\qquad m\in\bZ_{\geq0}
\end{align*}
we define an integration surface
\begin{align*}
 & \SurfAsy_{l_{1},\ldots,l_{j-1};l,m;l_{j+2},\ldots,l_{n}},
\end{align*}
rephased integrand
\begin{align*}
 & \fAsy_{l_{1},\ldots,l_{j-1};l,m;l_{j+2},\ldots,l_{n}}(\boldsymbol{x};\boldsymbol{w}),
\end{align*}
and an integral function
\begin{align*}
 & \AsyInt_{l_{1},\ldots,l_{j-1};l,m;l_{j+2},\ldots,l_{n}}^{(x_{0})}(\boldsymbol{x})
\end{align*}
as follows. The surface $\SurfAsy$ is a mixture between a family
of non-intersecting loops (see Section~\ref{sub: basis functions as FW integrals})
and a deformed hypercube (see Section~\ref{sub: deformed hypercube integral functions})
as illustrated in Figure~\ref{fig: asymptotics integration}.
Trusting that the figure is sufficiently similar to the earlier ones,
we content ourselves to give the following slightly informal descriptions:
\begin{itemize}
\item There are $m$ variables integrated from $x_{j}$ to $x_{j+1}$, and
they turn around each other like in the deformed hypercube integrals
of Section~\ref{sub: deformed hypercube integral functions}.
\item For $i\notin\set{j,j+1}$ there are $l_{i}$ loops around $x_{i}$
starting and ending at the anchor point $x_{0}$, and they are nested
and turning around each other like in the earlier families of non-intersecting
loops of Section~\ref{sub: basis functions as FW integrals}.
\item There are $l$ loops around the entire paths of the deformed hypercube
integrals from $x_{j}$ to $x_{j+1}$, and they are nested and turning
around each other like in the families of non-intersecting loops,
as if the entire deformed hypercube would be a single point.
\item The rephasing and branch choice of the integrand is such that $\fAsy$
is positive at a point $\boldsymbol{w'}$ illustrated by the red circles
in Figure~\ref{fig: asymptotics integration}.
\end{itemize}
\noindent 
\begin{figure}
\includegraphics[width=1.05\textwidth]{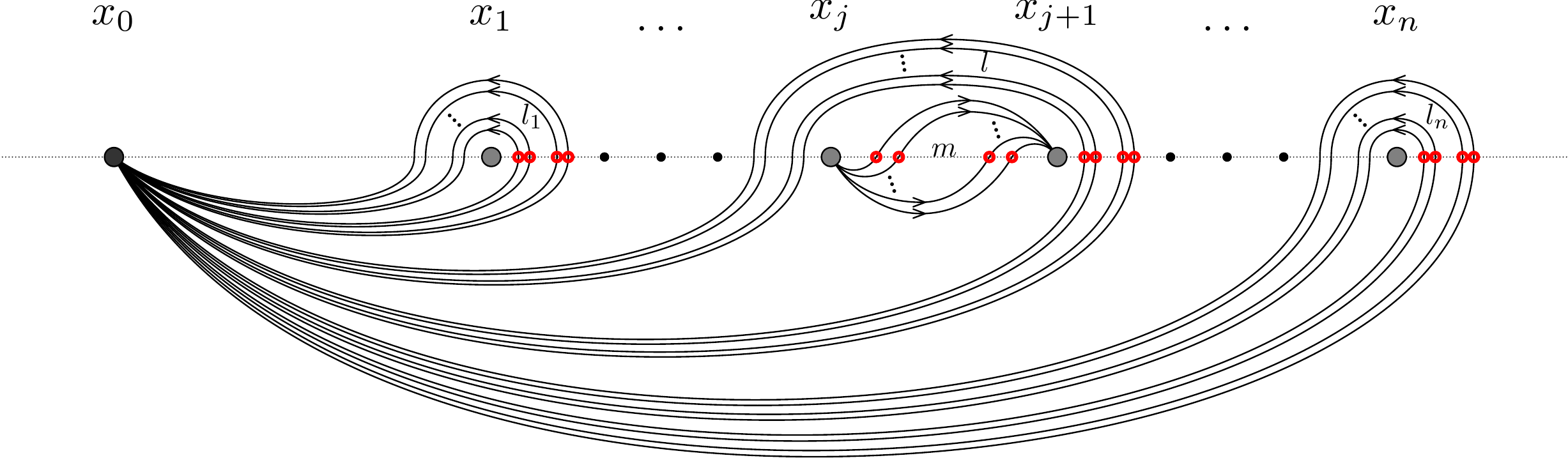}
\caption{\label{fig: asymptotics integration}
The integration surface $\SurfAsy_{l_{1},\ldots;l,m;\ldots,l_{n}}$
is used to construct functions that have an explicit asymptotic behavior
as $|x_{j+1}-x_{j}|\to0$. The point where the integrand $\fAsy{}_{l_{1},\ldots;l,m;\ldots,l_{n}}(\boldsymbol{x};\cdot)$
is rephased to be positive is marked by red circles.}
\end{figure}

The integral function $\AsyInt^{(x_{0})}$ is then defined as
\begin{align*}
\AsyInt_{l_{1},\ldots,l_{j-1};l,m;l_{j+2},\ldots,l_{n}}^{(x_{0})}(\boldsymbol{x})
:=\; & q^{\binom{m}{2}}\frac{1}{\qfact m}
\times\int_{\SurfAsy_{l_{1},\ldots,;l,m;,\ldots,l_{n}}}
\fAsy_{l_{1},\ldots,;l,m;,\ldots,l_{n}}(\boldsymbol{x};\boldsymbol{w})\,\ud\boldsymbol{w}.
\end{align*}
These, just like the integrals $\CubeInt$, are convergent for large
enough $\kappa$, namely $\kappa>4\times(\max(d_{j},d_{j+1})-1)$.
They can also be analytically continued in $\kappa$. The prefactor
is included to make the integrals more closely related to the integrals
over a real simplex --- compare with Lemma~\ref{lem: simplex in terms of hypercube}.

In order to state the results about the asymptotics, we use the exponents
\begin{align}\label{eq: asymptotics exponent}
\Delta_{d}^{d',d''}
:=\; & \frac{2(1+d^{2}-(d')^{2}-(d'')^{2})+\kappa(d'+d''-d-1)}{2\kappa},
\end{align}
and define the multiplicative constants $B_{d}^{d',d''}$ for $\kappa>4\times\big(\max\set{d',d''}-1\big)$
by the convergent integral over an $m$-dimensional simplex, with
$m=\frac{1}{2}\left(d'+d''-1-d\right)$,
\begin{align}\label{eq: generalized beta constants}
B_{d}^{d',d''} := \; & 
\int_{0}^{1}\ud w_{1}\int_{w_{1}}^{1}\ud w_{2}\cdots\int_{w_{m-1}}^{1}\ud w_{m}\;
\prod_{r=1}^{m}w_{r}^{-\frac{4}{\kappa}(d'-1)}
\prod_{r=1}^{m}(1-w_{r})^{-\frac{4}{\kappa}(d''-1)}
\prod_{1\leq r<s\leq m}(w_{s}-w_{r})^{\frac{8}{\kappa}}.
\end{align}
For general values of $\kappa$, the multiplicative constants $B_{d}^{d',d''}$
are defined by the analytic continuation of this generalized beta-function.

\begin{rem}\label{rem: Selberg integral}
\emph{
The multiplicative constant~\eqref{eq: generalized beta constants}
is trivial in the case $d=d'+d''-1$ (that is $m=0$): we then have $B_{d}^{d',d''} = 1$.}
\emph{In the case $d=d'+d''-3$ (that is $m=1$), the multiplicative
constants~\eqref{eq: generalized beta constants} are ordinary beta-functions.}
\emph{In general, \eqref{eq: generalized beta constants} is a Selberg integral, 
which can also
be evaluated in terms of gamma-functions
\cite{Selberg-integral,DF-four_point_correlation_functions,Forrester-log_gases_and_random_matrices}:
\begin{align*}
B_{d}^{d',d''} = \; & \frac{1}{m!} \; \prod_{p=1}^m
    \frac{\Gamma\big( 1 - \frac{4}{\kappa}(d'-p)\big) \; \Gamma\big( 1 - \frac{4}{\kappa}(d''-p)\big) \; \Gamma\big( 1 + \frac{4}{\kappa}p\big)}
        {\Gamma\big( 1 + \frac{4}{\kappa}\big) \; \Gamma\big( 2 - \frac{4}{\kappa}(d'+d''-m-p)\big)}.
\end{align*}
This expression gives explicitly the analytic continuation in $\kappa$, and it in particular
shows that poles and zeroes may only occur in $B_{d}^{d',d''}$ when
$\kappa$ tends to particular rational values.
}
\end{rem}

\begin{rem}
\emph{
In terms of the Kac labeled conformal weights~\eqref{eq: Kac labeled conformal 
weights},
one may express the overall homogeneity degree of
$\FWint_{l_{1},\ldots,l_{n}}^{(x_{0})}$,
given in Lemma~\ref{lem: translation and scaling of basis functions}, as
\begin{align*}
\Delta^{d_{1},\ldots,d_{n}}(\ell) = \Delta^{d_{1},\ldots,d_{n}}_d
=\; & h_{1,d}-\sum_{i=1}^{n}h_{1,d_{i}},
\end{align*}
where $d = \sum_{i=1}^{n}d_{i}-n+1-2\ell$ can be understood as one
of the dimensions of the irreducible subrepresentations of the tensor
product $\bigotimes_{i}\Wd_{d_{i}}$. We furthermore note that the
exponents~\eqref{eq: asymptotics exponent} are of this form, 
\begin{align*}
\Delta^{d_{1},d_{2}}(\ell)
=\; & \Delta_{d}^{d_{1},d_{2}}\qquad\text{ with }\qquad d=d_{1}+d_{2}-1-2\ell.
\end{align*}
}
\end{rem}

We now show an asymptotics property, which is particularly simple
for the functions $\AsyInt^{(x_{0})}$. This will be used later, in
Proposition~\ref{prop: asymptotics with subrepresentations}, to establish
the general asymptotics statement.
\begin{lem}\label{lem: asymptotic basis functions}
As $x_{j}$ and $x_{j+1}$
tend to a common limit $\xi$ (with $x_{j-1}<\xi<x_{j+2}$), we have
\begin{align*}
\frac{\AsyInt_{l_{1},\ldots,l_{j-1};l,m;l_{j+2},\ldots,l_{n}}^{(x_{0})}(x_{1},\ldots,x_{n})}{(x_{j+1}-x_{j})^{\Delta}} \;\longrightarrow\;\; & 
B\times\FWint_{l_{1},\ldots,l_{j-1},l,l_{j+2},\ldots,l_{n}}^{(x_{0})}(x_{1},\ldots,x_{j-1},\xi,x_{j+2},\ldots,x_{n}),
\end{align*}
where $d=d_{j}+d_{j+1}-1-2m$, the exponent $\Delta$ and the constant
$B$ are given by
\begin{align*}
 & \Delta=\Delta_{d}^{d_{j},d_{j+1}}\qquad\text{and}\qquad B=B_{d}^{d_{j},d_{j+1}},
\end{align*}
and $\FWint^{(x_{0})}$ on the right hand side is the basis function
of $n-1$ points associated to the choice of dimensions $d_{1},\ldots,d_{j-1},d,d_{j+2},\ldots,d_{n}$.
\end{lem}
\begin{proof}
As a warm-up, note that by a simple scaling 
(compare with Lemma~\ref{lem: translation and scaling of basis functions}),
the integrals $\CubeInt_{0,m}^{(x_{0})}(x_{1},x_{2})$ over deformed
hypercubes have the behavior
\begin{align*}
q^{\binom{m}{2}}\frac{1}{\qfact m}\times\CubeInt_{0,m}^{(x_{0})}(x_{1},x_{2})=\SimplexInt_{0,m}^{(x_{0})}(x_{1},x_{2})
=\; & (x_{2}-x_{1})^{\Delta_{d}^{d_{1},d_{2}}}\times B_{d}^{d_{1},d_{2}},
\end{align*}
for $\kappa$ large enough so that both sides are given by convergent
integrals.
In the general case, the same idea is used with the dominated convergence theorem
applied to a subset of the integration variables.

Consider the factors in the integrand
$\fAsy_{l_{1},\ldots,;l,m;,\ldots,l_{n}}(\boldsymbol{x};\boldsymbol{w})$
that involve $x_{j}$ or $x_{j+1}$, or any of the $m$ variables
integrated over the deformed hypercube in $\SurfAsy_{l_{1},\ldots,;l,m;,\ldots,l_{n}}$.
The integral over these $m$ variables, divided by $|x_{j+1}-x_j|^{\Delta}$,
tends to the integrand of the function
$\FWint_{l_{1},\ldots,l_{j-1},l,l_{j+2},\ldots,l_{n}}^{(x_{0})}(x_{1},\ldots,x_{j-1},\xi,x_{j+2},\ldots,x_{n})$,
times $B$.
Notice that the other integration contours remain bounded away
from these points, so dominated convergence theorem gives the asserted limit.

To conclude, note that both sides of the asserted formula are analytic in $\kappa$.
\end{proof}

\bigskip{}

\section{\label{sec: FW correspondence}The spin chain~- Coulomb gas correspondence}

In this section, we define the spin chain~- Coulomb gas correspondence,
and show how the representation theoretical properties are translated
to properties of the functions. 
As indicated in the introduction, the underlying idea in our construction
is Euler integral solutions to partial differential equations.
Informally, the effect of the quantum group is to act on the functions
of Section~\ref{sec: various integral functions} by modifying the
integration surfaces and branch choices of the integrand. 

The basic ingredients of Euler integral solutions
are the closedness of the integration surface, which is the content of
Sections~\ref{sub: general case of closed integration surface} and~\ref{sub: integration by parts formula},
and the exactness of the form where a differential operator acts on the
integrand, which is covered in Section~\ref{sub: differential equations}.
The asymptotics properties, which are crucial in applications for handling boundary
conditions, are analyzed in Section~\ref{sub: asymptotics with subrepresentations}.
M\"obius covariance is treated in Section~\ref{sub: Special conformal transformations}.
In summary, by a succession of small steps, which
combine formulas for the representations of the quantum group given
in Section~\ref{sec: q stuff} with properties of the functions established
in Section~\ref{sec: various integral functions}, in the end of this
section we will have proven the main result, whose precise formulation
is given in Theorems~\ref{thm: SCCG correspondence non-hwv} and 
\ref{thm: SCCG correspondence hwv} in Section~\ref{sub: the main correspondence result}.

From Section~\ref{sec: q stuff}, we use in particular the $d$-dimensional irreducible
representation $\Wd_{d}$ of the quantum group $\Uqsltwo$, defined
in Lemma~\ref{lem: representations of quantum sl2}. 
Its basis $(\Wbas_{l}^{(d)})_{l=0,1,\ldots,d-1}$,
introduced in the same lemma, is also used below. For simplicity of
notation, we often omit the superscript reference to the dimension
$d$.

\subsection{\label{sub: definition of the correspondence}Definition of the correspondence}

Fix $d_{1},\ldots,d_{n}\in\bZpos$. Denote by $\Wd_{d_{i}}$ the
$d_{i}$-dimensional irreducible representation of $\Uqsltwo$ defined
in Section~\ref{subsub: irreps of the quantum group}. Consider the
tensor product representation
\begin{align}\label{eq: order of tensorands}
\bigotimes_{i=1}^{n}\Wd_{d_{i}}
=\; & \Wd_{d_{n}}\tens\Wd_{d_{n-1}}\tens\cdots\tens\Wd_{d_{2}}\tens\Wd_{d_{1}}.
\end{align}
The order of tensorands will always be as shown on the right hand
side, but for brevity we usually use the notation on the left hand
side, and the above order is implicitly understood.

To construct the correspondence, we define mappings
\begin{align}\label{eq: the SCCG mapping}
\sF_{d_{1},\ldots,d_{n}}^{(x_{0})} \;\colon\; & 
\bigotimes_{i=1}^{n}\Wd_{d_{i}} \;\rightarrow\; \sC^{\infty}(\chamber_{n}^{(x_{0})})
\end{align}
from the tensor products~\eqref{eq: order of tensorands} to smooth
functions on the restricted chambers~\eqref{eq: restricted chamber}.
To simplify the notation, when the dimensions $d_{1},\ldots,d_{n}$
are clear from the context, we usually omit the subscripts and write
simply $\sF^{(x_{0})}=\sF_{d_{1},\ldots,d_{n}}^{(x_{0})}$. In the
representation $\Wd_{d_{i}}$, let $\Wbas_{l_{i}}$ denote the basis
vector obtained by applying the generator $F$ repeatedly $l_{i}$
times to the highest weight vector $\Wbas_{0}$, as in 
Section~\ref{subsub: irreps of the quantum group}.
The mapping~\eqref{eq: the SCCG mapping} is defined by setting the
images of the tensor product basis vectors to be the functions defined
in Section~\ref{sub: basis functions as FW integrals}, 
\begin{align}\label{eq: SCCG basis vector images}
\left(\sF^{(x_{0})}[\Wbas_{l_{n}}\tens\cdots\tens\Wbas_{l_{1}}]\right)(x_{1},\ldots,x_{n})
=\; & \FWint_{l_{1},\ldots,l_{n}}^{(x_{0})}(x_{1},\ldots,x_{n}),
\end{align}
and extending linearly.

\subsection{\label{sub: two point closed integration contour}Auxiliary formula for two points}

The following result will be used both for the proofs of asymptotics
in Section~\ref{sub: asymptotics with subrepresentations}, 
and for the anchor point independence, 
Proposition~\ref{prop: closed integration surface}.
In fact, this result is literally a special case of the latter.
\begin{lem}\label{lem: two point closed integration surface}
Let $\Tbas_{0}^{(d;d_{1},d_{2})} \in \Wd_{d_{2}}\tens\Wd_{d_{1}}$ be 
as in~\eqref{eq: tensor product hwv},
with $d=d_{1}+d_{2}-1-2m$. Then we have
\begin{align}\label{eq: explicit two point hwv function}
\left(\sF^{(x_{0})}[\Tbas_{0}^{(d;d_{1},d_{2})}]\right)(x_{1},x_{2})
=\; & q^{\binom{m}{2}}\frac{1}{\qfact m}\times\CubeInt_{0,m}^{(x_{0})}(x_{1},x_{2})
=\; \SimplexInt_{0,m}^{(x_{0})}(x_{1},x_{2}).
\end{align}
More generally, whenever $v\in\Wd_{d_{2}}\tens\Wd_{d_{1}}$
is such that $E.v=0$, then we have 
\begin{align*}
\left(\sF^{(x_{0})}[v]\right)(x_{1},x_{2})
=\; & \sum_{m}c_{m}\times\CubeInt_{0,m}^{(x_{0})}(x_{1},x_{2})
\end{align*}
for some coefficients $c_{m}\in\bC$ and, in particular, 
$\left(\sF^{(x_{0})}[v]\right)(x_{1},x_{2})$
is then independent of $x_{0}$.
\end{lem}
\begin{proof}
Equation~\eqref{eq: tensor product hwv} gives
\begin{align*}
v=\; & \Tbas_{0}^{(d;d_{1},d_{2})}
= \sum_{k=0}^{m}T_{0;m}^{k,m-k}\times(\Wbas_{m-k}\tens\Wbas_{k}),
\end{align*}
where 
$T_{0;m}^{k,m-k}$ are given by~\eqref{eq: tensor product hwv coefficients}.
Correspondingly, we have
\begin{align*}
\left(\sF^{(x_{0})}[\Tbas_{0}^{(d;d_{1},d_{2})}]\right)(x_{1},x_{2})
=\; & \sum_{k=0}^{m}T_{0;m}^{k,m-k}\times\FWint_{k,m-k}^{(x_{0})}(x_{1},x_{2}) .
\end{align*}
We use Lemma~\ref{lem: two variable FW basis function} to rewrite
$\FWint_{k,m-k}^{(x_{0})}$ as a linear combination of $\CubeInt_{k+t,m-k-t}^{(x_{0})}$,
for $t=0,1,\ldots,m-k$. After straightforward simplifications and
a change of the summation index $k$ to $u=k+t$, we obtain
\begin{align*}
\sF^{(x_{0})}[\Tbas_{0}^{(d;d_{1},d_{2})}
]=\; & \sum_{u=0}^{m}\CubeInt_{u,m-u}^{(x_{0})}\times
\Bigg((-1)^{u}q^{\binom{m}{2}+u(d_{1}-m)}\frac{1}{\qfact{m-u}}\sum_{t=0}^{u}(-1)^{-t}q^{t(u-1)}\frac{1}{\qfact t\qfact{u-t}}\Bigg).
\end{align*}
The above sum over $t$ simplifies by Lemma~\ref{lem: q-combinatorics}(d) to
\begin{align*}
\frac{1}{\qfact u}\sum_{t=0}^{u}(-1)^{t}q^{t(u-1)}\qbin ut
=\; & \frac{1}{\qfact u} \,q^{\half u(u-1)}\prod_{s=0}^{u-1}(q^{-s}-q^{s})
=\begin{cases}
1 & \text{ if }u=0\\
0 & \text{ if }u>0
\end{cases}.
\end{align*}
Thus finally, for $d=d_{1}+d_{2}-1-2m$, we obtain the expression
\begin{align*}
\left(\sF^{(x_{0})}[\Tbas_{0}^{(d;d_{1},d_{2})}]\right)(x_{1},x_{2})
=\; & q^{\binom{m}{2}}\frac{1}{\qfact m}\times\CubeInt_{0,m}^{(x_{0})}(x_{1},x_{2}).
\end{align*}
This and Lemma~\ref{lem: simplex in terms of hypercube} give the first asserted
formula~\eqref{eq: explicit two point hwv function}. 
By Lemma~\ref{lem: tensor product representations of quantum sl2}
and Remark~\ref{rem: tensor product different bases}, the vectors
$v\in\Wd_{d_{2}}\tens\Wd_{d_{1}}$ that satisfy $E.v=0$ are linear
combinations of the vectors $\Tbas_{0}^{(d;d_{1},d_{2})}$, so
the second statement follows.
\end{proof}

\subsection{\label{sub: asymptotics with subrepresentations}Asymptotics via projections to subrepresentations}
The detailed understanding of the asymptotics of functions is
fundamentally important in applications, since it pertains to the boundary conditions
of solutions. Our method describes the asymptotics conceptually with the underlying
quantum group in terms of projections to subrepresentations.

The asymptotics are easiest for the functions $\AsyInt_{l_{1},\ldots,l_{j-1};0,m;l_{j+2},\ldots,l_{n}}^{(x_{0})}$
of Section~\ref{sub: asymptotic integral functions}. Here it is therefore
desirable to write the function $\sF[v](\boldsymbol{x})$ in terms of
these functions. We begin with the following particular case.

\begin{lem}
\label{lem: asymptotics for hwv in two consequtive}
Let $\Tbas_{0}^{(d;d_{j},d_{j+1})}\in\Wd_{d_{j+1}}\tens\Wd_{d_{j}}$
be as in Lemma~\ref{lem: tensor product representations of quantum sl2},
with $d=d_{j}+d_{j+1}-1-2m$, and suppose that
\begin{align*}
v=\; & 
\Wbas_{l_{n}}\tens\cdots\tens\Wbas_{l_{j+2}}\tens\Tbas_{0}^{(d;d_{j},d_{j+1})}\tens\Wbas_{l_{j-1}}\tens\cdots\tens\Wbas_{l_{1}}.
\end{align*}
Then we have 
\begin{align*}
\sF^{(x_{0})}[v](\boldsymbol{x})
=\; & \AsyInt_{l_{1},\ldots,l_{j-1};0,m;l_{j+2},\ldots,l_{n}}^{(x_{0})}(\boldsymbol{x}).
\end{align*}
\end{lem}
\begin{proof}
Notice first that for $n=2$ we have $\AsyInt_{;0,m;}^{(x_{0})}(x_{1},x_{2})=q^{\binom{m}{2}}\frac{1}{\qfact m}\times\CubeInt_{0,m}^{(x_{0})}(x_{1},x_{2})$,
so the assertion is the same as 
Equation~\eqref{eq: explicit two point hwv function}.
The general case is similar. Indeed, the vector $v\in\bigotimes_{i=1}^{n}\Wd_{d_{i}}$
can be expanded in the standard tensor product basis as 
\begin{align*}
v=\; & \sum_{k=0}^{m}T_{0;m}^{k,m-k}(d_{j},d_{j+1})\times
\big(\Wbas_{l_{n}}^{(d_{n})}\tens\cdots\tens\Wbas_{l_{j+2}}^{(d_{j+2})}\tens\Wbas_{m-k}^{(d_{j+1})}\tens\Wbas_{k}^{(d_{j})}\tens\Wbas_{l_{j-1}}^{(d_{j-1})}\tens\cdots\tens\Wbas_{l_{1}}^{(d_{1})}\big),
\end{align*}
according to Lemma~\ref{lem: tensor product representations of quantum sl2},
Equations~\eqref{eq: tensor product hwv} and \eqref{eq: tensor product hwv coefficients}.
Correspondingly, we obtain
\begin{align*}
\sF^{(x_{0})}[v](\boldsymbol{x})
=\; & \sum_{k=0}^{m}T_{0;m}^{k,m-k}(d_{j},d_{j+1})\times
\FWint_{l_{1},\ldots,l_{j-1},k,m-k,l_{j+2},\ldots,l_{n}}^{(x_{0})}(\boldsymbol{x}).
\end{align*}
The terms in this sum are all integrals of the same integrand (up
to phase factors). As for the integration contours, the number of
loops around any $x_{i}$ is fixed, with the exceptions of $i=j$
and $i=j+1$. The loops around $x_{j}$ and $x_{j+1}$ can be combined
by essentially the same calculations that were done in the proofs of
Lemmas~\ref{lem: two variable FW basis function}
and \ref{lem: two point closed integration surface}, with the asserted
result.
\end{proof}
With the above particular case established, we now proceed to the
following generalization.
\begin{lem}\label{lem: asymptotics for tensor basis in two consequtive}
Denote $\tau_\ell = F^\ell.\tau_0$. If
\begin{align*}
v=\; & 
\Wbas_{l_{n}}\tens\cdots\tens\Wbas_{l_{j+2}}\tens\Tbas_{l}^{(d;d_{j},d_{j+1})}\tens\Wbas_{l_{j-1}}\tens\cdots\tens\Wbas_{l_{1}},
\end{align*}
then (with $d=d_{j}+d_{j+1}-1-2m$)
\begin{align*}
\sF^{(x_{0})}[v](\boldsymbol{x})
=\; & \AsyInt_{l_{1},\ldots,l_{j-1};l,m;l_{j+2},\ldots,l_{n}}^{(x_{0})}(\boldsymbol{x}).
\end{align*}
\end{lem}
\begin{proof}
By definition of the correspondence, we have 
\begin{align*}
\sF^{(x_{0})}[v](\boldsymbol{x})
=\; & \sum_{l_{j},l_{j+1}}T_{l;m}^{l_{j},l_{j+1}}(d_{j},d_{j+1})\times
\FWint_{l_{1},\ldots,l_{n}}^{(x_{0})}(\boldsymbol{x}),
\end{align*}
with the coefficients $T_{l;m}^{l_{j},l_{j+1}}(d_{j},d_{j+1})$ given
by Equation~\eqref{eq: tensor product submodule basis} for $\Tbas_{l}^{(d;d_{j},d_{j+1})}=F^{l}.\Tbas_{0}^{(d;d_{j},d_{j+1})}$.
Therefore, our goal is to rewrite, for any $l\geq0$, 
\begin{align}\label{eq: asymptotic function in terms of basis functions}
\AsyInt_{l_{1},\ldots,l_{j-1};l,m;l_{j+2},\ldots,l_{n}}^{(x_{0})}(\boldsymbol{x})
=\; & \sum_{l_{j},l_{j+1}}T_{l;m}^{l_{j},l_{j+1}}(d_{j},d_{j+1})\times\FWint_{l_{1},\ldots,l_{n}}^{(x_{0})}(\boldsymbol{x}),
\end{align}
with the same coefficients. To achieve this, we will proceed by recursion
on $l$. The base case $l=0$ was the content of 
Lemma~\ref{lem: asymptotics for hwv in two consequtive}.
Using the formula~\eqref{eq: coproduct} for the coproduct $\Hcp(F)$,
we see that the coefficients~\eqref{eq: tensor product submodule basis}
satisfy the recursion
\begin{align*}
T_{l;m}^{l_{j},l_{j+1}}
=\; & q^{1-d_{j+1}+2l_{j+1}}\,T_{l-1;m}^{l_{j}-1,l_{j+1}}
+T_{l-1;m}^{l_{j},l_{j+1}-1}.
\end{align*}
Hence, it suffices to show the same recursion for the coefficients
appearing in Equation~\eqref{eq: asymptotic function in terms of basis functions}.
For that, note that the difference between the integrations defining
$\AsyInt_{\ldots;l-1,m;\ldots}^{(x_{0})}(\boldsymbol{x})$ and $\AsyInt_{\ldots;l,m;\ldots}^{(x_{0})}(\boldsymbol{x})$
is that the latter has one extra integration variable, integrated
along a loop that surrounds both $x_{j}$ and $x_{j+1}$ and their
associated $w$-variables. 
Assume~\eqref{eq: asymptotic function in terms of basis functions}
for the former, and decompose the extra loop of the latter to two
pieces, loops anchored at $x_{0}$ that surround $x_{j}$ and $x_{j+1}$
separately. The desired recursion follows by comparing the phase factors
after the decomposition.
\end{proof}
Fix an index $j\in\set{1,\ldots,n-1}$. Recall 
Lemma~\ref{lem: tensor product representations of quantum sl2},
especially the decomposition~\eqref{eq: decomposition of tensor product}
of a tensor product to subrepresentations $\Wd_{d_{j+1}}\tens\Wd_{d_{j}}\isom\bigoplus_{d}\Wd_{d}$.
For any $d$ appearing in this sum, the linear map
\begin{align}
& \iota_{j,j+1}^{(d)} \;\colon\;
\Big(\bigotimes_{i=j+2}^{n}\Wd_{d_{i}}\Big)\tens\Wd_{d}\tens\Big(\bigotimes_{i=1}^{j-1}\Wd_{d_{i}}\Big)
    \;\to\; \bigotimes_{i=1}^{n}\Wd_{d_{i}}
    \label{eq: embedding of a subrepresentation}\\
\nonumber 
& \iota_{j,j+1}^{(d)} \left(\Wbas_{l_{n}}\tens\cdots\tens\Wbas_{l_{j+2}}\tens\Wbas_{l}\tens\Wbas_{l_{j-1}}\tens\cdots\tens\Wbas_{l_{1}}\right) \\
\nonumber 
& \qquad \qquad \qquad \qquad \qquad \qquad = \Wbas_{l_{n}}\tens\cdots\tens\Wbas_{l_{j+2}}\tens\Tbas_{l}^{(d;d_{j},d_{j+1})}\tens\Wbas_{l_{j-1}}\tens\cdots\tens\Wbas_{l_{1}}
\end{align}
is an embedding that respects the action of $\Uqsltwo$. Hence we
may interpret the shorter tensor product as a subrepresentation,
\begin{align}\label{eq: subrepresentation identification}
\Big(\bigotimes_{i=j+2}^{n}\Wd_{d_{i}}\Big)\tens\Wd_{d}\tens\Big(\bigotimes_{i=1}^{j-1}\Wd_{d_{i}}\Big)\,\subset\,\; & \bigotimes_{i=1}^{n}\Wd_{d_{i}}.
\end{align}
We denote the projection to this subrepresentation by
\begin{align*}
\pi_{j,j+1}^{(d)} \;\colon\; & \bigotimes_{i=1}^{n}\Wd_{d_{i}}\to\bigotimes_{i=1}^{n}\Wd_{d_{i}}
\end{align*}
and we denote by
\begin{align*}
\hat{\pi}_{j,j+1}^{(d)} \;\colon\; & \bigotimes_{i=1}^{n}\Wd_{d_{i}} \;\to\; 
\Big(\bigotimes_{i=j+2}^{n}\Wd_{d_{i}}\Big)\tens\Wd_{d}\tens\Big(\bigotimes_{i=1}^{j-1}\Wd_{d_{i}}\Big)
\end{align*}
the projection combined with 
the identification~\eqref{eq: subrepresentation identification},
so that $\pi_{j,j+1}^{(d)}=\iota_{j,j+1}^{(d)}\circ\hat{\pi}_{j,j+1}^{(d)}$.
A vector $v\in\bigotimes_{i=1}^{n}\Wd_{d_{i}}$ lies in this subrepresentation
if and only if $\pi_{j,j+1}^{(d)}(v)=v$, and in this case we typically
denote $\hat{v}=\hat{\pi}_{j,j+1}^{(d)}(v)$.

We are now ready to write down the asymptotics of the functions as
two consecutive variables tend to a common limit. By the following
proposition, the asymptotics are determined by the above decompositions
to subrepresentations.
\begin{prop}\label{prop: asymptotics with subrepresentations}
If $v\in\Wd_{d_{n}}\tens\cdots\tens\Wd_{d_{1}}$
satisfies $\pi_{j,j+1}^{(d)}(v)=v$, and we denote
\begin{align*}
\hat{v}=\; & 
\hat{\pi}_{j,j+1}^{(d)}(v) \,\in\,
\Big(\bigotimes_{i=j+2}^{n}\Wd_{d_{i}}\Big)\tens\Wd_{d}\tens\Big(\bigotimes_{i=1}^{j-1}\Wd_{d_{i}}\Big),
\end{align*}
then we have the asymptotics
\begin{align*}
\; \lim_{x_{j},x_{j+1}\to\xi} 
\Big( & (x_{j+1}-x_{j})^{-\Delta_{d}^{d_{j},d_{j+1}}}\times
\sF_{d_{1},\ldots,d_{n}}^{(x_{0})}[v](x_{1},\ldots,x_{n})\Big)\\
=\; & B_{d}^{d_{j},d_{j+1}}\times\sF_{d_{1},\ldots,d_{j-1},d,d_{j+2},\ldots,d_{n}}^{(x_{0})}[\hat{v}](x_{1},\ldots,x_{j-1},\xi,x_{j+2},\ldots,x_{n}),
\end{align*}
where the exponent $\Delta_{d}^{d_{j},d_{j+1}}$ is given 
by~\eqref{eq: asymptotics exponent}
and the multiplicative constant $B_{d}^{d_{j},d_{j+1}}$ 
by~\eqref{eq: generalized beta constants}.
\end{prop}
\begin{proof}
The vector $v$ can be expressed as a linear combination of  $\Wbas_{l_{n}}\tens\cdots\tens\Wbas_{l_{j+2}}\tens\Tbas_{l}^{(d;d_{j},d_{j+1})}\tens\Wbas_{l_{j-1}}\tens\cdots\tens\Wbas_{l_{1}}$.
For these vectors, the assertion follows by combining 
Lemmas~\ref{lem: asymptotics for tensor basis in two consequtive}
and \ref{lem: asymptotic basis functions}.
\end{proof}

\subsection{\label{sub: general case of closed integration surface}Anchor point independence}

Next we show that for highest weight vectors, the corresponding functions
become well-defined on the chamber $\chamber_{n}$ of~\eqref{eq: chamber}.
This is the simplest manifestation of the closedness of the corresponding integration surface.
\begin{prop}\label{prop: closed integration surface}
If $v\in\Wd_{d_{n}}\tens\cdots\tens\Wd_{d_{1}}$ is such that $E.v=0$, then we have 
\begin{align*}
\left(\sF^{(x_{0})}[v]\right)(\boldsymbol{x})
=\; & \sum_{m_{2},m_{3},\ldots,m_{n}}c_{m_{2},\ldots,m_{n}}\times
\CubeInt_{0,m_{2},\ldots,m_{n}}^{(x_{0})}(\boldsymbol{x})
\end{align*}
for some coefficients $c_{m_{2},\ldots,m_{n}}\in\bC$. In particular
$\left(\sF^{(x_{0})}[v]\right)(\boldsymbol{x})$ is independent of $x_{0}$.
\end{prop}
\begin{proof}
As a warm up, observe that in the case $n=1$ the statement is immediate:
by Lemma~\ref{lem: explicit formula for one point FW integral} the
basis vectors $\Wbas_{l_{1}}\in\Wd_{d_{1}}$ are mapped to $\sF^{(x_{0})}[\Wbas_{l_{1}}]\propto\CubeInt_{l_{1}}^{(x_{0})}$,
and we have $E.\Wbas_{l_{1}}=0$ only if $l_{1}=0$. For the case
$n=2$ the statement was already shown in Lemma~\ref{lem: two point closed integration surface}.

We proceed by induction on $n$. Let $n>2$ and write $v$ in the
basis $\Wbas_{l_{n}}\tens\cdots\tens\Wbas_{l_{3}}\tens\Tbas_{l}^{(d)}$,
with $l,l_{3},\ldots,l_{n}$ non-negative integers and $d=d_{1}+d_{2}-2m$
for some $m$ as in Lemma~\ref{lem: tensor product representations of quantum sl2},
\begin{align*}
v=\; & 
\sum_{l,d,l_{3},\ldots,l_{n}}b_{l,d;l_{3},\ldots,l_{n}}\times
\Wbas_{l_{n}}\tens\cdots\tens\Wbas_{l_{3}}\tens\Tbas_{l}^{(d)}.
\end{align*}
Separate the parts corresponding to fixed values of $d$, and denote
(no summation over $d$ here)
\begin{align*}
v^{(d)}=\; & 
\sum_{l,l_{3},\ldots,l_{n}}b_{l,d;l_{3},\ldots,l_{n}}
\times\Wbas_{l_{n}}\tens\cdots\tens\Wbas_{l_{3}}\tens\Tbas_{l}^{(d)}
\;\in\;\bigotimes_{i=1}^{n}\Wd_{d_{i}}\qquad\text{and}\\
\pi_{1,2}^{(d)}(v^{(d)}) = \hat{v}^{(d)}
=\; & \sum_{l,l_{3},\ldots,l_{n}}b_{l,d;l_{3},\ldots,l_{n}}\times\Wbas_{l_{n}}\tens\cdots\tens\Wbas_{l_{3}}\tens\Wbas_{l}\;\in\;\Big(\bigotimes_{i=3}^{n}\Wd_{d_{i}}\Big)\tens\Wd_{d}.
\end{align*}
Now $\hat{v}^{(d)}$ also satisfies $E.\hat{v}^{(d)}=0$, and the
tensor product only has $n-1$ tensorands. By the induction hypothesis,
\begin{align*}
\left(\sF^{(x_{0})}[\hat{v}^{(d)}]\right)(\xi,x_{3},\ldots,x_{n})
=\; & \sum_{m_{3},\ldots,m_{n}}\hat{c}_{m_{3},\ldots,m_{n}}^{(d)}\times
\CubeInt_{0,m_{3},\ldots,m_{n}}^{(x_{0})}(\xi,x_{3},\ldots,x_{n}),
\end{align*}
where the implicit dimension parameters are $d,d_{3},d_{4},\ldots,d_{n}$
now. In view of Lemma~\ref{lem: asymptotics for tensor basis in two consequtive}
and the rearrangement procedure of integration contours
similar to the proofs of Lemmas~\ref{lem: general FW basis function} and
\ref{lem: two variable FW basis function}, one sees that also
\begin{align*}
\left(\sF^{(x_{0})}[v^{(d)}]\right)(x_{1},x_{2},x_{3},\ldots,x_{n})
=\; & \sum_{m_{3},\ldots,m_{n}}c_{m_{3},\ldots,m_{n}}^{(d)}\times
\CubeInt_{0,m_{2},\ldots,m_{n}}^{(x_{0})}(x_{1},x_{2},x_{3},\ldots,x_{n}).
\end{align*}
The conclusion is finally obtained by summing over $d$, since $v=\sum_{d}v^{(d)}$.
\end{proof}

\begin{rem}\label{rem: anchor point independence}
\emph{
When $\kappa$ is large
enough, \eqref{eq: large enough kappa}, the integrals $\CubeInt_{m_{1},\ldots,m_{n}}^{(x_{0})}$
are convergent. When $E.v=0$, the above proposition then shows
that $\sF^{(x_{0})}[v](x_{1},\ldots,x_{n})$ is independent of
$x_{0}$. By analyticity in $\kappa$, this independence of $x_{0}$
holds for all $\kappa$, and we get a well-defined function $\sF[v]\colon\chamber_{n}\to\bC$ by
\begin{align*}
\sF[v](x_{1},\ldots,x_{n})
=\; & \sF^{(x_{0})}[v](x_{1},\ldots,x_{n})\qquad\text{ for any }\qquad x_{0}<x_{1}.
\end{align*}
}
\end{rem}

\subsection{\label{sub: integration by parts formula}Integration by parts formula}

According to the definition given in 
Section~\ref{sub: definition of the correspondence},
the vectors $v\in\Wd_{d_{n}}\tens\cdots\tens\Wd_{d_{1}}$ determine
linear combinations of integration surfaces. With a suitable interpretation,
there is a homology theory for these, in which the boundary operator
corresponds to the action of the quantum group generator $E$, see
\cite{FW-topological_representation_of_Uqsl2}. We will only make
use of a version of Stokes formula, i.e., integration by parts, which
we state and prove next.
\begin{lem}\label{lem: general integration by parts}
Let $l_{1},\ldots,l_{n}\in\bZ_{\geq0}$
and $\ell=\sum_{j}l_{j}$. Suppose that $g(w_{\ell};\boldsymbol{x};w_{1},\ldots,w_{\ell-1})$
is a holomorphic function of the $\ell$ variables $w_{1},\ldots,w_{\ell}$
defined on $\Wchamber^{(\ell)}$, which is symmetric in the $\ell-1$
variables $w_{1},\ldots,w_{\ell-1}$. Then we have
\begin{align}
 & \int_{\SurfFW_{l_{1},\ldots,l_{n}}}\sum_{r=1}^{\ell}\pder{w_{r}}\Big(g(w_{r};\boldsymbol{x};w_{1},\ldots,w_{r-1},w_{r+1},\ldots,w_{\ell})\, 
 f_{l_{1},\ldots,l_{n}}^{\Supset}(\boldsymbol{x};\boldsymbol{w})\Big)
 \,\ud w_{1}\cdots\ud w_{\ell}
 \label{eq: formula to be integrated by parts}\\
=\; & \sum_{j=1}^{n}\Bigg\{(q^{-1}-q)\qnum{l_{j}}\qnum{d_{j}-l_{j}}q^{\sum_{i<j}(d_{i}-1-2l_{i})}\nonumber \\
 & \qquad\quad\times\,\int_{\SurfFW_{l_{1},\ldots,l_{j}-1,\ldots,l_{n}}}\Big(\gamma(\boldsymbol{x};w_{1},\ldots,w_{\ell-1})\, f_{l_{1},\ldots,l_{j}-1,\ldots,l_{n}}^{\Supset}(\boldsymbol{x};w_{1},\ldots,w_{\ell-1})\Big)
 \,\ud w_{1}\cdots\ud w_{\ell-1}\Bigg\},\nonumber 
\end{align}
where
\begin{align*}
\gamma(\boldsymbol{x};w_{1},\ldots,w_{\ell-1})
=\; & \prod_{i=1}^{n}|x_{0}-x_{i}|^{-\frac{4}{\kappa}(d_{i}-1)}\prod_{s\neq r}|x_{0}-w_{r}|^{\frac{8}{\kappa}}\; g(x_{0};\boldsymbol{x};w_{1},\ldots,w_{\ell-1}).
\end{align*}
\end{lem}
\begin{proof}
Let us first perform an integration by parts in a single term 
in~\eqref{eq: formula to be integrated by parts}.
Fix $r$, and let $x_{j}$ be the point encircled by the loop of $w_{r}$
in $\SurfFW_{l_{1},\ldots,l_{n}}$, and denote by $u$ the number
of $s<r$ such that $w_{s}$ also encircles the same point $x_{j}$,
that is, $u=r-1-\sum_{i<j}l_{i}$. We then perform integration by
parts in the integral over $w_{r}$, and notice that the boundary
terms from the beginning and  end points of the loop only differ
by a phase (formally $w_{r}=x_{0}$ both at the beginning and the
end, but on different sheets of a Riemann surface).
After this integration by parts, the contribution of the $r$:th term becomes
\begin{align*}
 & q^{\sum_{i<j}(d_{i}-1-2l_{i})}(q^{1-d_{j}+2u}-q^{-1+d_{j}-2u})\\
 & \qquad\qquad\times\int_{\SurfFW_{l_{1},\ldots,l_{j}-1,\ldots,l_{n}}}\Big(\gamma(\boldsymbol{x};w_{1},\ldots,w_{\ell-1})\, f_{l_{1},\ldots,l_{j}-1,\ldots,l_{n}}^{\Supset}(\boldsymbol{x};w_{1},\ldots,w_{\ell-1})\Big)
 \,\ud w_{1}\cdots\ud w_{\ell-1},
\end{align*}
where we relabeled the other integration variables and used the assumption
of symmetric dependence of $g$ on them (and a similar property of
$f$). We collect the terms corresponding to the same $j$, and use
Lemma~\ref{lem: q-combinatorics}(a) in the form
\begin{align*}
\sum_{u=0}^{l_{j}-1}(q^{1-d_{j}+2u}-q^{-1+d_{j}-2u})
=\; & (q^{-1}-q)\,\qnum{l_{j}}\qnum{d_{j}-l_{j}}
\end{align*}
to simplify the sum of these terms. This concludes the proof.
\end{proof}

Besides the anchor point independence of Proposition~\ref{prop: closed integration surface},
the closedness of the integration surface for highest weight vectors
is used in combination with Stokes' formula via the following corollary.
\begin{cor}\label{cor: integration by parts for hwv}
Let $\ell\in\bZ_{\geq0}$.
Suppose that $g(w_{\ell};\boldsymbol{x};w_{1},\ldots,w_{\ell-1})$ is
a holomorphic function of the $\ell$ variables $w_{1},\ldots,w_{\ell}$
defined on $\Wchamber^{(\ell)}$, which is symmetric in the $\ell-1$
variables $w_{1},\ldots,w_{\ell-1}$. If
\begin{align*}
v= & \sum_{\substack{l_{1},\ldots,l_{n}\geq0\\
l_{1}+\cdots+l_{n}=\ell}}t_{l_{1},\ldots,l_{n}}\,\Wbas_{l_{n}}\tens\cdots\tens\Wbas_{l_{1}}
\end{align*}
satisfies $E.v=0$, then
\begin{align*}
\sum_{l_{1},\ldots,l_{n}}t_{l_{1},\ldots,l_{n}}\int_{\SurfFW_{l_{1},\ldots,l_{n}}}\sum_{r=1}^{\ell}\pder{w_{r}}\Big(g(w_{r};\boldsymbol{x};w_{1},\ldots,w_{r-1},w_{r+1},\ldots,w_{\ell})\, f_{l_{1},\ldots,l_{n}}^{\Supset}(\boldsymbol{x};\boldsymbol{w})\Big)\,\ud w_{1}\cdots\ud w_{\ell}=\; & 0.
\end{align*}
 \end{cor}
\begin{proof}
On the left hand side of the asserted formula, we use the previous
lemma in each term, and write the left hand side as
\begin{align}\label{eq: formula integrated by parts}
 & (q^{-1}-q)\,\sum_{j=1}^n\sum_{l_{1},\ldots,l_{n}}\qnum{l_{j}}\qnum{d_{j}-l_{j}}q^{\sum_{i<j}(d_{i}-1-2l_{i})}
 \,t_{l_{1},\ldots,l_{n}}
 \times\left(\int_{\SurfFW_{l_{1},\ldots,l_{j}-1,\ldots,l_{n}}}\gamma\, f_{l_{1},\ldots,l_{j}-1,\ldots,l_{n}}^{\Supset}\right).
\end{align}
On the other hand, using Lemma~\ref{lem: multiple coproducts} for the iterated coproduct
$\Hcp^{(n)}(E)$, we can expand $E.v$ in the tensor product basis as
\begin{align*}
\sum_{j=1}^n \sum_{l_{1},\ldots,l_{n}} \qnum{l_{j}}\qnum{d_{j}-l_{j}}q^{\sum_{i<j}(d_{i}-1-2l_{i})} \,t_{l_{1},\ldots,l_{n}}
    \times \Wbas_{l_n} \tens \cdots \tens \Wbas_{l_j -1} \tens \cdots \tens \Wbas_{l_1} .
\end{align*}
Apart from the factor $(q^{-1}-q)$,
the coefficient of $\Wbas_{l_n} \tens \cdots \tens \Wbas_{l_j -1} \tens \cdots \tens \Wbas_{l_1}$
is the same as the coefficient of the integral
over $\SurfFW_{l_{1},\ldots,l_{j}-1,\ldots,l_{n}}$ in~\eqref{eq: formula integrated by parts}.
Therefore, by the assumption $E.v=0$,
the expression~\eqref{eq: formula integrated by parts} vanishes.
\end{proof}

\subsection{\label{sub: differential equations}Partial differential equations}

Now we turn to partial differential equations satisfied by the functions
associated to highest weight vectors.

Define, for any $j=1,\ldots,n$ and $p\in\bZ$, the first order partial
differential operators
\begin{align*}
\sL_{-p}^{(j)} 
:= \; & -\sum_{i\neq j}\left((x_{i}-x_{j})^{1-p}\pder{x_{i}}+(1-p)\, h_{1,d_{i}}\,(x_{i}-x_{j})^{-p}\right),
\end{align*}
where $h_{1,d}$ is given by Equation~\eqref{eq: Kac labeled conformal weights}.
Define also, for any $j=1,\ldots,n$, a partial differential operator
of order $d_{j}$ by the Benoit~\& Saint-Aubin formula \cite{BSA-degenerate_CFTs_and_explicit_expressions}
\begin{align*}
\sD_{d_{j}}^{(j)} 
:= \; & \sum_{k=1}^{d_{j}}\sum_{\substack{p_{1},\ldots,p_{k}\geq1\\
p_{1}+\cdots+p_{k}=d_{j}}}
\frac{(-4/\kappa)^{d_{j}-k}\,(d_{j}-1)!^{2}}{\prod_{u=1}^{k-1}(\sum_{i=1}^{u}p_{i})(\sum_{i=u+1}^{k}p_{i})}\times\sL_{-p_{1}}^{(j)}\cdots\sL_{-p_{k}}^{(j)}.
\end{align*}

The following special case without any integrations will be used to verify
the PDEs in the general case.
\begin{lem}\label{lem: BSA diff op for vertex operator corr fn}
Let $d_1, \ldots, d_n$ be real numbers. The function
\begin{align*}
f^{(0)}(x_{1},\ldots,x_{n})
=\; & \prod_{1\leq i<j\leq n}(x_{j}-x_{i})^{\frac{2}{\kappa}(d_{i}-1)(d_{j}-1)}
\end{align*}
satisfies the partial differential equation
$\sD_{d_{j}}^{(j)}f^{(0)}  =  0$, for all $j=1,\ldots,n$ such that $d_j$ is a 
positive integer.
\end{lem}
We postpone the proof to Appendix~\ref{app: exact form lemma}.

\begin{rem}
\emph{%
The function $f^{(0)}$ is a very simple product of powers of differences,
but it is nevertheless not entirely trivial to verify 
that the Benoit~\& Saint-Aubin differential operators $\sD^{(j)}_{d_j}$ annihilate it in general.
The proof that we present in Appendix~\ref{app: exact form lemma}
requires only performing the explicit calculation in the case of $d_j=2$,
and then using recursively a fusion argument of the type considered in~\cite[Theorem~15]{Dubedat-fusion}.
There is one stage in this argument, \cite[Lemma~1]{Dubedat-fusion},
which appeals to the structure of Verma modules for the Virasoro algebra, but
only in a simple case when $\kappa \notin \bQ$.
}

\emph{%
As an alternative proof for
Lemma~\ref{lem: BSA diff op for vertex operator corr fn},
which is shorter but appeals to more involved constructions,
one could use properties of vertex operators and the BRST charge, 
see e.g.~\cite{Felder-BRST_approach}.
}
\end{rem}

The interpretation of the following corollary is that the
Benoit~\& Saint-Aubin differential operators acting on our integrand produce exact forms. 
\begin{cor}\label{cor: the BSA operator gives total derivative}
The function $f^{(\ell)}$ satisfies, for any $j=1,\ldots,n$,
\begin{align*}
\Big(\sD_{d_{j}}^{(j)}f^{(\ell)}\Big)(\boldsymbol{x};\boldsymbol{w})
=\; & \sum_{r=1}^{\ell}\pder{w_{r}}\left(g(w_{r};\boldsymbol{x};w_{1},\ldots,w_{r-1},w_{r+1},\ldots,w_{\ell})
\times f^{(\ell)}(\boldsymbol{x};\boldsymbol{w})\right),
\end{align*}
where $g$ is a rational function which is symmetric in the last $\ell-1$
variables, and whose only poles are where some of its arguments coincide.
\end{cor}
\begin{proof}
Denote $\hat{n}=n+\ell$ and 
\begin{align*}
\hat{\boldsymbol{x}} =\; & (x_{1},\ldots,x_{n},w_{1},\ldots,w_{\ell})
\end{align*}
and define the function, a product of powers of differences of variables,
\begin{align*}
\hat{f}^{(0)}(\hat{\boldsymbol{x}}) =\; & f^{(\ell)}(\boldsymbol{x};\boldsymbol{w}).
\end{align*}
The conformal weights $h$ associated to the variables $w_{r}$ are
$h=h_{1,-1}=1$. We will apply 
Lemma~\ref{lem: BSA diff op for vertex operator corr fn}
to the function $\hat{f}^{(0)}$ of $\hat{n}$ variables. We keep
the notation $\sL_{p}^{(j)}$ and $\sD_{d_{j}}^{(j)}$ for the differential
operators in the $n$ variables $x_{1},\ldots,x_{n}$, and use the
notation $\hat{\sL}_{p}^{(j)}$ and $\hat{\sD}_{d_{j}}^{(j)}$ for
the differential operators in $\hat{n}=n+\ell$ variables, that are
appropriate for the application of the previous lemma. Explicitly,
we have
\begin{align*}
\hat{\sL}_{p}^{(j)} =\; & \sL_{p}^{(j)}-\sum_{r=1}^{\ell}D_{w_{r}}M_{(w_{r}-x_{j})^{1+p}},
\end{align*}
where $D_{w_{r}}$ is the differential operator $\pder{w_{r}}$, and
$M_{(w_{r}-x_{j})^{1+p}}$ is the multiplication operator by the function
$(w_{r}-x_{j})^{1+p}$, and we used the fact that, since $h_{1,-1}=1$,
\begin{align*}
D_{w_{r}}M_{(w_{r}-x_{j})^{1+p}}
=\; & (w_{r}-x_{j})^{1+p}\pder{w_{r}}+(1+p)h_{1,-1}(w_{r}-x_{j})^{p}.
\end{align*}
The operator $\hat{\sD}_{d_{j}}^{(j)}$ is a linear combination of
terms, which we expand by the binomial formula
\begin{align*}
\hat{\sL}_{-n_{1}}^{(j)}\cdots\hat{\sL}_{-n_{k}}^{(j)}
=\; & \sum_{A\subset\set{1,\ldots,k}}O_{1}^{A}\cdots O_{k}^{A},\qquad\qquad
\text{where } \quad 
O_{a}^{A} = \begin{cases}
\sL_{-n_{a}}^{(j)} & \text{if }a\notin A\\
\sum_{r}D_{w_{r}}M_{(w_{r}-x_{j})^{1-n_{a}}} & \text{if }a\in A
\end{cases}.
\end{align*}
The conclusion of Lemma~\ref{lem: BSA diff op for vertex operator corr fn}
reads 
\begin{align*}
0=\; & \Big(\hat{\sD}_{d_{j}}^{(j)}\hat{f}^{(0)}\Big)(\hat{\boldsymbol{x}})=\Big(\hat{\sD}_{d_{j}}^{(j)}f^{(\ell)}\Big)(\boldsymbol{x};\boldsymbol{w}).
\end{align*}
Expanding by the binomial formula, we observe that the terms with
$A=\emptyset$ give precisely the LHS of the assertion, namely $\sD_{d_{j}}^{(j)}f^{(\ell)}$.
When $A\neq\emptyset$, choose the minimal $a\in A$, and write the
term in the form
\begin{align*}
 & \sum_{r}D_{w_{r}}\sL_{-n_{1}}^{(j)}\cdots\sL_{-n_{a-1}}^{(j)}M_{(w_{r}-x_{j})^{1-n_{a}}}O_{a+1}^{A}\cdots O_{k}^{A}\; f^{(\ell)}
\end{align*}
by noticing that $\sL_{-n_{1}}^{(j)}\cdots\sL_{-n_{a-1}}^{(j)}$ does
not contain $w_{r}$ and can be moved inside the differentiation $D_{w_{r}}$.
These remaining terms put together constitute the RHS of the assertion
\begin{align*}
 & \sum_{r=1}^{\ell}\pder{w_{r}}\left(g(w_{r};\boldsymbol{x};w_{1},\ldots,w_{r-1},w_{r+1},\ldots,w_{\ell})\times f^{(\ell)}(\boldsymbol{x};\boldsymbol{w})\right).
\end{align*}
It is easy to see that $g$ has no poles where the variables do not
coincide, and is symmetric in its last $\ell-1$ variables.
\end{proof}
We now conclude by the important property that the functions which
correspond to highest weight vectors satisfy the Benoit~\& Saint-Aubin
partial differential equations.
\begin{prop}\label{prop: BSA PDEs}
If $E.v=0$, then the function $\sF[v]$ satisfies
$\sD_{d_{j}}^{(j)}\sF[v]=0$.
\end{prop}
\begin{proof}
By linearity, we may assume that $v$ is an eigenvector of $K$, i.e.
that the number of integration variables $\ell$ is fixed. Dominated
convergence ensures that we can take the differential operator $\sD_{d_{j}}^{(j)}$
inside the integral, and thus let it act directly to the integrand
$f^{(\ell)}$. Corollary~\ref{cor: the BSA operator gives total derivative}
then implies that $\sD_{d_{j}}^{(j)}\sF[v](\boldsymbol{x})$ is a linear
combination of terms of the form studied in the integration by parts
formula, Lemma~\ref{lem: general integration by parts}. Assuming
$E.v=0$ it follows from Corollary~\ref{cor: integration by parts for hwv}
that this linear combination vanishes.
\end{proof}

\begin{rem}
\emph{
If one of the tensorands is a trivial representation, $d_{i}=1$,
the $i$:th PDE is merely
\begin{align*}
\pder{x_{i}}\sF[v](x_{1},\ldots,x_{i},\ldots,x_{n})=\; & 0,
\end{align*}
which is also a consequence of the statement of 
Lemma~\ref{lem: trivial dependence of basis functions on singlet variables}.
If $d_{i}=2$, the $i$:th PDE is similar to~\eqref{eq: multiple SLE PDEs},
and it can always be interpreted as a local martingale property of
the function for a chordal $\SLEk$.
}
\end{rem}

\subsection{\label{sub: Special conformal transformations}M\"obius covariance}

Translation invariance and scaling covariance were shown for all basis
functions $\FWint^{(x_{0})}$ in 
Lemma~\ref{lem: translation and scaling of basis functions}.
We will now show that if a vector $v$ is in a trivial (one-dimensional)
subrepresentation of the entire tensor product, then the corresponding
function $\sF[v]$ transforms covariantly under all M\"obius transformations
as in~\eqref{eq: Mobius covariance for a function}, with
covariance weights $h_{1,d_{i}}$ given 
by~\eqref{eq: Kac labeled conformal weights}.

We first record a property of the integrand, which will be used in
the proof of M\"obius covariance.
\begin{lem}\label{lem: infinitesimal special conformal transformations}
If $\ell=\frac{1}{2}\sum_{i=1}^{n}(d_{i}-1)$,
then the function $f^{(\ell)}$ satisfies the partial differential equation
\begin{align*}
\left(\sum_{i=1}^{n}\Big(x_{i}^{2}\pder{x_{i}}+2h_{1,d_{i}}\,x_{i}\Big)\right)
f^{(\ell)}(\boldsymbol{x};\boldsymbol{w})
=\; & \sum_{r=1}^{\ell}\pder{w_{r}}\left(g(w_{r};\boldsymbol{x};w_{1},\ldots,w_{r-1},w_{r+1},\ldots,w_{\ell})\times f^{(\ell)}(\boldsymbol{x};\boldsymbol{w})\right),
\end{align*}
where $g$ is a rational function which is symmetric in the last $\ell-1$
variables, and whose only poles are where some of its arguments coincide.\end{lem}
\begin{proof}
We will perform an explicit calculation that shows the claimed identity, with
\begin{align}\label{eq: the explicit rational function in Mobius covariance}
g(w_{r};\boldsymbol{x};w_{1},\ldots,w_{r-1},w_{r+1},\ldots,w_{\ell})
=\; & -w_{r}^{2}+\prod_{i=1}^{n}(w_{r}-x_{i})^{d_{i}-1}\prod_{s\neq r}(w_{r}-w_{s})^{-2}.
\end{align}
One can begin by observing, by a direct calculation that uses the
assumption $\ell=\frac{1}{2}\sum_{i=1}^{n}(d_{i}-1)$, that
\begin{align*}
 & \left(\sum_{i=1}^{n}\Big(x_{i}^{2}\pder{x_{i}}+2h_{1,d_{i}}\,x_{i}\Big)\right)
 f^{(\ell)}(\boldsymbol{x};\boldsymbol{w})\\
=\; & -\sum_{r=1}^{\ell}\pder{w_{r}}\left(w_{r}^{2}\, f^{(\ell)}(\boldsymbol{x};\boldsymbol{w})\right)+\left(1-\frac{4}{\kappa}\right)\left(\sum_{i=1}^{n}(1-d_{i})x_{i}+2\sum_{r=1}^{\ell}w_{r}\right)f^{(\ell)}(\boldsymbol{x};\boldsymbol{w}).
\end{align*}
Comparing with the explicit $g$ in 
Equation~\eqref{eq: the explicit rational function in Mobius covariance},
the claim now reduces to the vanishing of
\begin{align}\label{eq: the rational function that should vanish}
 & \sum_{r=1}^{\ell}\prod_{i=1}^{n}(w_{r}-x_{i})^{d_{i}-1}\prod_{s\neq r}(w_{r}-w_{s})^{-2}\left(\sum_{j=1}^{n}\frac{d_{j}-1}{w_{r}-x_{j}}-2\sum_{u\neq r}\frac{1}{w_{r}-w_{u}}\right)-2\sum_{r=1}^{\ell}w_{r}+\sum_{i=1}^{n}(d_{i}-1)x_{i},
\end{align}
which is a rational function in the variables $w_{r}$, $r=1,\ldots,\ell$.
Note that there are no poles 
of~\eqref{eq: the rational function that should vanish}
except possibly poles of degree at most three at $w_{r}=w_{s}$ for
some $r\neq s$. To show that these points are in fact not poles,
it is by symmetry sufficient to consider the Laurent series expansion
in, for example, the difference $\epsilon=w_{2}-w_{1}$. This can
be done in a straightforward manner. In addition, one verifies that
with fixed $w_{2},\ldots,w_{\ell}$, as $w_{1}\to\infty$, the function
tends to zero. Thus, the 
expression~\eqref{eq: the rational function that should vanish}
is an entire function of $w_{1}$ tending to zero at infinity, and
as such vanishes identically. This concludes the proof.\end{proof}
\begin{prop}\label{prop: full Mobius covariance}
If $E.v=0$ and $K.v=v$, then for any M\"obius transformation 
$\Mob\colon\bH\to\bH$ such that $\Mob(x_{1})<\Mob(x_{2})<\cdots<\Mob(x_{n})$,
the function $\sF[v]$ satisfies
\begin{align*}
\prod_{i=1}^{n}\Mob'(x_{i})^{h_{1,d_{i}}}
\times\sF[v](\Mob(x_{1}),\ldots,\Mob(x_{n}))
=\; & \sF[v](x_{1},\ldots,x_{n}).
\end{align*}
\end{prop}
\begin{proof}
Any M\"obius transformation $\Mob\colon\bH\to\bH$ can be written as
a composition of a translation $z\mapsto z+\xi$ (for some $\xi\in\bR$),
a scaling $z\mapsto\lambda z$ (for some $\lambda>0$), and a special
conformal transformation $z\mapsto\frac{z}{1+az}$ (for some $a\in\bR$).
Lemma~\ref{lem: translation and scaling of basis functions} shows
the assertion for translations and scalings. It suffices to prove
the statement for special conformal transformations.

For the special conformal transformation,
we may assume that $x_1 < 0$ and $x_n > 0$, by precomposing with
a translation if necessary.
Then, the special conformal transformation $\Mob_{a}(z)=\frac{z}{1+az}$ respects
the order of the boundary points $x_{1},\ldots,x_{n}$
if  $a\in(\frac{-1}{x_{n}},\frac{-1}{x_{1}})$.
To obtain the general case we will integrate an infinitesimal form
of the formula starting from the trivial case of $a=0$.

Consider a term 
\begin{align*}
 & \prod_{i=1}^{n}\Mob'(x_{i})^{h_{1,d_{i}}}\times\int_{\SurfFW_{l_{1},\ldots,l_{n}}}f_{l_{1},\ldots,l_{n}}^{\Supset}(\Mob_{a}(x_{1}),\dots,\Mob_{a}(x_{n});w_{1},\ldots,w_{\ell})\;\ud w_{1}\cdots\ud w_{\ell}
\end{align*}
on the left hand side of the asserted equation. 
Using the identities $\der a\Mob_{a}(z)=-\Mob_{a}(z)^{2}$
and $\der a(\Mob_{a}'(z)^{h})=-2h\Mob_{a}(z)\times\Mob_{a}'(z)^{h}$
we compute its derivative with respect to $a$,
\begin{align*}
 & \der a\left(\prod_{i=1}^{n}\Mob'(x_{i})^{h_{1,d_{i}}}\times\int_{\SurfFW_{l_{1},\ldots,l_{n}}}f_{l_{1},\ldots,l_{n}}^{\Supset}(\Mob_{a}(x_{1}),\dots,\Mob_{a}(x_{n});\boldsymbol{w})\;\ud w_{1}\cdots\ud w_{\ell}\right)\\
=\; & \left(\prod_{i=1}^{n}\Mob'(x_{i})^{h_{1,d_{i}}}\right)\times\int_{\SurfFW_{l_{1},\ldots,l_{n}}}\big(\sL_{1}f_{l_{1},\ldots,l_{n}}^{\Supset}\big)(\Mob_{a}(x_{1}),\dots,\Mob_{a}(x_{n});\boldsymbol{w})\;\ud w_{1}\cdots\ud w_{\ell},
\end{align*}
where 
\begin{align*}
\sL_{1}=\; & -\sum_{i=1}^{n}\left(x_{i}^{2}\pder{x_{i}}+2h_{1,d_{i}}x_{i}\right).
\end{align*}
The assumption $K.v=v$ guarantees that in these terms $\ell=\sum_{i}l_{i}$
takes the value $\ell=\frac{1}{2}\sum_{i=1}^{n}(d_{i}-1)$. Thus,
by Lemma~\ref{lem: infinitesimal special conformal transformations},
we have
\begin{align*}
\big(\sL_{1}f_{l_{1},\ldots,l_{n}}^{\Supset}\big)(\boldsymbol{x};\boldsymbol{w})
=\; & \sum_{r=1}^{\ell}\pder{w_{r}}\left(g(w_{r};\boldsymbol{x};w_{1},\ldots,w_{r-1},w_{r+1},\ldots,w_{\ell})\times f_{l_{1},\ldots,l_{n}}^{\Supset}(\boldsymbol{x};\boldsymbol{w})\right),
\end{align*}
where $g$ single-valued and symmetric with respect to the last $\ell-1$
variables. Since $E.v=0$ we can apply 
Corollary~\ref{cor: integration by parts for hwv} 
to the $a$-derivative of the left hand side of the asserted formula,
and get
\begin{align*}
\der a\left(\prod_{i=1}^{n}\Mob'(x_{i})^{h_{1,d_{i}}}\times\sF[v](\Mob(x_{1}),\ldots,\Mob(x_{n}))\right)=\; & 0.
\end{align*}
It now follows that also the left hand side of the asserted formula
is constant in $a$ for $a\in(\frac{-1}{x_{n}},\frac{-1}{x_{1}})$.
At $a=0$ we have $\Mob_{a}=\id_{\bH}$, so this constant equals $\sF[v](x_{1},\ldots,x_{n})$.
\end{proof}

\subsection{\label{sub: the main correspondence result}Main theorems about the correspondence}

We now give the precise statements of the main results about 
the spin chain~- Coulomb gas correspondence, and finish their proofs.
These results show how the properties of the vector 
$v\in\Wd_{d_{n}}\tens\cdots\tens\Wd_{d_{1}}$
translate to properties of the function 
$(x_{0};x_{1},\ldots,x_{n})\mapsto\sF^{(x_{0})}[v](x_{1},\ldots,x_{n})$.
Three types of properties of the functions are considered: partial
differential equations (PDE), covariance properties (COV) and asymptotics (ASY).

Recall the following notation, to be used in the statement below.
The exponents $h_{1,d}$ and $\Delta_{d}^{d_{1},\ldots,d_{n}}$ are given by 
\begin{align*}
h_{1,d}=\; & \frac{(d-1)(2(d+1)-\kappa)}{2\kappa}\qquad\qquad\text{ and } & 
\Delta_{d}^{d_{1},\ldots,d_{n}}=\; & h_{1,d}-\sum_{i=1}^{n}h_{1,d_{i}}.
\end{align*}
The multiplicative constants $B_{d}^{d',d''}$ are defined by the analytic 
continuation in $\kappa$ of the
integrals~\eqref{eq: generalized beta constants} over $m$-dimensional simplices, with $m=\frac{1}{2}\left(d'+d''-1-d\right)$, which by 
Remark~\ref{rem: Selberg integral} take the form
\begin{align*}
B_{d}^{d',d''} = \; & \frac{1}{m!} \; \prod_{p=1}^m
    \frac{\Gamma\big( 1 - \frac{4}{\kappa}(d'-p)\big) \; \Gamma\big( 1 - \frac{4}{\kappa}(d''-p)\big) \; \Gamma\big( 1 + \frac{4}{\kappa}p\big)}
        {\Gamma\big( 1 + \frac{4}{\kappa}\big) \; \Gamma\big( 2 - \frac{4}{\kappa}(d'+d''-m-p)\big)} .
\end{align*}
The partial differential operators
$\sD_{d_{j}}^{(j)}$ (which depend also on $d_{1},\ldots d_{n}$)
are defined in terms of
\begin{align*}
\sL_{-p}^{(j)} 
:= \; & -\sum_{i\neq j}\left((x_{i}-x_{j})^{1-p}\pder{x_{i}}+(1-p)\, h_{1,d_{i}}\,(x_{i}-x_{j})^{-p}\right),
\end{align*}
by the Benoit~\& Saint-Aubin formula
\begin{align}\label{eq: BSA differential operator}
\sD_{d_{j}}^{(j)} 
:= \; & \sum_{k=1}^{d_{j}}\sum_{\substack{p_{1},\ldots,p_{k}\geq1\\
p_{1}+\cdots+p_{k}=d_{j}}}
\frac{(-4/\kappa)^{d_{j}-k}\,(d_{j}-1)!^{2}}{\prod_{u=1}^{k-1}(\sum_{i=1}^{u}p_{i})(\sum_{i=u+1}^{k}p_{i})}\times\sL_{-p_{1}}^{(j)}\cdots\sL_{-p_{k}}^{(j)}.
\end{align}

The most important cases concern highest weight vectors, i.e., vectors
$v$ satisfying $E.v=0$. We first state, however, the properties
that do not depend on this assumption.
\begin{thm}\label{thm: SCCG correspondence non-hwv}
Let $v\in\Wd_{d_{n}}\tens\cdots\tens\Wd_{d_{1}}$.
The function $(x_{0};x_{1},\ldots,x_{n})\mapsto\sF^{(x_{0})}[v](x_{1},\ldots,x_{n})$
satisfies the following properties:
\begin{description}
\item [{(COV)}] For any $\xi\in\bR$ we have the translation invariance
\begin{align*}
\sF^{(x_{0}+\xi)}[v](x_{1}+\xi,\ldots,x_{n}+\xi)
=\; & \sF^{(x_{0})}[v](x_{1},\ldots,x_{n}).
\end{align*}
If furthermore $K.v=q^{d-1}\,v$, then for any $\lambda>0$ we have
the scaling covariance
\begin{align*}
\sF^{(\lambda x_{0})}[v](\lambda x_{1},\ldots,\lambda x_{n})
=\; & \lambda^{\Delta_{d}^{d_{1},\ldots,d_{n}}}
\times\sF^{(x_{0})}[v](x_{1},\ldots,x_{n}).
\end{align*}

\item [{(ASY)}] If $v=\pi_{j,j+1}^{(d)}(v)$ and we denote
\begin{align*}
\hat{v}=\; & \hat{\pi}_{j,j+1}^{(d)}(v)\,\in\,
\Big(\bigotimes_{i=j+2}^{n}\Wd_{d_{i}}\Big)\tens\Wd_{d}\tens\Big(\bigotimes_{i=1}^{j-1}\Wd_{d_{i}}\Big),
\end{align*}
then we have 
\begin{align*}
\lim_{x_{j},x_{j+1}\to\xi}
\Big( & (x_{j+1}-x_{j})^{-\Delta_{d}^{d_{j},d_{j+1}}}
\times\sF^{(x_{0})}[v](x_{1},\ldots,x_{n}) \Big) \\
=\; & B_{d}^{d_{j},d_{j+1}} \times 
\sF^{(x_{0})}[\hat{v}](x_{1},\ldots,x_{j-1},\xi,x_{j+2},\ldots,x_{n}).
\end{align*}
\end{description}
\end{thm}
\begin{proof}
The statements (COV) about the translation invariance and homogeneity
of $\sF^{(x_{0})}[v](x_{1},\ldots,x_{n})$ are clear from the corresponding
properties of the basis functions $\FWint_{l_{1},\ldots,l_{n}}^{(x_{0})}$
stated in Lemma~\ref{lem: translation and scaling of basis functions}.
The statement (ASY) was proved in 
Proposition~\ref{prop: asymptotics with subrepresentations}.
\end{proof}
In the case of highest weight vectors, the functions become independent
of the anchor point $x_{0}$ and they solve Benoit~\& Saint-Aubin partial differential
equations. In the further special case of vectors in a trivial subrepresentation,
one obtains full M\"obius covariance. These properties are stated precisely
in the following theorem.
\begin{thm}\label{thm: SCCG correspondence hwv}
Assume that $v\in\Wd_{d_{n}}\tens\cdots\tens\Wd_{d_{1}}$
satisfies $E.v=0$. Then $\sF^{(x_{0})}[v](x_{1},\ldots,x_{n})$ is
independent of $x_{0}$, and thus defines a function $\sF[v]$ on
$\chamber_{n}$. This function $(x_{1},\ldots,x_{n})\mapsto\sF[v](x_{1},\ldots,x_{n})$
satisfies the following properties:
\begin{description}
\item [{(PDE)}] With $\sD_{d_{j}}^{(j)}$ the differential 
operator~\eqref{eq: BSA differential operator}, we have
\begin{align*}
\sD_{d_{j}}^{(j)}\sF[v]=\; & 0\qquad\text{for }j=1,\ldots,n.
\end{align*}
\item [{(COV)}] For any $\xi\in\bR$ we have the translation invariance
\begin{align*}
\sF[v](x_{1}+\xi,\ldots,x_{n}+\xi)=\; & \sF[v](x_{1},\ldots,x_{n}).
\end{align*}
If furthermore $K.v=q^{d-1}\,v$, then for any $\lambda>0$ we have
the scaling covariance
\begin{align*}
\sF[v](\lambda x_{1},\ldots,\lambda x_{n})
=\; & \lambda^{\Delta_{d}^{d_{1},\ldots,d_{n}}}\times\sF[v](x_{1},\ldots,x_{n}).
\end{align*}
If furthermore $K.v=v$, then we have the full M\"obius covariance
\begin{align*}
\prod_{j=1}^{n}\Mob'(x_{j})^{h_{1,d_{j}}}
\times\sF[v]\left(\Mob(x_{1}),\ldots,\,\Mob(x_{n})\right)
=\; & \sF[v](x_{1},\ldots,\, x_{n})
\end{align*}
for any M\"obius transformation $\Mob\colon\bH\to\bH$ such that $\Mob(x_{1})<\Mob(x_{2})<\cdots<\Mob(x_{n})$.
\item [{(ASY)}] If $v=\pi_{j,j+1}^{(d)}(v)$ and we denote
\begin{align*}
\hat{v}=\; & \hat{\pi}_{j,j+1}^{(d)}(v)\,\in\,
\Big(\bigotimes_{i=j+2}^{n}\Wd_{d_{i}}\Big)\tens\Wd_{d}\tens\Big(\bigotimes_{i=1}^{j-1}\Wd_{d_{i}}\Big),
\end{align*}
then also $E.\hat{v}=0$ and we have 
\begin{align*}
\lim_{x_{j},x_{j+1}\to\xi}\Big( & \, (x_{j+1}-x_{j})^{-\Delta_{d}^{d_{j},d_{j+1}}}\times\sF[v](x_{1},\ldots,x_{n})\Big) \\
=\; & B_{d}^{d_{j},d_{j+1}}\times\sF[\hat{v}](x_{1},\ldots,x_{j-1},\xi,x_{j+2},\ldots,x_{n}).
\end{align*}
\end{description}
\end{thm}
\begin{proof} 
That $\sF^{(x_{0})}[v]$ does not depend on $x_{0}$ follows from
Proposition~\ref{prop: closed integration surface} and 
Remark~\ref{rem: anchor point independence}
after it --- we may define
\begin{align*}
\sF[v](x_{1},\ldots,x_{n})
=\; & \sF^{(x_{0})}[v](x_{1},\ldots,x_{n})\qquad\text{for any }x_{0}<x_{1}.
\end{align*}

The statement (PDE) was shown in Proposition~\ref{prop: BSA PDEs}.

Of the statements (COV), the first two are direct consequences of
the corresponding statements in Theorem~\ref{thm: SCCG correspondence non-hwv}
for $\sF^{(x_{0})}[v]$. The third statement was shown in 
Proposition~\ref{prop: full Mobius covariance}.

The statement (ASY) is a direct consequence of 
the fact that $\pi_{j,j+1}^{(d)}$ is a projection to a subrepresentation, and
the corresponding statement in Theorem~\ref{thm: SCCG correspondence non-hwv}
for $\sF^{(x_{0})}[v]$.
\end{proof}

\begin{rem}\label{rem: benign roots of unity}
\emph{%
We have stated our main results assuming $q$ is not a root of unity.
Let us, however, point out that certain
root of unity cases work out completely parallel to the generic case.
}

\emph{%
For $q$ a root of unity, denote by $p$ the smallest positive
integer such that $q^{p} \in \set{\pm 1}$, or equivalently,
the smallest positive integer such that $\qnum{p}=0$.
Then the representations~$\Wd_d$ are still irreducible if ${d-1 < p}$,
and moreover, the conclusion and formulas in the Clebsch-Gordan
decomposition of Lemma~\ref{lem: tensor product representations of quantum sl2}
remain valid as long as $(d'-1)+(d''-1) < p$.
Using this repeatedly, one sees that the tensor product $\bigotimes_{j=1}^n \Wd_{d_j}$
used in the spin chain~-~Coulomb gas correspondence retains the generic 
direct sum decomposition as long as $\sum_{j=1}^n (d_j - 1) < p$.
It is also easy to check that, under the assumption $\sum_{j=1}^n (d_j - 1) < p$, all
formulas in Section~\ref{sec: various integral functions} for the integral functions
remain valid, and that they
can be applied as we have done in this section to derive the above main theorems.
In conclusion, while the general root of unity case may be very complicated
(as discussed in Section~\ref{sec: conclusions}), at least for $\sum_{j=1}^n (d_j - 1) < p$ the
statements need not be changed!}

\emph{%
By this observation, for instance the pure partition functions $\mathcal{Z}_\alpha$
of multiple $\SLEk$ described in Section~\ref{sss: multiple SLEs} can be constructed
for rational $\kappa \in (0,8) \cap \bQ$ exactly like in the
generic case, as long as the number $N$ of curves is sufficiently small, namely $N < \frac{p}{2}$
where $p=p(\kappa)$ is the smallest positive integer such that $\frac{4}{\kappa} p $ is an integer.
}
\end{rem}

\bigskip{}

\section{\label{sec: further properties}Further properties}

In this section, we establish two more properties of the correspondence.
First, in Section~\ref{sub: general asymptotics}, we treat
a generalization of the asymptotics property (ASY),
to the case where more than two of the variables tend to a common limit.
Then, in Section~\ref{sub: role of infinity}, we apply the general
asymptotics statement to consider a limit
where either the first or the last variable of
a M\"obius covariant function is taken to infinity.
This property pertains to the manifestation of the
periodicity of the boundary of a simply connected domain,
discussed in Section~\ref{sub: cyclic permutations}.

\subsection{\label{sub: general asymptotics}Asymptotics as several variables tend to a common limit}

Suppose that $1 \leq j < k \leq n$.
We first address what happens to the function
\[ \sF^{(x_0)}[v](x_1, \ldots, x_j, \ldots, x_k , \ldots , x_n) \]
in the limit
\begin{align}
\label{eq: limit at same rate}
x_j, x_{j+1}, \ldots, x_{k-1}, x_k \to \; & \xi
    \qquad \;\text{ taken in such a way that } \\
\nonumber
\frac{x_{i}-x_j}{x_{k}-x_j} \to \; & \eta_i
    \qquad \text{ for } i \in \set{j,j+1, \ldots, k-1,k} .
\end{align}
We assume that $x_{j-1} < \xi < x_{k+1}$
and $0 = \eta_j < \eta_{j+1} < \cdots < \eta_{k-1} < \eta_k = 1$.

The case $k=j+1$, when only two points come together, was the content of part (ASY) of
Theorems~\ref{thm: SCCG correspondence hwv} and \ref{thm: SCCG correspondence non-hwv}.
It was based on the decomposition \[\Wd_{d_{j+1}} \tens \Wd_{d_j} \isom \bigoplus_d \Wd_d . \]
An important difference in the present case will be that in the decomposition
\[ \Wd_{d_{k}} \tens \cdots \tens \Wd_{d_j} \isom \bigoplus_d m_d \Wd_d \]
some of the multiplicities $m_d$ may be greater than one.
Note, however, that any vector in the space
$m_d \Wd_d \subset \Wd_{d_{k}} \tens \cdots \tens \Wd_{d_j}$
may be written as a linear combination of vectors of the form $F^l.\tau_0$, 
where $0 \leq l < d$, and $\tau_0$ is some highest weight vector of 
an irreducible subrepresentation of dimension $d$
(there are $m_d$ linearly independent such highest weight vectors).
It is therefore possible to deduce the behavior of the functions in general
from the following result.
\begin{prop}\label{prop: general asymptotics with subrepresentations}
Suppose that $\tau_0 \in \Wd_{d_{k}} \tens \cdots \tens \Wd_{d_j}$
satisfies $E.\tau_0 = 0$ and $K.\tau_0 = q^{d-1} \,\tau_0$. Let
\begin{align*}
v = \; & \Wbas_{l_n} \tens \cdots \tens \Wbas_{l_{k+1}} \tens F^l.\tau_0 \tens \Wbas_{l_{j-1}} \tens \cdots \tens \Wbas_{l_1}
    & & \in \, \Wd_{d_n} \tens \cdots \tens \Wd_{d_1} \\
\hat{v} = \; & \Wbas_{l_n} \tens \cdots \tens \Wbas_{l_{k+1}} \tens \Wbas_l \tens \Wbas_{l_{j-1}} \tens \cdots \tens \Wbas_{l_1}
    & & \in \, \Wd_{d_n} \tens \cdots \tens \Wd_{d_{k+1}} \tens \Wd_d \tens \Wd_{d_{j-1}} \tens \cdots \tens \Wd_{d_1} ,
\end{align*}
and let $\Delta = \Delta_d^{d_j , \ldots, d_k}$ as in 
Section~\ref{sub: the main correspondence result}.
Then in the limit~\eqref{eq: limit at same rate} we have
\[ \frac{\sF^{(x_0)}[v](x_1 , \ldots, x_n)}{|x_k - x_j|^\Delta}
 \; \longrightarrow \;
\sF[\tau_0](\eta_j, \ldots \eta_k) \times 
\sF^{(x_0)}[\hat{v}](x_1 , \ldots, x_{j-1} , \xi , x_{k+1} , \ldots, x_n) .
\]
\end{prop}
We have preferred to formulate the above proposition for a vector $v$ of 
specific form. The reason is that the limit function's dependence on the ratios 
$\eta_j,\ldots,\eta_k$ depends on the $d$-dimensional irreducible 
subrepresentation whose highest weight vector is $\tau_0$.
Note, however, that the exponent $\Delta = \Delta_d^{d_j , \ldots, d_k}$  in the
asymptotics is the same for all $d$-dimensional subrepresentations
of $\Wd_{d_k} \tens \cdots \tens \Wd_{d_j}$,
and therefore a similar limit exists also more generally, as stated in the
following.
\begin{cor}
Suppose that $v \in \bigotimes_{i=1}^n \Wd_{d_i}$ belongs
to the subrepresentation
\[ \Big(\bigotimes_{i=k+1}^n \Wd_{d_i}\Big) \tens (m_d \Wd_{d}) \tens \Big(\bigotimes_{i=1}^{j-1} \Wd_{d_i}\Big) . \]
Then the expression
$ {\sF^{(x_0)}[v](x_1 , \ldots, x_n)} \times {|x_k - x_j|^{-\Delta}} $,
with $\Delta = \Delta_d^{d_j , \ldots, d_k}$, 
has a limit~\eqref{eq: limit at same rate}.
\end{cor}

\begin{proof}[Proof of Proposition~\ref{prop: general asymptotics with subrepresentations}]
The proof follows a strategy parallel to the simpler asymptotics properties
shown before.
The crucial step is to rearrange the integrations to a form where
there are deformed hypercube type contours between those variables
which tend to a common limit, and all other contours are loops
which either encircle one other variable, or encircle all the points
with a common limit together. Once rearranged this way,
dominated convergence theorem may be applied to complete the proof.
The rearrangement itself is done in two steps, analogous to the two lemmas
in Section~\ref{sub: asymptotics with subrepresentations}.
We only sketch the proofs, as the omitted details are reasonably
straightforward modifications of the case already considered.

Analogously to Lemma~\ref{lem: asymptotics for hwv in two consequtive},
we first handle the case $l=0$ in which $\tau_0$ is a highest weight vector.
We use the fact that the highest weight vector $\tau_0$ corresponds to a closed
integration contour. This is expressed precisely in 
Proposition~\ref{prop: closed integration surface},
which allows us to rearrange the integrals in
$\sF^{(x_0)}[\tau_0](x_j, \ldots, x_k)$ so that no 
integration contour starts from $x_0$:
\[ \Big( \sF^{(x_0)}[\tau_0] \Big) (x_j, \ldots, x_k)
    = \sum c_{m_{j+1},\ldots,m_{k}} \times 
    \CubeInt^{(x_0)}_{0,m_{j+1},\ldots,m_{k}}(x_j, \ldots, x_k) . \]
The sum $m_{j+1} + \cdots + m_{k}$ of the indices
is fixed, and it equals $\half (\sum_{i=j}^k d_i - k + j - d)$.
For the vector
\[ v = \Wbas_{l_n} \tens \cdots \tens \Wbas_{l_{k+1}} \tens \tau_0 \tens \Wbas_{l_{j-1}} \tens \cdots \tens \Wbas_{l_1} , \]
by the same rearrangement of integrals, we obtain the expression
\[ \Big( \sF^{(x_0)}[v] \Big) (x_1, \ldots, x_n)
    = \sum c_{m_{j+1},\ldots,m_{k}}
        \times \AsyInt^{(x_0)}_{l_1,\ldots,l_{j-1};0,\set{m_{j+1},\ldots,m_{k}};l_{k+1},\ldots,l_n}(x_1, \ldots, x_n) , \]
where $\AsyInt^{(x_0)}_{\cdots;0,\set{m_{j+1},\ldots,m_{k}};\cdots}$ 
are generalizations of the mixed integrals $\AsyInt^{(x_0)}_{\cdots;0,m;\cdots}$ 
defined in Section~\ref{sub: asymptotic integral functions}:
the integration contours to variables $x_j,\ldots,x_k$ are as in 
$\CubeInt^{(x_0)}_{0,m_{j+1},\ldots,m_{k}}$,
and all other $x_i$ are encircled by $l_i$ non-intersecting nested loops
based at the anchor point $x_0$.

The next step is to consider the case of general $l$, in a manner
analogous to Lemma~\ref{lem: asymptotics for tensor basis in two consequtive}.
By comparison with the $(k-j)$-fold coproduct formula for the quantum group 
generator $F$ given in Lemma~\ref{lem: multiple coproducts},
one shows recursively in $l$ that for vectors
\[ v = \Wbas_{l_n} \tens \cdots \tens \Wbas_{l_{k+1}} \tens F^l .\tau_0 \tens \Wbas_{l_{j-1}} \tens \cdots \tens \Wbas_{l_1} , \]
the integrals can be rearranged to
\[ \Big( \sF^{(x_0)}[v] \Big) (x_1, \ldots, x_n)
    = \sum c_{m_{j+1},\ldots,m_{k}}
        \times\AsyInt^{(x_0)}_{l_1,\ldots,l_{j-1};l,\set{m_{j+1},\ldots,m_{k}};l_{k+1},\ldots,l_n}(x_1, \ldots, x_n) , \]
where $\AsyInt^{(x_0)}_{\cdots;l,\set{m_{j+1},\ldots,m_{k}};\cdots}$ 
are generalizations of the mixed integrals $\AsyInt^{(x_0)}_{\cdots;l,m;\cdots}$. 
Compared to the case $l=0$, the new feature is that the variables 
$x_j,\ldots,x_k$ together with all integration contours connected
to them are encircled by $l$ non-intersecting nested loops based at 
the anchor point $x_0$.

To reach the conclusion, we need to perform the 
limit~\eqref{eq: limit at same rate} of 
$\AsyInt^{(x_0)}_{l_1,\ldots,l_{j-1};l,\set{m_{j+1},\ldots,m_{k}};l_{k+1},\ldots,l_n}$.
Once we divide by $|x_k - x_j|^{\Delta}$, dominated convergence theorem can be 
applied to the integration over all variables whose contour is a loop, since these 
contours remain bounded away from the points $x_j , \ldots , x_k$ and any
hypercube type integration contours between them.
The loop type integration contours are the same as for
$\sF^{(x_0)}[\hat{v}](x_1 , \ldots , x_{j-1} , \xi , x_{k+1}, \ldots, x_n)$.
The integral over the hypercube type contour divided by $|x_k - x_j|^{\Delta}$
tends to the integrand of $\sF^{(x_0)}[\hat{v}]$
multiplied by $\sF[\tau_0](\eta_j, \ldots, \eta_k)$.
The asserted result follows.
\end{proof}

\subsection{\label{sub: role of infinity}Moving one point to infinity}

In the M\"obius covariant case, we will now consider what happens to the function
\[ \sF[v](x_1, \ldots , x_n) \]
as $x_n \to + \infty$. For this, we will need to be able to
move from the trivial subrepresentation of $\bigotimes_{j=1}^n \Wd_{d_j}$
to the sum of copies of $d_n$-dimensional irreducible subrepresentations
of $\bigotimes_{j=1}^{n-1} \Wd_{d_j}$.

Symmetrically, we consider what happens to the above function as 
$x_1 \to -\infty$, in which case we will need to be able to
move from the trivial subrepresentation of $\bigotimes_{j=1}^n \Wd_{d_j}$
to the sum of copies of $d_1$-dimensional irreducible subrepresentations
of $\bigotimes_{j=2}^{n} \Wd_{d_j}$. The following lemma provides the
needed mappings in the two cases.

\begin{lem}\label{lem: isomorphism of hwv spaces}
Let $\HWsp_1 \subset \bigotimes_{j=1}^n \Wd_{d_j}$ denote the trivial subrepresentation
\[ \HWsp_1 = 
\Big\{v \in \bigotimes_{j=1}^n \Wd_{d_j} \; \Big| \; E.v=0 , \; K.v=v \Big\} . \]
\begin{itemize}
\item[(a)]
Any vector $v \in \HWsp_1$ can be written uniquely in the form
\begin{align*}
v = \; & \sum_{l_n=0}^{d_n-1} A^+_{l_n} \times 
\Wbas_{l_n} \tens (F^{d_n-1-l_n} . \tau_0^+),
\qquad \text{ with } A^+_{l_n} = (-1)^{d_n-1-l_n} \, q^{(l_n+1)(d_n-1-l_n)},
\end{align*}
where $\tau_0^+ \in \bigotimes_{j=1}^{n-1} \Wd_{d_j}$
satisfies $E.\tau_0^+ = 0$ and $K.\tau_0^+ = q^{d_n-1} \,\tau_0^+$. 
The mapping $v \mapsto \tau_0^+ = R_+(v)$ defines a linear
isomorphism $R_+ \colon \HWsp_1 \to \HWsp_{d_n}^+$ to the space
\begin{align*}
\HWsp_{d_n}^+ = 
\Big\{\tilde{v} \in \bigotimes_{j=1}^{n-1} \Wd_{d_j} \; \Big| \;
    E.\tilde{v}=0 , \; K.\tilde{v} = q^{d_n - 1} \,\tilde{v} \Big\}
\end{align*}
of highest weight vectors of irreducible subrepresentations of dimension $d_n$ 
in
$\bigotimes_{j=1}^{n-1} \Wd_{d_j}$.
\item[(b)]
Any vector $v \in \HWsp_1$ can be written uniquely in the form
\begin{align*}
v = \; & \sum_{l_1=0}^{d_1-1} A^-_{l_1} \times 
(F^{d_1-1-l_1} . \tau_0^- ) \tens \Wbas_{l_1},
\qquad \text{ with } A^-_{l_1} = (-1)^{d_1-1-l_1} \,q^{(l_1-1)(d_1-1-l_1)},
\end{align*}
where $\tau_0^- \in \bigotimes_{j=2}^{n} \Wd_{d_j}$
satisfies $E.\tau_0^- = 0$ and $K.\tau_0^- = q^{d_1-1} \,\tau_0^-$. 
The mapping $v \mapsto \tau_0^- = R_-(v)$ defines a linear isomorphism 
$R_- \colon \HWsp_1 \to \HWsp_{d_1}^-$ to the space
\begin{align*}
\HWsp_{d_1}^- = \Big\{\tilde{v} \in \bigotimes_{j=2}^{n} \Wd_{d_j} \; \Big| \;
    E.\tilde{v}=0 , \; K.\tilde{v} = q^{d_1 - 1} \,\tilde{v} \Big\}
\end{align*}
of highest weight vectors of irreducible subrepresentations of dimension $d_1$ 
in
$\bigotimes_{j=2}^{n} \Wd_{d_j}$.
\end{itemize}
\end{lem}
\begin{proof}
The two parts are similar, so we only give the details for part (a).
Any vector $v$ in the tensor product $\bigotimes_{j=1}^{n} \Wd_{d_j}$ 
can be written in the form
$v = \sum_{l_n} \Wbas_{l_n} \tens u_{l_n}$ with
unique vectors $u_{l_n} \in \bigotimes_{j=1}^{n-1} \Wd_{d_j}$.
From the eigenvalue property 
$K.\Wbas_{l_n} = q^{d_n - 1 - 2 l_n}\,\Wbas_{l_n}$ and
the coproduct formula $\Hcp(K)=K \tens K$, it follows that we have
\[ K.v = \sum_{l_n} q^{d_n - 1 - 2 l_n} \times \Wbas_{l_n} \tens (K.u_{l_n}) .\]
The assumption $K.v = v$ thus implies that we have $K.u_{l_n} = q^{1- d_n + 2 l_n} \,u_{l_n}$.
Similarly, from the property 
$E.\Wbas_{l_n} = \qnum{l_n} \qnum{d_n - l_n} \Wbas_{l_n-1}$, and
the coproduct formula $\Hcp(E) = E \tens K + 1 \tens E$, and the already 
established $K$-eigenvalue of $u_{l_n}$, it follows that we have
\[ E.v = \sum_{l_n} \Wbas_{l_n} \tens \Big( E . u_{l_n}
    + \qnum{l_n+1}\qnum{d_n-l_n-1} q^{3 - d_n + 2 l_n} \,u_{l_n+1} \Big) .\]
The assumption $E.v = 0$ thus implies
$E.u_{l_n} = - \qnum{l_n+1}\qnum{d_n-l_n-1} q^{3 - d_n + 2 l_n} \,u_{l_n+1}$.
In particular, $u_{d_n-1}$ is a highest weight vector of an irreducible
subrepresentation of dimension $d_n$.
We denote this vector by $u_{d_n -1} = \tau_0^+$.
We furthermore claim that this vector determines the other $u_{l_n}$
uniquely as $u_{l_n} = A^+_{l_n} \, F^{d_n-1-l_n}.\tau_0^+$.
To see this, use the coproduct
formula $\Hcp(F) = F \tens 1 + K^{-1}\tens F$, to get
\[ F.v = \sum_{l_n} \Wbas_{l_n} \tens \Big( u_{l_n-1}
    + q^{1 - d_n + 2 l_n} F.u_{l_n} \Big) .\]
Note that this expression for $F.v$ must vanish, since $v$ is in the trivial
subrepresentation. Therefore the other $u_{l_n}$ are obtained recursively
from $u_{d_n -1} = \tau_0^+$,
by $u_{l_n - 1} = - q^{1 - d_n + 2 l_n} \,F.u_{l_n}$. 
The solution of this recursion is the asserted formula 
$u_{l_n} = A^+_{l_n} \, F^{d_n-1-l_n}.\tau_0^+$, and since this formula indeed 
satisfies $E.v=0$ and $K.v=v$, we get that the linear
mapping $R_+$ defined by $R_+(v) = u_{d_n -1}$ is bijective.
\end{proof}

We now show that the behavior of M\"obius covariant functions as $x_n \to +\infty$
(resp. $x_1 \to - \infty$) can be expressed in terms of the 
identification $R_+$ (resp. $R_-$) defined in 
Lemma~\ref{lem: isomorphism of hwv spaces}.
\begin{prop}\label{prop: the role of infinity in Mobius coveriant case}
Let $v \in \HWsp_1 \subset \bigotimes_{j=1}^n \Wd_{d_j}$, 
and use the mappings defined in Lemma~\ref{lem: isomorphism of hwv spaces} to 
construct
$
    R_+(v) \in \HWsp^+_{d_n} \subset \bigotimes_{j=1}^{n-1} \Wd_{d_j}$ and
$
    R_-(v) \in \HWsp^-_{d_1} \subset \bigotimes_{j=2}^{n} \Wd_{d_j}$.
Then we have
\[ \lim_{y \to + \infty} \Big( y^{2h_{1,d_n}} \times \sF[v](x_1,\ldots,x_{n-1},y) \Big)
    = C_+ \times \sF[R_+(v)](x_1,\ldots,x_{n-1}) , \]
with $C_+ = (q-q^{-1})^{d_n-1} \qfact{d_n - 1}^2 \times B^{d_n,d_n}_1$ and
\[ \lim_{y \to - \infty} \Big( |y|^{2h_{1,d_1}} \times \sF[v](y,x_2,\ldots,x_{n}) \Big)
    = C_- \times \sF[R_-(v)](x_2,\ldots,x_{n}) .\]
with $C_- = (q^{-2}-1)^{d_1-1} \qfact{d_1-1}^2 \times B^{d_1,d_1}_1$.
\end{prop}
\begin{proof}
Again, the two cases are similar, so we only give the details about the first.
We write the vector $v$ in the form given by 
Lemma~\ref{lem: isomorphism of hwv spaces}(a),
as $v = \sum_{l_n} A^+_{l_n} \times \Wbas_{l_n} \tens F^{d_n-1-l_n}.\tau_0^+ $,
where $\tau_0^+=R_+(v)$ is a highest weight vector of a $d_n$-dimensional subrepresentation
of $\bigotimes_{j=1}^{n-1} \Wd_{d_j}$.
We then write, using homogeneity,
\[ y^{2h_{1,d_n}} \times \sF[v](x_1,\ldots,x_{n-1},y) 
= y^{h_{1,d_n}-\sum_{j=1}^{n-1}h_{1,d_j}}
    \times \sF[v] \big(\frac{x_{1}}{y},\ldots,\frac{x_{n-1}}{y},1 \big) . \]
We apply Proposition~\ref{prop: general asymptotics with subrepresentations}
to this. More precisely, for all terms 
$\Wbas_{l_n} \tens F^{d_n-1-l_n}.\tau_0^+$ we have
\begin{align*}
& \lim_{y \to +\infty} \Big( y^{2h_{1,d_n}} \times
    \sF^{(x_0)}\big[\Wbas_{l_n} \tens F^{d_n-1-l_n}.\tau_0^+\big](x_1,\ldots,x_{n-1},y) \Big) \\
= \; & \lim_{y \to +\infty} \Big( y^{h_{1,d_n}-\sum_{j=1}^{n-1}h_{1,d_j}}
    \times \sF^{(x_0)}\big[\Wbas_{l_n} \tens F^{d_n-1-l_n}.\tau_0^+\big] 
    \big(\frac{x_{1}}{y},\ldots,\frac{x_{n-1}}{y},1 \big) \Big) \\
= \; & \sF[\tau_0^+](x_1,\ldots,x_{n-1}) \times
    \sF^{(x_0)}[\Wbas_{l_n}\tens\Wbas_{d_n-1-l_n}] (0,1) .
\end{align*}
Now the vector $v$ reads 
$v = \sum_{l_n} A^+_{l_n} \times \Wbas_{l_n} \tens F^{d_n-1-l_n}.\tau_0^+$,
so by linearity we have
\begin{align*}
\lim_{y \to +\infty} \Big( y^{2h_{1,d_n}} \times
    \sF[v](x_1,\ldots,x_{n-1},y) \Big) 
= \; & \sF[\tau_0^+](x_1,\ldots,x_{n-1}) \times \sum_{l_n} A^+_{l_n} \times
    \sF^{(x_0)}\big[ \Wbas_{l_n}\tens\Wbas_{d_n-1-l_n}\big] (0,1) \\
= \; & \big( q-q^{-1} \big)^{d_n-1} \qfact{d_n-1}^2 \times B^{d_n,d_n}_1 \times \sF[\tau_0^+](x_1,\ldots,x_{n-1}) ,
\end{align*}
where in the last step we used the facts that 
$\sum_{l_n} A^+_{l_n} \times \Wbas_{l_n}\tens\Wbas_{d_n-1-l_n}
    = \big( q-q^{-1} \big)^{d_n-1} \qfact{d_n-1}^2 \times \tau^{(1;d_n,d_n)}_0$
in the notation of Equation~\eqref{eq: tensor product hwv}, 
and $\sF^{(x_0)}[\tau^{(1;d_n,d_n)}_0](0,1) = B^{d_n,d_n}_1$,
by Proposition~\ref{prop: asymptotics with subrepresentations}.
\end{proof}

\begin{rem}\label{rem: Mobius covariance by adding variables}
\emph{
Proposition~\ref{prop: the role of infinity in Mobius coveriant case} can be 
seen in two ways: it allows us to trade M\"obius covariance to dependence on one 
variable less. Directly by the statement, any M\"obius covariant function of $n$ 
variables may be viewed as a function of $n-1$ variables --- we can either get 
rid of the variable $x_1$ or $x_n$. Conversely, if we assume that
$v \in \bigotimes_{j=1}^n \Wd_{d_j}$ lies in a $d$-dimensional
irreducible subrepresentation, then we can view it as a M\"obius covariant
function of $n+1$ variables --- either associated to the vector
$R_+^{-1}(v) \in \Wd_d \tens \big( \bigotimes_{j=1}^n \Wd_{d_j} \big)$
or to the vector
$R_-^{-1}(v) \in \big( \bigotimes_{j=1}^n \Wd_{d_j} \big) \tens \Wd_d$,
and with the additional variable respectively to the right or to the left of all
other variables.
}
\end{rem}

\subsection{\label{sub: cyclic permutations}Cyclic permutations of variables}

Remark~\ref{rem: Mobius covariance by adding variables}
suggests yet another
interesting way to interpret the operations
$R_+$ and $R_-$ in Lemma~\ref{lem: isomorphism of hwv spaces}.
Namely, in view of 
Proposition~\ref{prop: the role of infinity in Mobius coveriant case},
the composed map $S = R_-^{-1} \circ R_+$ gives rise to cyclic
permutations of variables in our functions.

Consider the M\"obius covariant case, and denote
by $H_1(V)$ the 
maximal trivial subrepresentation of a representation $V$.
Let $v$ be in the trivial subrepresentation
$H_1(\Wd_{d_n} \tens \Wd_{d_{n-1}} \tens \cdots \tens \Wd_{d_1})$.
Then the vector $S(v)$ is in the trivial subrepresentation
$H_1(\Wd_{d_{n-1}} \tens \cdots \tens \Wd_{d_1} \tens \Wd_{d_n})$,
and the function $\sF[S(v)]$ could be thought of
as corresponding to the original function $\sF[v]$
when the variables are ordered as
$x_n < x_1 < x_2 < \cdots < x_{n-1}$.

When the number of variables is $n$, and the operation of moving the
rightmost variable to the left of all others is repeated $n$ times,
one expects to recover the original function.
Above we have defined this operation by its action on vectors in the
trivial subrepresentation of the $n$-fold tensor product as
\[ S^{(d_n)} \;\colon\; H_1 ( \Wd_{d_n} \tens \Wd_{d_{n-1}} \tens \cdots \tens \Wd_{d_1} )
    \to H_1 ( \Wd_{d_{n-1}} \tens \cdots \tens \Wd_{d_1} \tens \Wd_{d_n} ) , \]
where we now
emphasize that the definition of the operation $S^{(d_n)} = R_-^{-1} \circ R_+$
depends on the dimension $d_n$.
The correct $n$:th iterate is thus
\[ S^{(d_1)} \circ S^{(d_2)} \circ \cdots \circ S^{(d_{n-1})} \circ S^{(d_n)} 
\; \colon \;
    H_1 \Big( \bigotimes_{j=1}^n \Wd_{d_j} \Big) \to H_1 \Big( \bigotimes_{j=1}^n \Wd_{d_j} \Big) . \]
This mapping turns out to be not exactly the identity, but rather a constant
multiple of the identity.

To see that $S^{(d_n)} \circ \cdots \circ S^{(d_1)}$ is a constant multiple
of the identity on $H_1(\bigotimes_{j=1}^n \Wd_{d_j})$,
it is convenient to characterize the vectors in $H_1(\bigotimes_{j=1}^n \Wd_{d_j})$
by their projections to different subrepresentations.
One then needs two basic commutative diagrams, which express how to
exchange the order of the projections and the operations $S^{(d_j)}$.
Let
\[ \hat{\pi}^{(\diadim)}_{j,j+1} \;\colon\; \bigotimes_{i=1}^n \Wd_{d_j}
    \to \Big( \bigotimes_{i=j+2}^n \Wd_{d_j} \Big) \tens \Wd_\diadim \tens \Big( \bigotimes_{i=1}^{j-1} \Wd_{d_j} \Big) \]
denote the projection $\Wd_{d_{j+1}} \tens \Wd_{d_j} \to \Wd_\diadim$
acting in the $j$:th and $j+1$:st tensorands, as
in Section~\ref{sub: asymptotics with subrepresentations}.
If~$j < n-1$, then the projection obviously commutes with the
operation $S^{(d_n)}$, according to the following square diagram
\begin{align}\label{eq: square diagram} 
\xymatrix{
    H_1 ( \Wd_{d_n} \tens \Wd_{d_{n-1}} \tens \cdots 
        \tens \Wd_{d_1} ) \; \; \ar[r]^{\hspace{-.9cm}\hat{\pi}^{(\diadim)}_{j,j+1}}\ar[d]_{S^{(d_n)}}
  & \; \; H_1 ( \Wd_{d_n} \tens \Wd_{d_{n-1}} \tens \cdots \tens \Wd_{\diadim} \tens
        \cdots \tens \Wd_{d_1} )\ar[d]^{S^{(d_n)}} \\
    H_1 ( \Wd_{d_{n-1}} \tens \cdots 
        \tens \Wd_{d_1} \tens \Wd_{d_{n}} ) \; \; \ar[r]_{\hspace{-.7cm}\hat{\pi}^{(\diadim)}_{j+1,j+2}}
  & \; \; H_1 ( \Wd_{d_{n-1}} \tens \cdots \tens \Wd_{\diadim} \tens
        \cdots \tens \Wd_{d_1} \tens \Wd_{d_n} ) .
} 
\end{align}
If $j=n-1$, then a diagram of the above type does not make sense --- instead, we
apply two $S$-operations before the projection. The following pentagon diagram
commutes up to a multiplicative constant
\begin{align}\label{eq: pentagon diagram} 
\xymatrix{
    H_1 ( \Wd_{d_n} \tens \Wd_{d_{n-1}} \tens \Wd_{d_{n-2}} \tens \cdots \tens \Wd_{d_1} ) \; \; \ar[r]^{\hspace{.9cm}\hat{\pi}^{(\diadim)}_{n-1,n}}\ar[d]_{S^{(d_n)}}
  & \; \; H_1 ( \Wd_{\diadim} \tens \Wd_{d_{n-2}} \tens \cdots \tens \Wd_{d_1} )\ar[d]^{S^{(\diadim)}} \\
    H_1 ( \Wd_{d_{n-1}} \tens \Wd_{d_{n-2}} \tens \cdots \tens \Wd_{d_1} \tens \Wd_{d_n} ) \; \; \ar[d]_{S^{(d_{n-1})}}
  & \; \; H_1 ( \Wd_{d_{n-2}} \tens \cdots \tens \Wd_{d_1} \tens \Wd_{\diadim} ) \\
    H_1 ( \Wd_{d_{n-2}} \tens \cdots \tens \Wd_{d_1} \tens \Wd_{d_n} \tens \Wd_{d_{n-1}} ) \; \; \ar[ur]_{\quad \hat{\pi}^{(\diadim)}_{1,2}} .
  &  
} 
\end{align}
Using the above square and pentagon diagrams~\eqref{eq: square diagram}
and \eqref{eq: pentagon diagram}, one can finally show that also
the larger diagram~\eqref{eq: full commutative diagram}
below commutes up to constants.
We choose a sequence of intermediate dimensions
$\diadim_2,\diadim_3,\ldots,\diadim_{n-2},\diadim_{n-1}=d_{n}$,
to specify an element of the dual of $H_1(\bigotimes_{i=1}^n \Wd_{d_i})$
through a sequence of projections.
In this diagram, the top left $H_1$ stands for $H_1(\bigotimes_{i=1}^n \Wd_{d_i})$,
and moving downwards amounts to a cyclic permutation of the tensorands
by an $S$-operation,
and moving to the right reduces the number of tensorands by a projection.
In the rightmost column, only two tensorands remain and the spaces $H_1$
stand for $H_1(\Wd_{d_n} \tens \Wd_{d_n})$ --- note that $\diadim_{n-1}=d_n$,
as is needed for $\Wd_{d_n} \tens \Wd_{\diadim_{n-1}}$ to contain
a trivial subrepresentation
(see Lemma~\ref{lem: tensor product representations of quantum sl2}).
The diagram
\begin{align}\label{eq: full commutative diagram} 
\xymatrix{
    H_1 \ar[r]^{\hat{\pi}^{(\diadim_2)}_{1,2}}\ar[d]_{S^{(d_n)}}
  & H_1 \ar[r]^{\hat{\pi}^{(\diadim_3)}_{1,2}}\ar[d]
  & \;\; \cdots \;\; \ar[r]^{\hat{\pi}^{(\diadim_{n-2})}_{1,2}}
  & H_1 \ar[r]^{\hat{\pi}^{(\diadim_{n-1})}_{1,2}}\ar[d]
  & H_1 \ar[d]^{S^{(d_n)}} \\
    H_1 \ar[r]^{\hat{\pi}^{(\diadim_2)}_{2,3}} \ar[d]_{S^{(d_{n-1})}}
  & H_1 \ar[r]^{\hat{\pi}^{(\diadim_3)}_{2,3}}\ar[d]
  & \;\; \cdots \;\; \ar[r]^{\hat{\pi}^{(\diadim_{n-2})}_{2,3}}
  & H_1 \ar[r]^{\hat{\pi}^{(\diadim_{n-1})}_{2,3}}\ar[d]
  & H_1 \ar[d]^{S^{(\diadim_{n-1})}} \\
    H_1 \ar[r]^{\hat{\pi}^{(\diadim_2)}_{3,4}} \ar[d]_{S^{(d_{n-2})}}
  & H_1 \ar[r]^{\hat{\pi}^{(\diadim_3)}_{3,4}}\ar[d]
  & \;\; \cdots \;\; \ar[r]^{\hat{\pi}^{(\diadim_{n-2})}_{3,4}}
  & H_1 \ar[d]
  & H_1 \\
    \vdots \ar[d]
  & \vdots \ar[d]
  & \;\; \vdots \;\;
  & H_1 \ar[ur]_{\hat{\pi}^{(\diadim_{n-1})}_{1,2}}
  & \\
    H_1 \ar[r]^{\hat{\pi}^{(\diadim_2)}_{n-1,n}} \ar[d]_{S^{(d_{2})}}
  & H_1 \ar[d]
  & \;\; \Ddots \;\; \ar[ur]_{\hat{\pi}^{(\diadim_{n-2})}_{1,2}}
  & \\
    H_1 \ar[d]_{S^{(d_{1})}}
  & H_1 \ar[ur]_{\hat{\pi}^{(\diadim_3)}_{1,2}}
  & & \\
    H_1 \ar[ur]_{\hat{\pi}^{(\diadim_2)}_{1,2}} ,
  & & & \\
          & & & 
} 
\end{align}
commutes up to constants, and we then deduce that
\begin{align*}
\hat{\pi}^{(\diadim_{n-1})}_{1,2} \circ \cdots \circ \hat{\pi}^{(\diadim_2)}_{1,2}
= C \times \hat{\pi}^{(\diadim_{n-1})}_{1,2} \circ \cdots \circ \hat{\pi}^{(\diadim_2)}_{1,2}
    \circ S^{(d_1)} \circ \cdots \circ S^{(d_n)} .
\end{align*}
The constant $C$ above is independent of the sequence of dimensions
$\diadim_2, \ldots, \diadim_{n-1}$. Such projections span the dual of
$H_1(\bigotimes_{i=1}^n \Wd_{d_i})$, which allows us to conclude that
\begin{align*}
S^{(d_1)} \circ \cdots \circ S^{(d_n)} = C \times \id_{H_1(\bigotimes_{i=1}^n \Wd_{d_i})} .
\end{align*}
Thus, the $S$-operations give rise to a projective action of cyclic permutations
on our M\"obius covariant functions $\sF[v] \colon \chamber_n \to \bC$.
The constant is explicitly calculated
in~\cite[Corollary~4.3]{Peltola-basis_of_solutions_of_BSA_PDEs}:
${C = \prod_{i=1}^n (-q)^{1-d_i}}$.

\section{\label{sec: conclusions}Conclusions and outlook}

We have defined the spin chain - Coulomb gas correspondence,
which associates screened Coulomb gas correlation functions to vectors
in a tensor product of representations of the quantum group $\Uqsltwo$.
Natural representation theoretical properties of the vectors have
been shown to imply properties of the corresponding functions.

The results presented here are used in \cite{JJK-SLE_boundary_visits,KP-pure_partition_functions_of_multiple_SLEs}
to explicitly solve two interesting problems about $\SLE$s. Other conformally
covariant boundary correlation functions could be treated similarly
with the techniques of the present article. Also,
in \cite{FP-monodromy_invariant_correlations},
the spin chain - Coulomb gas correspondence is applied
to the 
construction of monodromy invariant bulk correlation functions of conformal 
field theory.

The results of the present article apply in the generic, semisimple case,
in which the deformation parameter $q$ of the quantum group is not
a root of unity. If $q$ were a root of unity --- that is, $\kappa\in\bQ$
--- the representation theory of the quantum group would become non-semisimple
and the corresponding functions would have degeneracies (exceptional
linear dependencies or poles as a function of $\kappa$). Extending
a version of the spin chain~- Coulomb gas correspondence to
these degenerate cases is a natural topic of future research --- some
of the non-semisimple representation theory has been analyzed in,
e.g., \cite{BFGT-Lusztig_limit_at_root_of_unity_and_fusion}, and
examples of degeneracies of the functions have been resolved in, e.g.,
\cite{FK-solution_space_for_a_system_of_null_state_PDEs_3,
FSK-multiple_SLE_connectivity_weights_for_rectangles_hexagons_and_octagons}.

The boundary correlation functions obtained in this correspondence
have conformal weights labeled by the first row of the Kac table.
This is sufficient for those applications that served as our primary
motivation, but the correspondence could possibly be generalized by
considering another type of screening charges, and an appropriate
``two-screening quantum group'' \cite{DF-multipoint_correlation_functions,Fuchs-affine_Lie_algebras_and_quantum_groups}.

\appendix 

\section{\label{app: contour manipulations}Contour manipulations}

In this appendix, we collect the proofs of intermediate results
that involve in principle straightforward but occasionally
lengthy contour deformation and branch choice calculations.

\begin{lem*}[Lemma~\ref{lem: simplex in terms of hypercube}]
The deformed hypercube integral function 
$\CubeInt_{m_{1},\ldots,m_{n}}^{(x_{0})}
$
defined in~\eqref{eq: deformed hypercube integral}
is related to the real integral function
$\SimplexInt_{m_{1},\ldots,m_{n}}^{(x_{0})}$ defined in~\eqref{eq: real integral over product of simplices}
by
\begin{align*}
\CubeInt_{m_{1},\ldots,m_{n}}^{(x_{0})}(\boldsymbol{x})= & \left(\prod_{i=1}^{n}q^{-\binom{m_{i}}{2}}\qfact{m_{i}}\right)\times\SimplexInt_{m_{1},\ldots,m_{n}}^{(x_{0})}(\boldsymbol{x}),\qquad\text{for }\boldsymbol{x}\in\chamber_{n}^{(x_{0})}.
\end{align*}
\end{lem*}
\begin{proof}
Note that $\SimplexInt$ is obtained by integration over the set $\sR_{m_{1},\ldots,m_{n}}$
which is a product of simplices, whereas $\CubeInt$ is obtained by
integration over the set $\widetilde{\sR}_{m_{1},\ldots,m_{n}}$ which
is a product of slightly deformed hypercubes. We split each of the
$m_{i}$-dimensional hypercubes to $m_{i}!$ simplices, and thus express
$\CubeInt_{m_{1},\ldots,m_{n}}^{(x_{0})}$ as a sum of $\prod_{i=1}^{n}(m_{i}!)$
terms, each of which is a phase factor times $\SimplexInt_{m_{1},\ldots,m_{n}}^{(x_{0})}$.

The variables $w_{r}$, $r\in I^{(i)}$, are integrated over one of
the slightly deformed hypercubes. We select the deformed hypercube
integration contour so that the variables are on the real axis, except
when the distance between some of the variables becomes smaller than
a chosen $\eps>0$. In view of the integrand, proportional
to~\eqref{eq: integrand with generic phase},
the contribution from cases with some $|w_{r}-x_{i}|<\eps$ is $\OO(\eps^{1-\frac{4}{\kappa}\max(d_{i}-1)})$,
and the further contribution from cases with some $|w_{s}-w_{r}|<\eps$
is $\OO(\eps^{1+\frac{8}{\kappa}})$. We may thus neglect these contributions,
which tend to zero as $\eps\searrow0$, and only consider cases with
variables $w_{r}$ on the real line in some definite order.

We encode the possible orderings of the variables $w_{r}$ by $n$-tuples
$(\sigma^{(1)},\ldots,\sigma^{(n)})$, where $\sigma^{(i)}$ is a
permutation of $I^{(i)}$: the associated order of the variables is
\begin{align*}
w_{\sigma^{(i)}(r)}<w_{\sigma^{(i)}(r')} & \qquad\text{for all }r,r'\in I^{(i)}\text{ such that }r<r',\\
w_{\sigma^{(i)}(r)}<w_{\sigma^{(j)}(r')} & \qquad\text{for all }r\in I^{(i)},
\, r'\in I^{(j)}\text{ such that }i<j.
\end{align*}
By our definition, the phase of $f_{m_{1},\ldots,m_{n}}^{\approx}$
is positive when the ordering of the variables is the one corresponding
to all identity permutations, $\sigma^{(i)}=\id_{I^{(i)}}$ for $i=1,\ldots,n$.
In the limit $\eps\searrow0$, the integration over the set where the
variables respect this standard ordering thus simply reproduces $\SimplexInt_{m_{1},\ldots,m_{n}}^{(x_{0})}$.
If the ordering among the variables $w_{r}$, $r\in I^{(i)}$, is given by some 
other permutation $\sigma^{(i)}$, then the phase factors accumulated from 
exchanging the orders of these variables is 
$q^{-2\times\#{\rm inv}(\sigma^{(i)})}$, where 
\begin{align*}
{\rm inv}(\sigma^{(i)})
=\; & \set{r,s\in I^{(i)} \;\Big|\; r<s\text{ and }\sigma^{(i)}(r)>\sigma^{(i)}(s)}
\end{align*}
denotes the set of inversions of $\sigma^{(i)}$. Apart from these
phase factors, the contribution of the integral from the set corresponding
to the ordering $(\sigma^{(1)},\ldots,\sigma^{(n)})$ coincides with
the integral over the standard ordered part. In conclusion, we have
\begin{align*}
\CubeInt_{m_{1},\ldots,m_{n}}^{(x_{0})}
=\; & \sum_{\sigma^{(1)},\ldots,\sigma^{(n)}}\left(\prod_{i=1}^{n}q^{-2\times\#{\rm inv}(\sigma^{(i)})}\right)\times\SimplexInt_{m_{1},\ldots,m_{n}}^{(x_{0})}.
\end{align*}
By Lemma~\ref{lem: q-combinatorics}(b) we simplify the prefactor
to the asserted form.
\end{proof}

\begin{lem*}[Lemma~\ref{lem: explicit formula for one point FW integral}]
In the case $n=1$, the basis function~\eqref{eq: FW basis integral} is related 
to the
real integral by
\begin{align*}
\FWint_{\ell}^{(x_{0})}(x)
=\; & \left(\qfact{\ell}\prod_{m=1}^{\ell}(q^{d-m}-q^{m-d})\right)
\times\SimplexInt_{\ell}^{(x_{0})}(x).
\end{align*}
In particular, $\FWint_{\ell}^{(x_{0})}$ is identically zero if $\ell\geq d$.
\end{lem*}
\begin{figure}
\subfigure[The proof of Lemma~\ref{lem: explicit formula for one point FW integral} consists of manipulating the integral of $\FWint^{(x_0)}_\ell(x)$, which has $\ell$  non-intersecting positively oriented loops around $x$, anchored at $x_0$, as illustrated in this figure.]{
\includegraphics[width=0.4\textwidth]{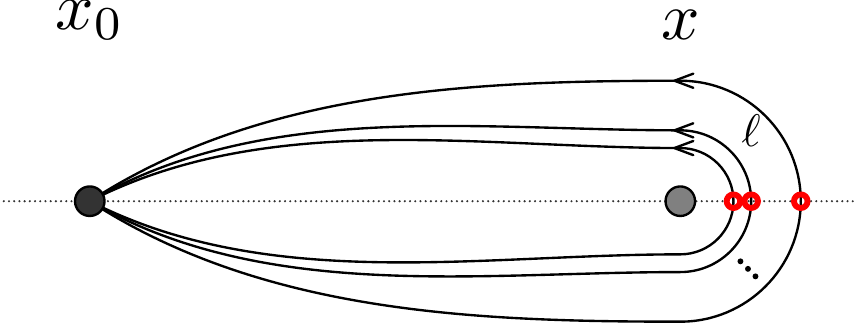}
\label{sfig: the full basis integral with one point}}\\
\bigskip{}
\bigskip{}
\bigskip{}
\bigskip{}

\subfigure[First, the integration over the innermost loop $\mathcal{C}_1$ --- corresponding to the variable $w_1$ --- is rewritten according to the illustrations in this figure. The small circular arc can be neglected. Two pieces between $x_0$ and $x$ remain, and their contributions are the same up to orientation and phase factors.]{
\includegraphics[width=0.9\textwidth]{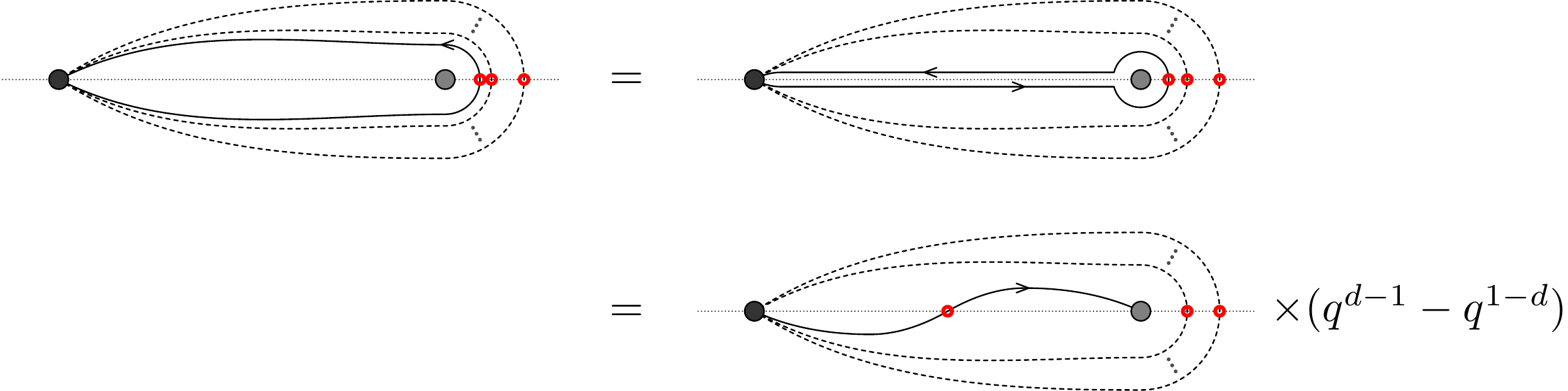}
\label{sfig: deforming the innermost loop}}\\
\bigskip{}
\bigskip{}
\bigskip{}
\bigskip{}

\subfigure[Once the integration over $w_1$ has been rewritten as an integral from $x_0$ to $x$, we start manipulating the integration over the next loop $\mathcal{C}_2$ --- corresponding to the variable $w_2$. The rewriting is illustrated in this figure. The result consists of two pieces, and upon relabeling the dummy integration variables $w_1$ and $w_2$ in one of them, the contributions are seen to be the same up to phase factors and signs.]{\includegraphics[width=0.9\textwidth]{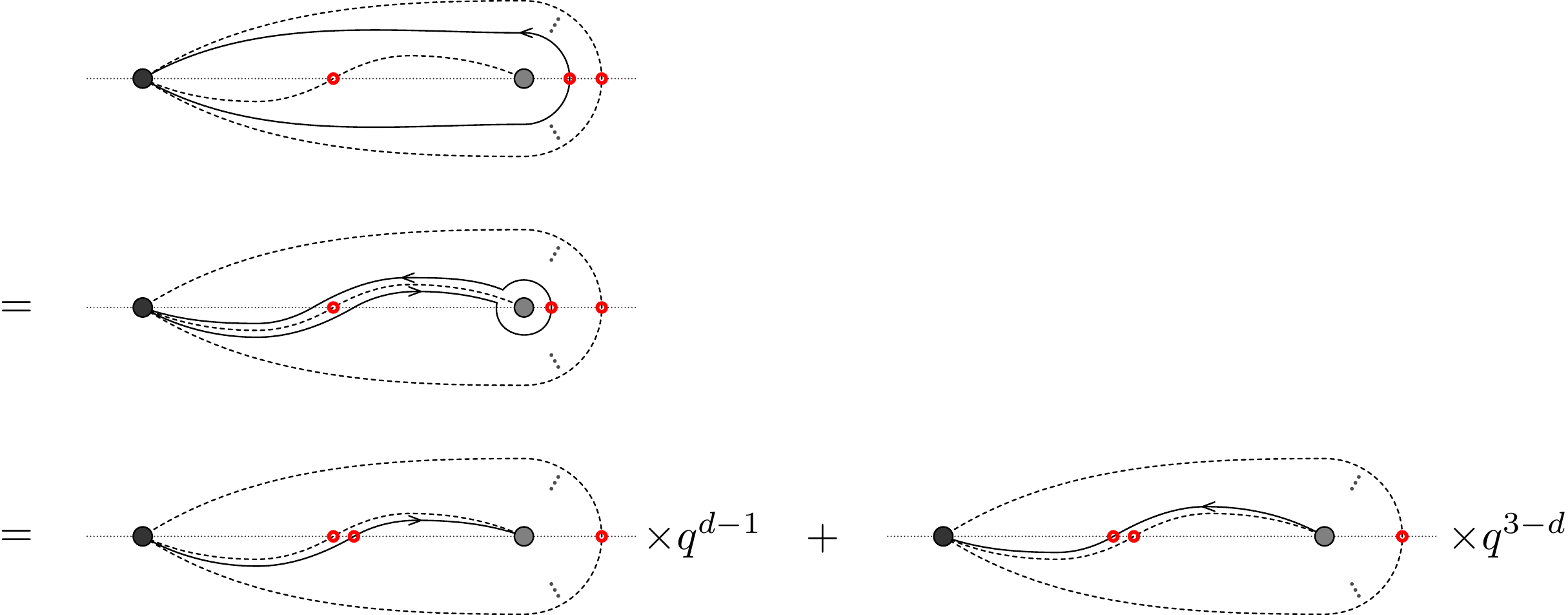}
\label{sfig: deforming the second loop}
}

\caption{\label{fig: decomposing basis function to hypercube integral}
Illustrations
for the proof of Lemma~\ref{lem: explicit formula for one point FW integral}.}
\end{figure}
\begin{proof}
We will prove that $\FWint_{\ell}^{(x_{0})}(x)=q^{\binom{\ell}{2}}\left(\prod_{m=1}^{\ell}(q^{d-m}-q^{m-d})\right)\times\CubeInt_{\ell}^{(x_{0})}(x)$
and the asserted equality will then follow from 
Lemma~\ref{lem: simplex in terms of hypercube}.
We achieve this by decomposing each of the loops based at $x_{0}$
encircling $x$ to two pieces: one from $x_{0}$ to $x$ and the other
from $x$ back to $x_{0}$. We begin from the innermost loop and inductively
proceed outwards. The procedure is illustrated in 
Figure~\ref{fig: decomposing basis function to hypercube integral}.

Let $\sC_{1},\ldots,\sC_{\ell}$ be the contours of integration of
the variables $w_{1},\ldots,w_{\ell}$, so that $\sC_{1}$ is a loop
encircling $x$, and for $r>1$ the loop $\sC_{r}$ encircles the
entire loop $\sC_{r-1}$. The basis function then reads
\begin{align*}
\FWint_{\ell}^{(x_{0})}(x)
=\; & \int_{\sC_{\ell}}\ud w_{\ell}\cdots\int_{\sC_{2}}\ud w_{2}
\int_{\sC_{1}}\ud w_{1}\;\fFW_{\ell}(x;w_{1},\ldots,w_{\ell}),
\end{align*}
see also Figure~\ref{sfig: the full basis integral with one point}.
Recall that the branch and phase of the integrand are chosen so that
\begin{align*}
\fFW_{\ell}(x;w_{1},\ldots,w_{\ell})
=\; & \prod_{r=1}^{\ell}(w_{r}-x)^{\frac{4}{\kappa}(1-d)}
\prod_{1\leq r<s\leq\ell}(w_{s}-w_{r})^{\frac{8}{\kappa}}
\end{align*}
is positive at the ``midpoint'' $\boldsymbol{w'}$ of the loops (illustrated
by the red dots in the figure). We may thus fix the branches of each
of the factors above so that they are positive when $x<w_{1}<w_{2}<\cdots<w_{\ell}$.

Choose a small $\eps>0$. To decompose the first loop, that is, the
innermost integral
\begin{align*}
 & \int_{\sC_{1}}\fFW_{\ell}\,\ud w_{1},
\end{align*}
deform it to a line segment along the real axis from $x_{0}$ to $x-\eps$,
a circle of radius $\eps$ around $x$, and a line segment along the
real axis from $x-\eps$ to $x_{0}$. The innermost integral then
becomes essentially
\begin{align*}
 & \int_{x_{0}}^{x-\eps}\fFW_{\ell}\,\ud w_{1}+\oint_{\bdry B_{\eps}(x)}\fFW_{\ell}\,\ud w_{1}+\int_{x-\eps}^{x_{0}}\fFW_{\ell}\,\ud w_{1},
\end{align*}
but in this expression we have abused notation and hidden the important
phase factors. In particular, the first and last terms do not cancel
--- the factor $(x-w_{1})^{\frac{4}{\kappa}(1-d)}$ in the integrand
$\fFW_{\ell}$ has a constant but different phase on these two line
segments. Indeed, on the first line segment 
$(x-w_{1})^{\frac{4}{\kappa}(1-d)} = q^{d-1}\,|x-w_{1}|^{\frac{4}{\kappa}(1-d)}$
and on the second 
$(x-w_{1})^{\frac{4}{\kappa}(1-d)} = q^{1-d}\,|x-w_{1}|^{\frac{4}{\kappa}(1-d)}$,
because from the reference point $\boldsymbol{w}'$ we must take $w_{1}$
half a turn around $x$ in the negative or positive direction, respectively.
The contribution of the integral around the circle is proportional
to $\eps^{1-\frac{4}{\kappa}(d-1)}$ and it thus vanishes in the limit
$\eps\searrow0$ when $\kappa>4(d-1)$. Hence, we can write
\begin{align*}
\int_{\sC_{1}}\ud w_{1}\;\fFW_{\ell}(x;w_{1},\ldots,w_{\ell})
=\; & (q^{d-1}-q^{1-d})\int_{x_{0}}^{x}\ud w_{1}
\;|x-w_{1}|^{\frac{4}{\kappa}(1-d)}
\prod_{r=2}^{\ell}(w_{r}-x)^{\frac{4}{\kappa}(1-d)}
\prod_{1\leq r<s\leq\ell}(w_{s}-w_{r})^{\frac{8}{\kappa}}.
\end{align*}
These steps are illustrated in Figure~\ref{sfig: deforming the innermost loop}.

We proceed similarly with the integration contours $\sC_{2},\ldots,\sC_{\ell}$,
cutting each of them into two pieces between $x_{0}$ and $x$, and a negligible
loop around $x$. The example of the second loop $\sC_{2}$ is illustrated
in Figure~\ref{sfig: deforming the second loop}. On each piece of
the $r$:th loop $\sC_{r}$, we rephase the factors of the integrand
that contain $w_{r}$, in order to finally compare with the integral
$\CubeInt$. The first piece of the integration contour $\sC_{r}$
is a path from $x_{0}$ to $x$ below the variables $w_{s}$, for
$s<r$, and we extract a phase factor $q^{d-1}$ resulting from taking
$w_{r}$ half a turn around $x$ in the negative direction. The second
piece of the integration contour $\sC_{r}$ is a path from $x$ to
$x_{0}$ above the variables $w_{s}$, for $s<r$, and we not only
take $w_{r}$ in the positive direction around $x$, but we also take
it positively around all $w_{s}$, $s<r$, in order to reach a position
where the corresponding piece of the contour remains in the subset
$\widetilde{\Wchamber}_{\ell}$. The phase factor accumulated in this
case is $q^{1-d+2(r-1)}$. After extracting all these phase factors,
the remaining integral is, up to an orientation,
equal to the integral $\CubeInt_{\ell}^{(x_{0})}(x)$. 
The final result is
\begin{align*}
\FWint_{\ell}^{(x_{0})}(x)
=\; & \prod_{r=1}^{\ell}(q^{d-1}-q^{1-d+2(r-1)})
\times\int_{\widetilde{\sR}_{\ell}}\fCube_{\ell}(x;w_{1},\ldots,w_{\ell})
\;\ud w_{1}\cdots\ud w_{\ell},
\end{align*}
and we finish the proof by simplifying the prefactor
\begin{align*}
\prod_{r=1}^{\ell}(q^{d-1}-q^{1-d+2(r-1)})
=\; & q^{\binom{\ell}{2}}\prod_{m=1}^{\ell}(q^{d-m}-q^{m-d})
\end{align*}
and using Lemma~\ref{lem: simplex in terms of hypercube}. If $\ell\geq d$,
the product contains a factor which vanishes, and then $\FWint_{\ell}^{(x_{0})}(x)=0$.
\end{proof}

\begin{lem*}[Lemma~\ref{lem: two variable FW basis function}]
We have 
\begin{align*}
\FWint_{l_{1},l_{2}}^{(x_{0})}(x_{1},x_{2})
=\; & q^{\binom{l_{1}}{2}+\binom{l_{2}}{2}}(q-q^{-1})^{l_{1}+l_{2}}\frac{\qfact{d_{1}-1}\qfact{d_{2}-1}}{\qfact{d_{1}-l_{1}-1}\qfact{d_{2}-l_{2}-1}} \\
& \qquad \times \sum_{m=0}^{l_{2}}q^{m(m-l_{2}+d_{1}-1)}\qbin{l_{2}}m\CubeInt_{l_{1}+m,l_{2}-m}^{(x_{0})}(x_{1},x_{2}).
\end{align*}
\end{lem*}
\begin{proof}
We use the method of the proof of 
Lemma~\ref{lem: explicit formula for one point FW integral}
to decompose the $l_{1}$ loops around $x_{1}$ and the $l_{2}$ loops
around $x_{2}$ to paths from $x_{0}$ to $x_{1}$ and $x_{2}$, respectively.
We obtain
\begin{align*}
\FWint_{l_{1},l_{2}}^{(x_{0})}(x_{1},x_{2})
=\; & \prod_{i=1,2}\left(q^{\binom{l_{i}}{2}}\prod_{t=1}^{l_{i}}(q^{d_{i}-t}-q^{t-d_{i}})\right)\times\DecoInt_{l_{1},l_{2}}^{(x_{0})}(x_{1},x_{2}),
\end{align*}
where $\DecoInt_{l_{1},l_{2}}^{(x_{0})}$ is the function defined
by an integral as in Figure~\ref{fig: decomposition auxiliary integral}.
The variables $w_{1},\ldots,w_{l_{1}}$ are integrated from $x_{0}$
to $x_{1}$ and the variables $w_{l_{1}+1},\ldots,w_{l_{1}+l_{2}}$
from $x_{0}$ to $x_{2}$ in such a way that for all $r<r'$, the
path of the variable $w_{r'}$ remains below the path of the variable
$w_{r}$. The integrand is rephased so that it is positive in the
region
\begin{align*}
 & x_{0}<w_{1}<w_{2}<\cdots<w_{l_{1}}<x_{1}<w_{l_{1}+1}<\cdots<w_{l_{1}+l_{2}}<x_{2}.
\end{align*}
We further split the integration contours of the $l_{2}$ variables
$w_{l_{1}+1},\ldots,w_{l_{1}+l_{2}}$ into two pieces: the first from
$x_{0}$ to $x_{1}$ and the second from $x_{1}$ to $x_{2}$. A contribution
proportional to $\CubeInt_{l_{1}+m,l_{2}-m}^{(x_{0})}$ is obtained
whenever we make the first choice for some $m$ of the $l_{2}$ variables
--- we again just have to keep track of the correct phase factors.
Suppose that we make the first choice for the variables
\begin{align*}
w_{l_{1}+s_{1}},w_{l_{1}+s_{2}},\ldots,w_{l_{1}+s_{m}} & \qquad(1\leq s_{1}<s_{2}<\cdots<s_{m}\leq l_{2})
\end{align*}
and the second choice for the rest. Then, when taking the variables
from the point where the integrand of $\DecoInt_{l_{1},l_{2}}^{(x_{0})}$
is positive to the point where the integrand of $\CubeInt_{l_{1}+m,l_{2}-m}^{(x_{0})}$
is positive, we accumulate some phase. Indeed, first of all, each
$w_{l_{1}+s_{p}}$ goes half a negative turn around $x_{1}$, contributing
a factor $q^{d_{1}-1}$. Moreover, each $w_{l_{1}+s_{p}}$ goes half
a negative turn around each $w_{r}$ with $l_{1}<r<l_{1}+s_{p}$,
for $r\neq l_{1}+s_{p'}$, contributing a phase factor $q^{-2(s_{p}-p)}$.
In the end, all the contributions to $\DecoInt_{l_{1},l_{2}}^{(x_{0})}$
with fixed $m$ give
\begin{align*}
 & q^{(d_{1}-1)m}\times\sum_{1\leq s_{1}<s_{2}<\cdots<s_{m}\leq l_{2}}q^{-2\sum_{p}(s_{p}-p)}\times\CubeInt_{l_{1}+m,l_{2}-m}^{(x_{0})}.
\end{align*}
By Lemma~\ref{lem: q-combinatorics}(c) we simplify the sum of the
prefactors to the form $q^{-m(l_{2}-m)}\qbin{l_{2}}m$, and obtain
\begin{align*}
\DecoInt_{l_{1},l_{2}}^{(x_{0})}
=\; & \sum_{m=0}^{l_{2}}q^{(d_{1}-1+m-l_{2})m}\qbin{l_{2}}m
\times\CubeInt_{l_{1}+m,l_{2}-m}^{(x_{0})}.
\end{align*}
We finally write
\begin{align*}
q^{\binom{l_{i}}{2}}\prod_{t=1}^{l_{i}}(q^{d_{i}-t}-q^{t-d_{i}})
=\; & q^{\binom{l_{i}}{2}}(q-q^{-1})^{l_{i}}\frac{\qfact{d_{i}-1}}{\qfact{d_{i}-l_{i}-1}},
\end{align*}
and the expression for $\FWint_{l_{1},l_{2}}^{(x_{0})}$ then takes
the form
\begin{align*}
\FWint_{l_{1},l_{2}}^{(x_{0})}(x_{1},x_{2})
=\; & q^{\binom{l_{1}}{2}+\binom{l_{2}}{2}}(q-q^{-1})^{l_{1}+l_{2}}\frac{\qfact{d_{1}-1}\qfact{d_{2}-1}}{\qfact{d_{1}-l_{1}-1}\qfact{d_{2}-l_{2}-1}}\\
 & \qquad\times\sum_{m=0}^{l_{2}}q^{(d_{1}-1+m-l_{2})m}\qbin{l_{2}}m
 \times\CubeInt_{l_{1}+m,l_{2}-m}^{(x_{0})}(x_{1},x_{2}).
\end{align*}
This finishes the proof.
\end{proof}
\noindent 
\begin{figure}
\includegraphics[width=0.5\textwidth]{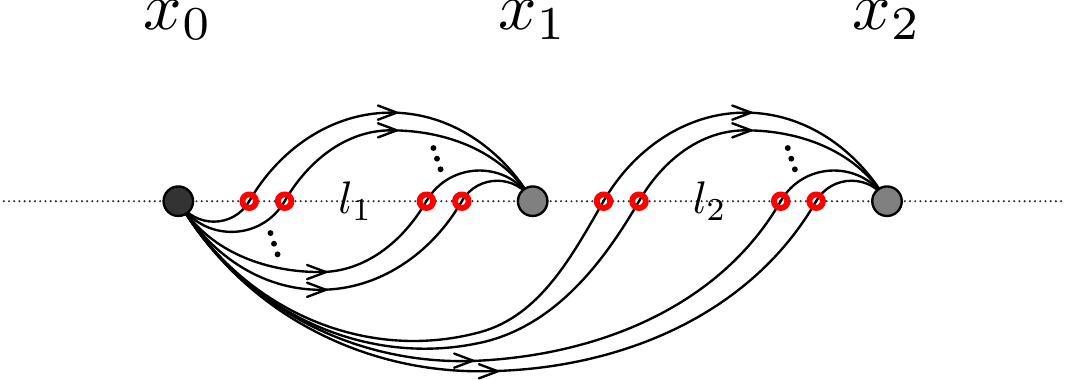}
\caption{\label{fig: decomposition auxiliary integral}
The integration surface
for the function $\DecoInt_{l_{1},l_{2}}^{(x_{0})}(x_{1},x_{2})$
used for the decomposition in the proof of 
Lemma~\ref{lem: two variable FW basis function}.
The integrand is rephased to be real and positive at the point $x_{0}<w_{1}<\cdots<w_{l_{1}}<x_{1}<w_{l_{1}+1}<\cdots<w_{l_{1}+l_{2}}<x_{2}$,
marked by the red circles.}
\end{figure}

\begin{lem*}[Lemma~\ref{lem: general FW basis function}]
We have
\begin{align*}
\FWint_{l_{1},\ldots,l_{n}}^{(x_{0})}(\boldsymbol{x})
=\; & \sum_{m_{1},\ldots,m_{n}}C_{l_{1},\ldots,l_{n}}^{m_{1},\ldots,m_{n}}\times\CubeInt_{m_{1},\ldots,m_{n}}^{(x_{0})}(\boldsymbol{x}),
\end{align*}
with some coefficients $C_{l_{1},\ldots,l_{n}}^{m_{1},\ldots,m_{n}}$,
which are zero unless $\sum_i l_{i}=\sum_i m_{i}$ and $\sum_{i=1}^{j}l_{i}\leq\sum_{i=1}^{j}m_{i}$
for all $j = 1 , \ldots, n$.
\end{lem*}
\begin{proof}
We only give a rough outline, and leave the details to the reader.
By a generalization of the method used in the proof of 
Lemma~\ref{lem: two variable FW basis function},
one can show that
\begin{align*}
\FWint_{l_{1},\ldots,l_{n}}^{(x_{0})}
= \; & \prod_{i=1}^{n}q^{\binom{l_{i}}{2}}\prod_{t=1}^{l_{i}}(q^{d_{i}-t}-q^{t-d_{i}}) \\
& \qquad \times 
\sum_{\substack{(k_{j}^{(i)})_{1\leq i\leq n;\,1\leq j\leq i}\\
\forall i:\;\sum_{j=1}^{i}k_{j}^{(i)}=l_{i}}}
C_{l_{1},\ldots,l_{n}}\left((k_{j}^{(i)})_{1\leq i\leq n;\,1\leq j\leq i}\right)
\times\CubeInt_{\sum_{i\geq1}k_{1}^{(i)},\sum_{i\geq2}k_{2}^{(i)},\ldots,k_{n}^{(n)}}^{(x_{0})},
\end{align*}
where
\begin{align*}
& C_{l_{1},\ldots,l_{n}}\left((k_{j}^{(i)})_{1\leq i\leq n;\,1\leq j\leq i}\right)\\
=\; & \prod_{i} \qbin{l_{i}}{k_{1}^{(i)};k_{2}^{(i)};\cdots;k_{i}^{(i)}}
q^{-\sum_{j'<j}k_{j}^{(i)}k_{j'}^{(i)}}\prod_{i<i'}q^{-2\sum_{j'<j}k_{j}^{(i)}k_{j'}^{(i')}}\prod_{i<i'}q^{\sum_{j\leq i}(d_{i}-1)k_{j}^{(i')}},
\end{align*}
and where we have used the following $q$-multinomial coefficients
\begin{align*}
\qbin{\ell}{k_{1};k_{2};\cdots;k_{i}}
=\; & \frac{\qfact{\ell}}{\prod_{j=1}^{i}\qfact{k_{j}}}.
\end{align*}
The desired coefficients of $\CubeInt_{m_{1},\ldots,m_{n}}$ are then
expressible as sums of the coefficients above,
\begin{align*}
C_{l_{1},\ldots,l_{n}}^{m_{1},\ldots,m_{n}}
=\; & \prod_{i=1}^{n}q^{\binom{l_{i}}{2}}\prod_{t=1}^{l_{i}}(q^{d_{i}-t}-q^{t-d_{i}})
\times\sum_{\substack{(k_{j}^{(i)}) \\ 
\forall j:\;\sum_{i\geq j}k_{j}^{(i)}=m_{j}}}
C_{l_{1},\ldots,l_{n}}\left((k_{j}^{(i)})_{1\leq i\leq n;\,1\leq j\leq i}\right).
\end{align*}
\end{proof}

\begin{lem*}[Lemma~\ref{lem: vanishing for too many bags}]
Whenever $l_{i}\geq d_{i}$
for some $i=1,\ldots,n$, we have
\begin{align*}
\FWint_{l_{1},\ldots,l_{n}}^{(x_{0})}(x_{1},\ldots,x_{n})\equiv\; & 0.
\end{align*}
\end{lem*}
\begin{proof}
We may rearrange the integrations over those $l_{i}$ $w$-variables
which encircle $x_{i}$ similarly as in the case of 
Lemma~\ref{lem: explicit formula for one point FW integral}.
After this rearrangement, the result of the whole integral $\FWint_{l_{1},\ldots,l_{n}}^{(x_{0})}(\boldsymbol{x})$
is the factor $\qfact{\ell_{i}} \, \prod_{m=1}^{l_{i}}(q^{d_{i}-m}-q^{m-d_{i}})$
times an integral which is convergent for large $\kappa$. The result
thus again vanishes if $l_{i}\geq d_{i}$ and $\kappa$ is large,
and by analyticity in $\kappa$, the same conclusion is valid for all
values of $\kappa>0$.
\end{proof}

\section{\label{app: exact form lemma}Differential operators acting on the integrand}

In this appendix, we outline the proof of Lemma~\ref{lem: BSA diff op for vertex operator corr fn},
which is used to show, roughly speaking, that the differential operators
acting on our integrands produce exact forms.

\begin{lem*}[Lemma~\ref{lem: BSA diff op for vertex operator corr fn}]
Let $d_1, \ldots, d_n$ be real numbers. The function 
\begin{align*}
f^{(0)}(x_{1},\ldots,x_{n})
=\; & \prod_{1\leq i<j\leq n}(x_{j}-x_{i})^{\frac{2}{\kappa}(d_{i}-1)(d_{j}-1)}
\end{align*}
satisfies the partial differential equation
$\sD_{d_{j}}^{(j)}f^{(0)}  =  0$, for all $j=1,\ldots,n$ such that $d_j$ is a 
positive integer.
\end{lem*}
\begin{proof}
If $d_j=1$, the claim $\sD^{(j)}_{1} f^{(0)}(x_1 , \ldots , x_n) = \pder{x_j} 
f(x_1 , \ldots , x_n) = 0 $ is obvious:
the function is constant in the variable $x_j$.

In the case $d_j=2$, the claim $\sD^{(j)}_{2} f^{(0)}(x_1 , \ldots , x_n) = 0$ 
has been verified by an
explicit calculation in~\cite[Lemma~5.1]{Kytola-local_mgales}.
We take this special case as our starting point, and proceed by a recursive argument,
which employs fusion similar to~\cite[Theorem~15]{Dubedat-fusion}.

We perform an induction on the number
\[ D = \sum_{\substack{j \\ d_j \in \set{3,4,\ldots}}} (d_j - 2) \in \bZ_{\geq 0} . \]
The base case $D=0$ is covered by the initial observations above.
In the induction step, we want to increment $D$ by one, which amounts to being able to
increment one of the parameters that are integers greater than or equal to $2$.
By permutation symmetry, we may assume that $d_1 = d$ is the parameter we wish to increment.
We may also take $d_2 = 2$, which will allow us to get to $d+1$ by fusion.

Hence, fix parameters $d_1 , \ldots, d_n$ such that
\[ d_1 = d \in \set{2,3,\ldots}
   \qquad \text{ and } \qquad
   d_2 = 2 , \]
and assume that the claim has been verified for the corresponding $f^{(0)}$.
For notational convenience, we will denote the two first variables by $x_1 = z$ and $x_2 = y$.
The fusion is based on the following obvious asymptotics of $f^{(0)}$.
If we denote 
\[ \Delta = h_{1,d+1} - h_{1,d} - h_{1,2} = \frac{2}{\kappa} (d-1) , \] 
then the function defined by
\begin{align}\label{eq: limit defining g0}
g_0(z , x_3, \ldots , x_n) = \lim_{y \to z} \frac{1}{(z-y)^\Delta} f^{(0)}(z,y,x_3,\ldots,x_n)
\end{align}
is of the same type as $f^{(0)}$, but with the new parameter sequence $d+1 , d_3 , \ldots , d_n$,
simply because
\[ (x_j - z)^{\frac{2}{\kappa}(d-1)(d_{j}-1)} \; (x_j - y)^{\frac{2}{\kappa}(d_{j}-1)}
    \; \underset{y \to z}{\longrightarrow} \; (x_j - z)^{\frac{2}{\kappa} d (d_{j}-1)} .\]
Our goal is to verify the asserted partial differential equations for this $g_0$,
assuming the ones on $f^{(0)}$.
The most involved among them is the partial differential equation of order $d+1$ given by
\begin{align}\label{eq: diff eqn goal}
\sD_{d+1} \; g_0 (z , x_3 , \ldots , x_n) = \; & 0 ,
\end{align}
where
\begin{align}\label{eq: diff op goal}
\sD_{d+1} = \; & \sum_{k=1}^{d+1}\sum_{\substack{p_{1},\ldots,p_{k}\geq1\\p_{1}+\cdots+p_{k}=d+1}}
    \frac{(-4/\kappa)^{d+1-k}\,d!^{2}}{\prod_{u=1}^{k-1}(\sum_{i=1}^{u}p_{i})(\sum_{i=u+1}^{k}p_{i})} \times \sL_{-p_{1}} \cdots \sL_{-p_{k}} 
\end{align}
and
\begin{align*}
\sL_{-p} = \; & - \sum_{i=3}^n (x_i-z)^{1-p} \pder{x_i} - (1-p) \sum_{i=3}^n (x_i-z)^{-p} \, h_{1,d_i} .
\end{align*}

In order to verify~\eqref{eq: diff eqn goal}, we use the 
induction assumption that $f^{(0)}$ satisfies, in particular, the following two differential equations
respectively associated to the points $z$ and $y$:
\begin{align}
\label{eq: assumed Dhat}
\widehat{\sD}_{d} \; f^{(0)} (z , y , x_3 , \ldots , x_n) = \; & 0
    \\
\label{eq: assumed Dtilde}
\widetilde{\sD}_{2} \; f^{(0)} (z , y , x_3 , \ldots , x_n) = \; & 0
    ,
\end{align}
where
\begin{align*}
\widehat{\sD}_{d} = \; & \sum_{k=1}^{d}\sum_{\substack{p_{1},\ldots,p_{k}\geq1\\p_{1}+\cdots+p_{k}=d}}
    \frac{(-4/\kappa)^{d-k}\,(d-1)!^{2}}{\prod_{u=1}^{k-1}(\sum_{i=1}^{u}p_{i})(\sum_{i=u+1}^{k}p_{i})} \times \widehat{\sL}_{-p_{1}} \cdots \widehat{\sL}_{-p_{k}} \\
\widetilde{\sD}_{2} = \; & \widetilde{\sL}_{-1}^{\; 2} - \frac{4}{\kappa} \widetilde{\sL}_{-2} ,
\end{align*}
with
\begin{align*}
\widehat{\sL}_{-p} = \; & - (y-z)^{1-p} \pder{y} - (1-p) (y-z)^{-p} h_{1,2} - \sum_{i=3}^n (x_i-z)^{1-p} \pder{x_i} - (1-p) \sum_{i=3}^n (x_i-z)^{-p} \, h_{1,d_i}  \\
\widetilde{\sL}_{-p} = \; & - (z-y)^{1-p} \pder{z} - (1-p) (z-y)^{-p} h_{1,d} - \sum_{i=3}^n (x_i-y)^{1-p} \pder{x_i} - (1-p) \sum_{i=3}^n (x_i-y)^{-p} \, h_{1,d_i} .
\end{align*}

To study the limit function $g_0$ defined by~\eqref{eq: limit defining g0},
we perform an expansion in $\eps = y-z$ around $\eps = 0$. We thus write first of all
\begin{align}\label{eq: expansion of f0}
f^{(0)}(z,y,x_3,\ldots,x_n) = \sum_{m=0}^\infty \eps^{\Delta + m} \, g_m(z, x_3 , \ldots, x_n) .
\end{align}
Moreover, we write
\begin{align*}
\widehat{\sL}_{-p} = \; & - \eps^{1-p} \pder{\eps} - (1-p) \eps^{-p} \, h_{1,2} + \sL_{-p} ,
\end{align*}
and note that the action of $\widetilde{\sL}_{-1}$ and $\widetilde{\sL}_{-2}$ on the translation
invariant function $f^{(0)}$ is
\begin{align*}
\widehat{\sL}_{-1} = \pder{\eps}
\qquad \text{ and } \qquad
\widehat{\sL}_{-2} = \sum_{s=-1}^{\infty} \eps^{s} \, \sL_{-2-s} - \eps^{-1} \pder{\eps} + \eps^{-2} \, h_{1,d} .
\end{align*}
With these, it is straightforward to calculate $\widetilde{\sD}_{2} \; f^{(0)}$ in a series expansion in $\eps$,
and after some simplifications the result is
\begin{align*}
\widetilde{\sD}_{2} \; f^{(0)} = \sum_{m=0}^\infty \eps^{\Delta + m -2} \Big( r(\Delta+m) \, g_m
    - \frac{4}{\kappa} \sum_{p'=1}^m \sL_{-p'} \, g_{m-p'} \Big) ,
\end{align*}
where $r(\alpha) := \alpha^2 - \alpha + \frac{4}{\kappa}\alpha - \frac{4}{\kappa} h_{1,d}$.
According to equation~\eqref{eq: assumed Dhat}, this series must vanish term by term.
Since also $r(\Delta+m) \neq 0$ for all $m = 1, 2, \ldots$, we can derive that
\begin{align*}
g_m = \sR_m \, g_0 ,
\end{align*}
where
\begin{align*}
\sR_m = \sum_{u=1}^m \sum_{\substack{p'_{1},\ldots,p'_{u}\geq1\\p'_{1}+\cdots+p'_{k}=m}}
    \frac{(4/\kappa)^u}{\prod_{j=1}^u r(\Delta+\sum_{i=j}^u p'_i) } \sL_{-p'_1} \cdots \sL_{-p'_u} \; g_0 .
\end{align*}

The next step is to study the expression $\widehat{D}_d \, f^{(0)}$,
as a series expanded in $\eps$.
Using the expression derived above, we get
\begin{align*}
\widehat{\sD}_d \, f^{(0)} = \; & \sum_{m=0}^\infty \widehat{\sD}_d \Big( \eps^{\Delta + m} \; (\sR_m \, g_0) \Big)
    = \sum_{m=0}^\infty ( \sP_m \, g_0 ) \; \eps^{\Delta - d + m} ,
\end{align*}
where $\sP_m$ is a polynomial in the operators $\sL_{-p}$, for $p = 1,\ldots$.
By assumption~\eqref{eq: assumed Dtilde}, the series above must vanish term by term, i.e., $\sP_m \, g_0 = 0$
for all $m$.
From the argument in the proof of~\cite[Lemma~1]{Dubedat-fusion} it follows that
the polynomial $\sP_{d+1}$ giving the coefficient of $\eps^{\Delta + 1}$ is a non-zero multiple of
the differential operator $\sD_{d+1}$ defined by~\eqref{eq: diff op goal}.
We conclude that
$g_0$ satisfies the differential equation~\eqref{eq: diff eqn goal} of order $d+1$.

The other asserted differential equations for $g_0$ are easier to verify.
The key is to still use the expansion~\eqref{eq: expansion of f0} in $\eps$, and
note that in the first order differential operators $\sL^{(j)}_{-p}$
acting on such an expansion, the following terms combine
\begin{align*}
& \Big\{ (z-x_j)^{1-p} \pder{z} + (1-p) (z-x_j)^{-p} \, h_{1,d} 
        + (y-x_j)^{1-p} \pder{y} + (1-p) (y-x_j)^{-p} \, h_{1,2} \Big\} 
            \sum_{m=0}^\infty \eps^{\Delta + m} \; g_m \\
= \; & \eps^{\Delta} \times \Big( (z-x_j)^{1-p} \pder{z} + (1-p) (z-x_j)^{-p} \, \big( h_{1,d} + h_{1,2} + \Delta \big) \Big) \, g_0
    + \OO(\eps^{\Delta+1}) .
\end{align*}
Moreover, here we have $h_{1,d} + h_{1,2} + \Delta = h_{1,d+1}$. With these observations,
it is routine to check that from the induction assumption that $f^{(0)}$ satisfies $\sD^{(j)}_{d_j} \, f^{(0)} = 0$
it follows that $g_0$ satisfies the corresponding partial differential equation of order $d_j$.
This finishes the proof.
\end{proof}

\section{\label{app: alternative conventions}Alternative conventions}

The main topic of the present article was the construction of 
the spin chain - Coulomb gas correspondence maps
\[ \sF^{(x_0)} \colon \Wd_{d_n} \tens \cdots \tens \Wd_{d_1}
    \to \mathcal{C}^\infty \big( \chamber^{(x_0)}_n \big) , \]
for which the representation theoretic properties on the quantum group side
translate to properties of the functions. Our construction of the maps $\sF^{(x_0)}$
necessarily involved certain somewhat arbitrary choices, and occasionally some other conventions could
be considered more natural. In particular, it may be desirable to have the order of the tensorands
correspond to the order of the variables $x_1 < \cdots < x_n$
of the functions on the real line (the boundary
of the upper half-plane), rather than to have the tensor product constructed in the
reverse order. This appendix gives two alternatives to our conventions, which can be used
to achieve the more intuitive order of tensor products. The first straightforward
option is to modify the coproduct, which defines the tensor product representations
for the quantum group. The second option is to modify the choice of basis functions and
the restricted chamber in which the basis functions are defined. The first option
alters the Hopf algebra structure of $\Uqsltwo$, and we will not explicitly list all the
necessary changes in the representation theoretic lemmas. For the second choice, the Hopf
algebra structure remains the same, so we can concisely state the corresponding versions of
our main results for that case.

\subsection{Opposite coproduct}

The coproduct in Section~\ref{subsub: definition of the quantum group} 
was chosen according to what appears more commonly in the quantum group literature.
The algebra $\Uqsltwo$ could alternatively be equipped with a unique 
Hopf algebra structure corresponding to the coproduct
\begin{align*}
\Hcp^\mathrm{op} \;\colon\; & \Uqsltwo\rightarrow\Uqsltwo\tens\Uqsltwo,
\end{align*}
given on the generators by the expressions
\begin{align*}\label{eq: coproduct}
\Hcp^\mathrm{op}(E) = K\tens E + E\tens 1, \qquad \Hcp^\mathrm{op}(K) =  K\tens K, \qquad
\Hcp^\mathrm{op}(F) = 1\tens F + F \tens K^{-1} ,
\end{align*}
i.e., with the order opposite to that of Section~\ref{subsub: definition of the quantum group}.
It is then obvious that for the correspondence defined by
\begin{align*}
\sF^{(x_0)}_\mathrm{op} \colon \Wd_{d_1} \tens \cdots \tens \Wd_{d_n}
    \to \; & \mathcal{C}^\infty \big( \chamber^{(x_0)}_n \big) \\
\sF^{(x_0)}_\mathrm{op} [\Wbas_{l_1} \tens \cdots \tens \Wbas_{l_n}] = \; & \FWint^{(x_0)}_{l_1 , \ldots, l_n}
\end{align*}
analogues of our main results still hold.

\subsection{Different basis functions}
\noindent 
\begin{figure}
\includegraphics[width=1\textwidth]{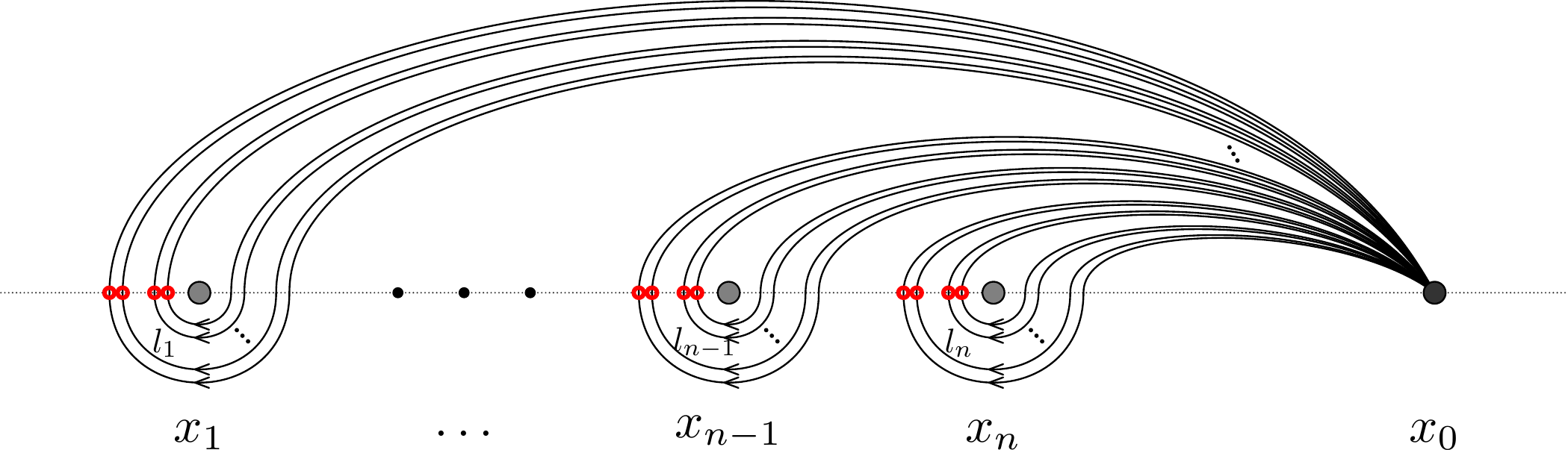}
\caption{\label{fig: alternative basis functions}
The integration surface $\SurfFWalt_{l_1 , \ldots , l_n}$.
The point where the integrand 
$\fFWalt_{l_1 , \ldots , l_n}(\boldsymbol{x};\cdot)$
is rephased to be positive is marked by red circles.}
\end{figure}

We finally present a way of keeping the standard coproduct while
achieving the more natural order of tensor products.
The Hopf algebra structure of $\Uqsltwo$ will thus be the same
as elsewhere in the article, and all representation theoretic formulas of
Section~\ref{sec: q stuff} hold without any changes. What we change now is the
basis functions of the spin chain - Coulomb gas correspondence.

We take the anchor point $x_0$ for the new basis functions to lie to the right 
of all
other variables, so the basis functions are defined on the chamber restricted 
from the right
\begin{align*}
\chamberalt^{(x_0)}_n := \set{ (x_1 , \ldots , x_n) \in \bR^n \; \Big| \; x_1 < \cdots < x_n < x_0 } .
\end{align*}
Fix again dimension parameters $d_1 , \ldots, d_n \in \bZpos$.
Analogously to the definition of basis functions
in Section~\ref{sub: basis functions as FW integrals}, define now,
for $l_1 , \ldots , l_n \in \bZnn$, the function
\begin{align*}
\FWintalt^{(x_0)}_{l_1 , \ldots , l_n} \colon \chamberalt^{(x_0)}_n \to \bC
\end{align*}
as an integral over the surface $\SurfFWalt_{l_1 , \ldots , l_n}$
depicted in Figure~\ref{fig: alternative basis functions}, with the integrand
$\fFWalt_{l_1 , \ldots , l_n}$ rephased to be positive at the point marked in the figure,
\begin{align*}
\FWintalt^{(x_0)}_{l_1 , \ldots , l_n} ( x_1 , \ldots , x_n )
    = \int_{\SurfFWalt_{l_1 , \ldots , l_n}} \fFWalt_{l_1 , \ldots , l_n}(x_1 , \ldots, x_n ; w_1 , \ldots, w_{\ell}) \; \ud w_1 \cdots \ud w_{\ell}  ,
\end{align*}
where $\ell = \sum_{j=1}^n l_j$ as before.

Then form the tensor product representation
\begin{align*}
\Wd_{d_1} \tens \cdots \tens \Wd_{d_n}
\end{align*}
in the order that is opposite to our convention~\eqref{eq: order of tensorands}
used elsewhere. We define
\begin{align*}
\sFalt^{(x_0)}_{d_1 , \ldots, d_n} \colon \Wd_{d_1} \tens \cdots \tens \Wd_{d_n}
    \to \sC^{\infty}\big(\chamberalt_{n}^{(x_{0})}\big)
\end{align*}
by setting
\begin{align*}
\sFalt^{(x_0)}_{d_1 , \ldots, d_n} \big[ \Wbas^{(d_1)}_{l_1} \tens \cdots \tens \Wbas^{(d_n)}_{l_n} \big]
    = \FWintalt^{(x_0)}_{l_1 , \ldots , l_n} ,
\end{align*}
and extending linearly.
Below we again omit the subscript dimensions from the notation.
This alternative correspondence $\sFalt$ has properties entirely similar to $\sF$.
In particular, for $v \in \Wd_{d_1} \tens \cdots \tens \Wd_{d_n}$ we have:
\begin{itemize}
\item (well-def.): If $E.v=0$, then $\sFalt^{(x_0)} [v] (x_1 , \ldots, x_n)$ is independent of $x_0$,
and thus defines a function \[ \sFalt[v] \colon \chamber_n \to \bC . \]
\item (COV): For any $\xi\in\bR$ we have the translation invariance
\begin{align*}
\sFalt^{(x_{0}+\xi)}[v](x_{1}+\xi,\ldots,x_{n}+\xi)
=\; & \sFalt^{(x_{0})}[v](x_{1},\ldots,x_{n}).
\end{align*}
If furthermore $K.v=q^{d-1}\,v$, then for any $\lambda>0$ we have
the scaling covariance
\begin{align*}
\sFalt^{(\lambda x_{0})}[v](\lambda x_{1},\ldots,\lambda x_{n})
=\; & \lambda^{\Delta_{d}^{d_{1},\ldots,d_{n}}}
\times\sFalt^{(x_{0})}[v](x_{1},\ldots,x_{n}).
\end{align*}
If $K.v=v$ and $E.v=0$, then we have the full M\"obius covariance
\begin{align*}
\prod_{j=1}^{n}\Mob'(x_{j})^{h_{1,d_{j}}}
\times\sFalt[v]\left(\Mob(x_{1}),\ldots,\,\Mob(x_{n})\right)
=\; & \sFalt[v](x_{1},\ldots,\, x_{n})
\end{align*}
for any M\"obius transformation $\Mob\colon\bH\to\bH$ such that $\Mob(x_{1})<\Mob(x_{2})<\cdots<\Mob(x_{n})$.

\item (PDE): If $E.v=0$, we have
\begin{align*}
\sD_{d_{j}}^{(j)}\sFalt[v]=\; & 0\qquad\text{for }j=1,\ldots,n ,
\end{align*}
where $\sD_{d_{j}}^{(j)}$ is the differential 
operator~\eqref{eq: BSA differential operator}.

\item (ASY): 
Suppose that $v$ belongs to the subrepresentation obtained by picking the
$d$-dimensional irreducible direct summand in the tensor product of the $j$:th
and $j+1$:st factors $\Wd_{d_{j}}$ and $\Wd_{d_{j+1}}$
(now counting from the left), and denote by
\begin{align*}
\hat{v} \in \Big( \Wd_{d_1} \tens \cdots \tens \Wd_{d_{j-1}} \Big) \tens \Wd_{d}
    \tens \Big(\Wd_{d_{j+2}} \tens \cdots \tens \Wd_{d_n} \Big) 
\end{align*}
the vector obtained by identifying $v$ as a vector in an $(n-1)$-fold
tensor product representation.
More precisely, with the
earlier notations, this means $v=\pi_{N-j,N-j+1}^{(d)}(v)$ and
$\hat{v} = \hat{\pi}_{N-j,N-j+1}^{(d)}(v)$.
Then we have 
\begin{align*}
\lim_{x_{j},x_{j+1}\to\xi}
\Big( & (x_{j+1}-x_{j})^{-\Delta_{d}^{d_{j},d_{j+1}}}
\times\sFalt^{(x_{0})}[v](x_{1},\ldots,x_{n}) \Big) \\
=\; & B_{d}^{d_{j},d_{j+1}} \times 
\sFalt^{(x_{0})}[\hat{v}](x_{1},\ldots,x_{j-1},\xi,x_{j+2},\ldots,x_{n}).
\end{align*}
\end{itemize}
Also the further results of Section~\ref{sec: further properties}
have straightforward analogues for $\sFalt$.

\bibliographystyle{annotate}

\newcommand{\etalchar}[1]{$^{#1}$}
\def\cprime{$'$}

\end{document}